\documentclass[11pt]{amsart}

\usepackage{amsfonts}
\usepackage[foot]{amsaddr}
\usepackage{mathrsfs}
\usepackage{mathtools}
\usepackage{amssymb}
\usepackage{latexsym,amsthm}
\usepackage[all]{xy}
\usepackage{graphicx}
\usepackage{wrapfig}
\usepackage{array,multirow}
\usepackage{bm}
\usepackage{enumitem,caption}
\usepackage{mathabx} 

\usepackage{longtable}

\usepackage{wasysym}

\usepackage{ifthen}

\usepackage{hyperref}
\hypersetup{%
    pdfborder = {0 0 0},
    colorlinks,
    citecolor=red!50!black,
    filecolor=Darkgreen,
    linkcolor=blue!40!black,
    urlcolor=cyan!50!black!90
}

\usepackage{tikz,tikz-cd}
\usetikzlibrary{matrix,arrows,shapes}

\usetikzlibrary{calc,decorations.pathmorphing,decorations.markings,
  decorations.pathreplacing,patterns,shapes}

\tikzset{boper/.style={rectangle,fill,inner sep=2pt,black}}
\tikzset{bblob/.style={circle,fill,inner sep=1.5pt,black}}
\tikzset{blob/.style={circle,draw,fill=white,outer sep=0mm,inner sep=0.3mm}}

\tikzset{mblob/.style={circle,fill=white,draw=black,inner sep=1.5pt}}
\tikzset{diam/.style={diamond,fill,inner sep=1.5pt,black}}
\tikzset{pblob/.style={circle,fill=white,draw=red,inner sep=1.5pt}}

\tikzset{
  aline/.style={
    decorate,
    decoration={
      meta-amplitude=#1,
      meta-segment length=0.15cm,
},
    postaction={decorate,ultra thick,decoration={markings,mark = at position #1 with {\arrow{>}}}}        
  },
  aline/.default=0.9
}

\tikzset{
  bline/.style={blue,
    decorate,
    decoration={
      meta-amplitude=#1,
      meta-segment length=0.3cm,
},
    postaction={decorate,ultra thick,decoration={markings,mark = at position #1 with {\arrow{>}}}}        
  },
  bline/.default=0.9
}

\tikzset{edge/.style={very thick,draw=blue}}%
\tikzset{contour/.style={brown,dashed,postaction={decorate,decoration={markings,mark = at position #1 with {\arrow{>}}}}}}
\tikzset{contour/.default=0.5}

\def\loos{0.35}
\def\slt{0.2}
\pgfmathsetmacro{\ae}{atan(\slt)}
\pgfmathsetmacro{\aw}{\ae+180}
\pgfmathsetmacro{\an}{90-\ae}
\pgfmathsetmacro{\as}{\an+180}
\pgfmathsetmacro{\sltb}{sqrt(1-\slt*\slt)}
\pgfmathsetmacro{\lcrot}{45-atan(\slt/\sltb)*0.5}
\tikzset{distort/.style={cm={1,0,-\slt,\sltb,(0,0)}}}
\def\goI#1(#2,#3){
\pgfextra{
\pgfmathparse{#1+180}\global\let\oldangle=\currentangle\global\let\newangle=\pgfmathresult\global\let\currentangle=#1
\pgfmathparse{\oldx+#2}\global\let\newx=\pgfmathresult\xdef\oldx{#2}
\pgfmathparse{\oldy+#3}\global\let\newy=\pgfmathresult\xdef\oldy{#3}
}
.. controls ++(\oldangle:\loos) and ++(\newangle:\loos) .. ++(\newx,\newy)
}
\def\go#1{\expandafter\goI#1}
\def\startI#1(#2,#3){
\pgfextra{\global\let\currentangle=#1\xdef\oldx{#2}\xdef\oldy{#3}}}
\def\start#1{\expandafter\startI#1}
\def\north{\an(0,0.5)}
\def\south{\as(0,-0.5)}
\def\east{\ae(0.5,0)}
\def\west{\aw(-0.5,0)}

\tikzset{bgplaq/.style={fill=lightgray!20!white}}
\tikzset{dgreen/.style={green!50!black,thick}}

\def\plaq(#1,#2){
\begin{scope}[shift={(#1,#2)}]
\draw[dotted] (-0.5,-0.5) rectangle ++(1,1); 
\end{scope}
}

\def\plaqz(#1,#2){
\begin{scope}[shift={(#1,#2)}]
\draw[bgplaq,dotted] (-0.5,-0.5) rectangle ++(1,1); 
\end{scope}
}

\def\plaqa(#1,#2){
\begin{scope}[shift={(#1,#2)}]
\draw[dotted,bgplaq] (-0.5,-0.5) rectangle ++(1,1);
\draw[edge] (0,-0.5) \start\north\go\east;
\draw[edge] (0,0.5) \start\south\go\west;
\end{scope}
}
\def\plaqb(#1,#2){
\begin{scope}[shift={(#1,#2)}]
\draw[dotted,bgplaq] (-0.5,-0.5) rectangle ++(1,1);
\draw[edge] (0,0.5) \start\south\go\east;
\draw[edge] (0,-0.5) \start\north\go\west;
\end{scope}
}

\def\bplaqe(#1,#2){
\begin{scope}[shift={(#1,#2)}]
\draw[dotted] (-0.5,0.5) -- (0.5,0.5) -- (-0.5,-0.5) -- cycle;
\end{scope}
}
\def\bplaqw(#1,#2){
\begin{scope}[shift={(#1,#2)}]
\draw[dotted] (0.5,-0.5) -- (0.5,0.5) -- (-0.5,-0.5) -- cycle;
\end{scope}
}
\def\bplaqn(#1,#2){
\begin{scope}[shift={(#1,#2)}]
\draw[dotted] (0.5,-0.5) -- (-0.5,0.5) -- (-0.5,-0.5) -- cycle;
\end{scope}
}
\def\bplaqs(#1,#2){
\begin{scope}[shift={(#1,#2)}]
\draw[dotted] (0.5,-0.5) -- (-0.5,0.5) -- (0.5,0.5) -- cycle;
\end{scope}
}
\def\bplaqez(#1,#2){
\begin{scope}[shift={(#1,#2)}]
\draw[dotted,bgplaq] (-0.5,0.5) -- (0.5,0.5) -- (-0.5,-0.5) -- cycle;
\end{scope}
}
\def\bplaqwz(#1,#2){
\begin{scope}[shift={(#1,#2)}]
\draw[dotted,bgplaq] (0.5,-0.5) -- (0.5,0.5) -- (-0.5,-0.5) -- cycle;
\end{scope}
}
\def\bplaqnz(#1,#2){
\begin{scope}[shift={(#1,#2)}]
\draw[dotted,bgplaq] (0.5,-0.5) -- (-0.5,0.5) -- (-0.5,-0.5) -- cycle;
\end{scope}
}
\def\bplaqsz(#1,#2){
\begin{scope}[shift={(#1,#2)}]
\draw[dotted,bgplaq] (0.5,-0.5) -- (-0.5,0.5) -- (0.5,0.5) -- cycle;
\end{scope}
}
\def\bplaqea(#1,#2){
\begin{scope}[shift={(#1,#2)}]
\draw[dotted,bgplaq] (-0.5,0.5) -- (0.5,0.5) -- (-0.5,-0.5) -- cycle;
\draw[edge] (0,0.5) \start\south\go\west;
\end{scope}
}
\def\bplaqeb(#1,#2){
\begin{scope}[shift={(#1,#2)}]
\draw[dotted,bgplaq] (-0.5,0.5) -- (0.5,0.5) -- (-0.5,-0.5) -- cycle;
\draw[edge] (0,0.5) \start\south\go\west
node[blob] {};
\end{scope}
}
\def\bplaqwa(#1,#2){
\begin{scope}[shift={(#1,#2)}]
\draw[dotted,bgplaq] (0.5,-0.5) -- (0.5,0.5) -- (-0.5,-0.5) -- cycle;
\draw[edge] (0,-0.5) \start\north\go\east;
\end{scope}
}
\def\bplaqwb(#1,#2){
\begin{scope}[shift={(#1,#2)}]
\draw[dotted,bgplaq] (0.5,-0.5) -- (0.5,0.5) -- (-0.5,-0.5) -- cycle;
\draw[edge] (0,-0.5) \start\north\go\east
node[blob] {};
\end{scope}
}
\def\bplaqna(#1,#2){
\begin{scope}[shift={(#1,#2)}]
\draw[dotted,bgplaq] (0.5,-0.5) -- (-0.5,0.5) -- (-0.5,-0.5) -- cycle;
\draw[edge] (0,-0.5) \start\north\go\west;
\end{scope}
}
\def\bplaqnb(#1,#2){
\begin{scope}[shift={(#1,#2)}]
\draw[dotted,bgplaq] (0.5,-0.5) -- (-0.5,0.5) -- (-0.5,-0.5) -- cycle;
\draw[edge] (0,-0.5) \start\north\go\west 
node[blob] {};
\end{scope}
}
\def\bplaqsa(#1,#2){
\begin{scope}[shift={(#1,#2)}]
\draw[dotted,bgplaq] (0.5,-0.5) -- (-0.5,0.5) -- (0.5,0.5) -- cycle;
\draw[edge] (0,0.5) \start\south\go\east;
\end{scope}
}
\def\bplaqsb(#1,#2){
\begin{scope}[shift={(#1,#2)}]
\draw[dotted,bgplaq] (0.5,-0.5) -- (-0.5,0.5) -- (0.5,0.5) -- cycle;
\draw[edge] (0,0.5) \start\south\go\east
node[blob] {}; 
\end{scope}
}

\def\plaqff(#1,#2){
\begin{scope}[shift={(#1,#2)}]
\draw[dotted,bgplaq] (-0.5,-0.5) rectangle ++(1,1);
\draw[edge] (0,-0.5) \start\north\go\east;
\end{scope}
}
\def\plaqf(#1,#2){
\begin{scope}[shift={(#1,#2)}]
\draw[dotted,bgplaq] (-0.5,-0.5) rectangle ++(1,1);
\draw[edge] (0,0.5) \start\south\go\west;
\end{scope}
}
\def\plaqd(#1,#2){
\begin{scope}[shift={(#1,#2)}]
\draw[dotted,bgplaq] (-0.5,-0.5) rectangle ++(1,1);
\draw[edge] (0,-0.5) \start\north\go\west;
\end{scope}
}
\def\plaqdd(#1,#2){
\begin{scope}[shift={(#1,#2)}]
\draw[dotted,bgplaq] (-0.5,-0.5) rectangle ++(1,1);
\draw[edge] (0,0.5) \start\south\go\east;
\end{scope}
}
\def\plaqc(#1,#2){
\begin{scope}[shift={(#1,#2)}]
\draw[dotted,bgplaq] (-0.5,-0.5) rectangle ++(1,1);
\draw[edge] (0,-0.5) \start\north\go\north;
\end{scope}
}
\def\plaqcc(#1,#2){
\begin{scope}[shift={(#1,#2)}]
\draw[dotted,bgplaq] (-0.5,-0.5) rectangle ++(1,1);
\draw[edge] (-0.5,0) \start\east\go\east;
\end{scope}
}
\def\plaqg(#1,#2){
\begin{scope}[shift={(#1,#2)}]
\draw[dotted,bgplaq] (-0.5,-0.5) rectangle ++(1,1);
\end{scope}
}

\def\cross(#1,#2){
\begin{scope}[shift={(#1,#2)}]
\draw[thick,blue] (-0.1,-0.1) -- (0.1,0.1); 
\draw[thick,blue] (-0.1,0.1) -- (0.1,-0.1); 
\end{scope}
}

\def\gcross(#1,#2){
\begin{scope}[shift={(#1,#2)}]
\draw[thick,green!50!black] (-0.1,-0.1) -- (0.1,0.1); 
\draw[thick,green!50!black] (-0.1,0.1) -- (0.1,-0.1); 
\end{scope}
}

\def\scross(#1,#2){
\begin{scope}[shift={(#1,#2)}]
\draw[thick,blue] (-0.07,-0.07) -- (0.07,0.07); 
\draw[thick,blue] (-0.07,0.07) -- (0.07,-0.07); 
\end{scope}
}

\def\triplaq(#1,#2,#3,#4,#5){
\begin{scope}[shift={(#1,#2)}]
\draw[dgreen] (1.6,-1) -- (0,0) -- (1.6,1);
\draw(0,0) node[bblob] {};
\draw(1.6,1) node[bblob] {};
\draw(1.6,-1) node[bblob] {};
\draw(0,0) node[left] {${#3}$};
\draw(1.6,1) node[right] {${#4}$};
\draw(1.6,-1) node[right] {${#5}$};
\draw[aline=0.9] (0.8,0.5) -- ( (0.8,-0.5);
\cross(0.8,0);
\draw[blue,wavy] (1.6,0) -- (0.8,0);
\end{scope}
}

\def\recpplaq(#1,#2,#3,#4,#5,#6){
\begin{scope}[shift={(#1,#2)}]
\draw[dgreen] (-0.5,-0.5) -- (0.5,-0.5);
\draw[dgreen] (-0.5,0.5) -- (0.5,0.5);
\draw (-0.5,-0.5) node[left] {${#6}$};
\draw (0.5,-0.5)  node[right] {${#5}$};
\draw (-0.5,0.5) node[left] {${#3}$};
\draw  (0.5,0.5) node[right] {${#4}$};
\draw (-0.5,-0.5) node[bblob] {};
\draw (0.5,-0.5)  node[bblob] {};
\draw (-0.5,0.5) node[bblob] {};
\draw  (0.5,0.5) node[bblob] {};
\draw[aline=0.9] (0,0.5) -- ( (0,-0.5);
\draw[blue,wavy=0.2] (-0.5,0) -- (0.5,0);
\end{scope}
}

\def\recmplaq(#1,#2,#3,#4,#5,#6){
\begin{scope}[shift={(#1,#2)}]
\draw[dgreen] (-0.5,-0.5) -- (0.5,-0.5);
\draw[dgreen] (-0.5,0.5) -- (0.5,0.5);
\draw (-0.5,-0.5) node[left] {${#6}$};
\draw (0.5,-0.5)  node[right] {${#5}$};
\draw (-0.5,0.5) node[left] {${#3}$};
\draw  (0.5,0.5) node[right] {${#4}$};
\draw (-0.5,-0.5) node[bblob] {};
\draw (0.5,-0.5)  node[bblob] {};
\draw (-0.5,0.5) node[bblob] {};
\draw  (0.5,0.5) node[bblob] {};
\draw[aline=0.9] (0,0.5) -- ( (0,-0.5);
\draw[blue,wavy=0.2] (0.5,0) -- (-0.5,0);
\end{scope}
}

\def\recpdashplaq(#1,#2,#3,#4,#5,#6){
\begin{scope}[shift={(#1,#2)}]
\draw[dgreen,dashed] (-0.5,-0.5) -- (0.5,-0.5);
\draw[dgreen] (-0.5,0.5) -- (0.5,0.5);
\draw (-0.5,-0.5) node[left] {${#6}$};
\draw (0.5,-0.5)  node[right] {${#5}$};
\draw (-0.5,0.5) node[left] {${#3}$};
\draw  (0.5,0.5) node[right] {${#4}$};
\draw (-0.5,-0.5) node[bblob] {};
\draw (0.5,-0.5)  node[bblob] {};
\draw (-0.5,0.5) node[bblob] {};
\draw  (0.5,0.5) node[bblob] {};
\draw[aline=0.9] (0,0.5) -- ( (0,-0.5);
\draw[blue,wavy=0.2] (-0.5,0) -- (0.5,0);
\end{scope}
}

\def\recmdashplaq(#1,#2,#3,#4,#5,#6){
\begin{scope}[shift={(#1,#2)}]
\draw[dgreen,dashed] (-0.5,-0.5) -- (0.5,-0.5);
\draw[dgreen] (-0.5,0.5) -- (0.5,0.5);
\draw (-0.5,-0.5) node[left] {${#6}$};
\draw (0.5,-0.5)  node[right] {${#5}$};
\draw (-0.5,0.5) node[left] {${#3}$};
\draw  (0.5,0.5) node[right] {${#4}$};
\draw (-0.5,-0.5) node[bblob] {};
\draw (0.5,-0.5)  node[bblob] {};
\draw (-0.5,0.5) node[bblob] {};
\draw  (0.5,0.5) node[bblob] {};
\draw[aline=0.9] (0,0.5) -- ( (0,-0.5);
\draw[blue,wavy=0.2] (0.5,0) -- (-0.5,0);
\end{scope}
}

\def\sosplaq(#1,#2){
\begin{scope}[shift={(#1,#2)}]
\draw[dgreen] (-0.5,-0.5) rectangle ++(1,1);
\draw(0.5,0.5) node[bblob] {};
\draw(0.5,-0.5) node[bblob] {};
\draw(-0.5,0.5) node[bblob] {};
\draw(-0.5,-0.5) node[bblob] {};
\draw[dgreen] (-0.5,0.2) -- (-0.3,0.5);
\end{scope}
}

\def\udashsosplaq(#1,#2){
\begin{scope}[shift={(#1,#2)}]
\draw[dgreen,dashed] (-0.5,0.5) -- (0.5,0.5) -- (0.5,-0.5);
\draw[dgreen] (-0.5,0.5)-- (-0.5,-0.5) -- (0.5,-0.5);
\draw(0.5,0.5) node[bblob] {};
\draw(0.5,-0.5) node[bblob] {};
\draw(-0.5,0.5) node[bblob] {};
\draw(-0.5,-0.5) node[bblob] {};
\draw[dgreen] (-0.5,0.2) -- (-0.3,0.5);
\end{scope}
}

\def\lsosplaq(#1,#2){
\begin{scope}[shift={(#1,#2)}]
\draw[dgreen] (-0.5,-0.5) rectangle (0.5,0.2) ++(1,1);
\draw(0.5,0.2) node[bblob] {};
\draw(0.5,-0.5) node[bblob] {};
\draw(-0.5,0.2) node[bblob] {};
\draw(-0.5,-0.5) node[bblob] {};
\draw[dgreen] (-0.5,-0.1) -- (-0.3,0.2);
\end{scope}
}

\def\usosplaq(#1,#2){
\begin{scope}[shift={(#1,#2)}]
\draw[dgreen] (-0.5,-0.2) rectangle (0.5,0.5) ++(1,1);
\draw(0.5,0.5) node[bblob] {};
\draw(0.5,-0.2) node[bblob] {};
\draw(-0.5,0.5) node[bblob] {};
\draw(-0.5,-0.2) node[bblob] {};
\draw[dgreen] (-0.5,0.2) -- (-0.3,0.5);
\end{scope}
}
\def\lslantsosplaq(#1,#2){
\begin{scope}[shift={(#1,#2)}]
\draw[dgreen] (-0.5,-0.5) -- (0.5,-0.5) -- (0.5,0.2) -- (-0.5,0.5) -- (-0.5,-0.5) ++(1,1);
\draw(0.5,0.2) node[bblob] {};
\draw(0.5,-0.5) node[bblob] {};
\draw(-0.5,0.5) node[bblob] {};
\draw(-0.5,-0.5) node[bblob] {};
\draw[dgreen] (-0.5,0.2) -- (-0.3,0.4);
\end{scope}
}

\def\uslantsosplaq(#1,#2){
\begin{scope}[shift={(#1,#2)}]
\draw[dgreen] (-0.5,-0.5) -- (0.5,-0.2) -- (0.5,0.5) -- (-0.5,0.5) -- (-0.5,-0.5) ++(1,1);
\draw(0.5,0.5) node[bblob] {};
\draw(0.5,-0.2) node[bblob] {};
\draw(-0.5,0.5) node[bblob] {};
\draw(-0.5,-0.5) node[bblob] {};
\draw[dgreen] (-0.5,0.2) -- (-0.3,0.5);
\end{scope}
}

\def\luslantsosplaq(#1,#2){
\begin{scope}[shift={(#1,#2)}]
\draw[dgreen] (-0.5,-0.2) -- (0.5,-0.5) -- (0.5,0.5) -- (-0.5,0.5) -- (-0.5,-0.2) ++(1,1);
\draw(0.5,0.5) node[bblob] {};
\draw(0.5,-0.5) node[bblob] {};
\draw(-0.5,0.5) node[bblob] {};
\draw(-0.5,-0.2) node[bblob] {};
\draw[dgreen] (-0.5,0.2) -- (-0.3,0.5);
\end{scope}
}

\def\rcasosplaq(#1,#2){
\begin{scope}[shift={(#1,#2)}]
\draw[dgreen] (-0.75,-0.5) -- (0.5,-0.5) -- (0.5,0.5) -- (-0.5,0.5) -- (-0.75,-0.5) ++(1,1);
\draw(0.5,0.5) node[bblob] {};
\draw(0.5,-0.5) node[bblob] {};
\draw(-0.5,0.5) node[bblob] {};
\draw(-0.75,-0.5) node[bblob] {};
\draw[dgreen] (-0.6,0.15) -- (-0.3,0.5);
\end{scope}
}

\def\rcbsosplaq(#1,#2){
\begin{scope}[shift={(#1,#2)}]
\draw[dgreen] (-0.5,-0.5) -- (0.75,-0.5) -- (0.5,0.5) -- (-0.5,0.5) -- (-0.5,-0.5) ++(1,1);
\draw(0.5,0.5) node[bblob] {};
\draw(0.75,-0.5) node[bblob] {};
\draw(-0.5,0.5) node[bblob] {};
\draw(-0.5,-0.5) node[bblob] {};
\draw[dgreen] (-0.5,0.2) -- (-0.3,0.5);
\end{scope}
}

\def\rwidelsosplaq(#1,#2){
\begin{scope}[shift={(#1,#2)}]
\draw[dgreen] (-0.5,-0.2) rectangle (1,0.5) ++(1,1);
\draw(1,0.5) node[bblob] {};
\draw(1,-0.2) node[bblob] {};
\draw(-0.5,0.5) node[bblob] {};
\draw(-0.5,-0.2) node[bblob] {};
\draw[dgreen] (-0.5,0.2) -- (-0.3,0.5);
\end{scope}
}

\def\rwidesosplaq(#1,#2){
\begin{scope}[shift={(#1,#2)}]
\draw[dgreen] (-0.5,-0.5) rectangle (1,0.5) ++(1,1);
\draw(1,0.5) node[bblob] {};
\draw(1,-0.5) node[bblob] {};
\draw(-0.5,0.5) node[bblob] {};
\draw(-0.5,-0.5) node[bblob] {};
\draw[dgreen] (-0.5,0.2) -- (-0.3,0.5);
\end{scope}
}

\def\tallsosplaq(#1,#2){
\begin{scope}[shift={(#1,#2)}]
\draw[dgreen] (-0.5,-0.5) rectangle (0.5,0.8) ++(1,1);
\draw(0.5,0.8) node[bblob] {};
\draw(0.5,-0.5) node[bblob] {};
\draw(-0.5,0.8) node[bblob] {};
\draw(-0.5,-0.5) node[bblob] {};
\draw[dgreen] (-0.5,0.5) -- (-0.3,0.8);
\end{scope}
}

\pgfdeclaremetadecoration{prewavy}{zigzag}{
  \state{zigzag}[width=\pgfmetadecorationsegmentamplitude*\pgfmetadecoratedpathlength - 0.5*\pgfmetadecorationsegmentlength ,
                  next state=straight] {
        \decoration{zigzag}
      }
    \state{straight}[width=\pgfmetadecorationsegmentlength,
                   next state=final] {
        \decoration{curveto}
    }
    \state{final}{
        \decoration{zigzag}
        \beforedecoration{\pgfpathmoveto{\pgfpointmetadecoratedpathfirst}}
    }
}

\tikzset{
  wavy/.style={
    decorate,
    decoration={
      prewavy,
      meta-amplitude=#1,
      meta-segment length=0.3cm,
      amplitude=1.5pt, 
      segment length=6pt 
},
    postaction={decorate,ultra thick,decoration={markings,mark = at position #1 with {\arrow{>}}}}        
  },  
  wavy/.default=0.5
}

\tikzset{oper/.style={rectangle,fill,inner sep=2.5pt}}
\tikzset{arr/.style={postaction={decorate,thick,decoration={markings,mark = at position #1 with {\arrow{>}}}}}}

\tikzset{
  abline/.style={blue,dashed,
    decorate,
    decoration={
      meta-amplitude=#1,
      meta-segment length=0.3cm,
},
    postaction={decorate,ultra thick,decoration={markings,mark = at position #1 with {\arrow{>}}}}        
  },
  abline/.default=0.5
}

\tikzset{
  aline/.style={
    decorate,
    decoration={
      meta-amplitude=#1,
      meta-segment length=0.3cm,
},
    postaction={decorate,ultra thick,decoration={markings,mark = at position #1 with {\arrow{>}}}}        
  },
  aline/.default=0.5
}

\tikzset{
  agline/.style={green!50!black,thick,
    decorate,
    
    decoration={
      meta-amplitude=#1,
      meta-segment length=0.3cm,
},
    postaction={decorate,ultra thick,decoration={markings,mark = at position #1 with {\arrow{>}}}}        
  },
  agline/.default=0.5
}

\tikzset{
  gline/.style={color=green!50!black,thick}}

  \tikzset{
  blueline/.style={blue!50,thick,
    decorate,
    decoration={
      meta-amplitude=#1,
      meta-segment length=0.3cm,
},
    postaction={decorate,ultra thick,decoration={markings,mark = at position #1 with {\arrow{>}}}}        
  },
  aline/.default=0.5
}

  \tikzset{
  redline/.style={red!50,thick,
    decorate,
    decoration={
      meta-amplitude=#1,
      meta-segment length=0.3cm,
},
    postaction={decorate,ultra thick,decoration={markings,mark = at position #1 with {\arrow{>}}}}        
  },
  aline/.default=0.5
}

 \tikzset{
  dashblueline/.style={blue!50,thick,dashed,
    decorate,
    decoration={
      meta-amplitude=#1,
      meta-segment length=0.3cm,
},
    postaction={decorate,ultra thick,decoration={markings,mark = at position #1 with {\arrow{>}}}}        
  },
  aline/.default=0.5
}

  \tikzset{
  dashredline/.style={red!50,thick,dashed,
    decorate,
    decoration={
      meta-amplitude=#1,
      meta-segment length=0.3cm,
},
    postaction={decorate,ultra thick,decoration={markings,mark = at position #1 with {\arrow{>}}}}        
  },
  aline/.default=0.5
}

\tikzset{
  dashline/.style={dashed,thick,
    decorate,
    decoration={
      meta-amplitude=#1,
      meta-segment length=0.3cm,
},
    postaction={decorate,ultra thick,decoration={markings,mark = at position #1 with {\arrow{>}}}}        
  },
  aline/.default=0.5
}

 \tikzset{
  greenline/.style={green!50!black,thick,
    decorate,
    decoration={
      meta-amplitude=#1,
      meta-segment length=0.3cm,
},
    postaction={decorate,ultra thick,decoration={markings,mark = at position #1 with {\arrow{>}}}}        
  },
  aline/.default=0.5
}

 \tikzset{
  dashgreenline/.style={green!50!black,thick,dashed,         decorate,
    decoration={
      meta-amplitude=#1,
      meta-segment length=0.3cm,
},
    postaction={decorate,ultra thick,decoration={markings,mark = at position #1 with {\arrow{>}}}}        
  },
  aline/.default=0.5
}

\newcommand{\nc}{\newcommand}
\newcommand{\rnc}{\renewcommand}

\setlist[itemize,1]{leftmargin=.4in}
\setlist[enumerate,1]{leftmargin=.4in,label=(\roman*)}
\setlist[description,1]{leftmargin=.4in,font=\normalfont\itshape}

\allowdisplaybreaks[1]

\nc{\qq}{\qquad}
\nc{\qu}{\quad}

\newtheorem{thrm}{Theorem}[section]
\newtheorem{prop}[thrm]{Proposition}
\newtheorem{crl}[thrm]{Corollary}
\newtheorem{lemma}[thrm]{Lemma}

\theoremstyle{definition}
\newtheorem{defn}[thrm]{Definition}

\theoremstyle{remark}
\newtheorem{rmk}[thrm]{Remark}

\nc{\rmkend}{\ensuremath{\diameter}}
\nc{\examend}{\ensuremath{\diameter}}
\nc{\defnend}{\ensuremath{\diameter}}

\numberwithin{equation}{section}

\nc{\eq}[1]{\begin{equation} #1 \end{equation}}
\nc{\eqrefs}[2]{\text{(\ref{#1}-\ref{#2})}}

\rnc\appendixname{}

\nc{\al}{\alpha}
\nc{\be}{\beta}
\nc{\eps}{\epsilon}
\nc{\veps}{\varepsilon}
\nc{\ga}{\gamma}
\nc{\del}{\delta}
\nc{\Del}{\Delta}
\nc{\ze}{\zeta}
\nc{\ka}{\kappa}
\nc{\la}{\lambda}
\nc{\La}{\Lambda}
\nc{\vrho}{\varrho}
\nc{\si}{\sigma}
\nc{\Si}{\Sigma}
\nc{\ups}{\upsilon}
\nc{\vphi}{\varphi}
\nc{\om}{\omega}
\nc{\Om}{\Omega}

\nc{\A}{\mathbb{A}}
\nc{\C}{\mathbb{C}}
\nc{\F}{\mathbb{F}}
\nc{\K}{\mathbb{K}}
\nc{\N}{\mathbb{N}}
\nc{\Q}{\mathbb{Q}}
\nc{\R}{\mathbb{R}}
\nc{\Z}{\mathbb{Z}}

\nc{\mfgl}{\mathfrak{g}\mathfrak{l}}
\nc{\mfsl}{\mathfrak{s}\mathfrak{l}}
\nc{\mfb}{\mathfrak{b}}
\nc{\mfg}{\mathfrak{g}}
\nc{\mfh}{\mathfrak{h}}
\nc{\mfk}{\mathfrak{k}}
\nc{\mfn}{\mathfrak{n}}

\nc{\ot}{\otimes}

\nc{\h}{{\sf ht}}
\nc{\e}{{\sf e}}
\nc{\id}{{\sf id}}
\nc{\Id}{{\sf Id}}
\rnc{\t}{{\sf t}}

\DeclareMathOperator{\Ad}{Ad}

\DeclareMathOperator{\End}{End}
\DeclareMathOperator{\GL}{GL}

\nc{\wb}{\overline}
\nc{\wh}{\widehat}
\nc{\wt}{\widetilde}

\nc{\mc}{\mathcal}
\nc{\mf}{\mathfrak}

\nc{\red}{\color{red}}
\nc{\blu}{\color{blue}}
\nc{\br}{\color{brown}}
\nc{\grn}{\color{green!55!black}}
\nc{\gry}{\color{gray}}

\nc{\cA}{{\mathcal{A}}}
\nc{\cB}{{\mathcal{B}}}
\nc{\cD}{{\mathcal{D}}}
\nc{\cF}{{\mathcal{F}}}
\nc{\cG}{{\mathcal{G}}}
\nc{\cK}{{\mathcal{K}}}
\nc{\cM}{{\mathcal{M}}}
\nc{\cL}{{\mathcal{L}}}
\nc{\cO}{{\mathcal{O}}}
\nc{\cP}{{\mathcal{P}}}
\nc{\cQ}{{\mathcal{Q}}}
\nc{\cR}{{\mathcal{R}}}
\nc{\cT}{{\mathcal{T}}}
\nc{\cV}{{\mathcal{V}}}

\nc{\ba}{{\bar a}}
\nc{\bb}{{\bar b}}
\nc{\bD}{{\wb D}}
\nc{\bcA}{\bar \cA}
\nc{\bcK}{{\wb\cK}}
\nc{\bcL}{{\wb\cL}}
\nc{\bM}{{\wb M}}
\nc{\bcM}{{\wb\cM}}
\nc{\bp}{{\wb p}}
\nc{\bcQ}{{\wb\cQ}}
\nc{\bR}{{\wb R}}
\nc{\bcT}{{\wb\cT}}
\nc{\brho}{{\bar\vrho}}
\nc{\bpi}{\wb \pi}

\nc{\adag}{a^\dag}
\nc{\badag}{\ba^\dag}

\nc{\tc}{\wt c}
\nc{\tcK}{{\wt\cK}}
\nc{\tbcK}{{\wt\bcK}}
\nc{\tcL}{{\wt\cL}}
\nc{\tbcL}{{\wt\bcL}}
\nc{\txi}{\wt \xi}
\nc{\tK}{\wt K}

\nc{\uqbp}{U_q(\wh\mfb^+)} 
\nc{\uqbm}{U_q(\wh\mfb^-)} 
\nc{\uq}{U_q(\wh{\mathfrak{sl}}_2)}

\begin{document}

\title{A Q-operator for open spin chains II: boundary factorization}
\author{Alec Cooper$^1$}
\author{Bart Vlaar$^{1,2,3}$}
\author{Robert Weston$^1$}
\address{
$^1$Department of Mathematics, Heriot-Watt University, Edinburgh, EH14 4AS, Scotland U.K.}
\address{
$^2$Beijing Institute for Mathematical Sciences and Applications, 544 Hefangkou Village, Huaibei Town, Huairou District, Beijing, 101408, China}
\address{
$^3$Max Planck Institute for Mathematics, Vivatsgasse 7, 53111 Bonn, Germany}
\email{a.cooper@hw.ac.uk, b.vlaar@bimsa.cn, r.a.weston@hw.ac.uk}

\begin{abstract} 
One of the features of Baxter's Q-operators for many closed spin chain models is that all transfer matrices arise as products of two Q-operators with shifts in the spectral parameter.
In the representation-theoretical approach to Q-operators, underlying this is a factorization formula for L-operators (solutions of the Yang-Baxter equation associated to particular infinite-dimensional representations).
To have such a formalism to open spin chains, one needs a factorization identity for solutions of the reflection equation (boundary Yang-Baxter equation) associated to these representations. 
In the case of quantum affine $\mathfrak{sl}_2$ and diagonal K-matrices, we derive such an identity using the recently formulated theory of universal K-matrices for quantum affine algebras.
\end{abstract}


\maketitle

\setcounter{tocdepth}{1} 

\tableofcontents


\section{Introduction}

\subsection{Background and overview}

Baxter first introduced his Q-operator in \cite{Ba72,Ba73} as an auxiliary tool in the derivation of Bethe Equations for the eigenvalues of the 8-vertex model transfer matrix. 
The key characters in the story are the transfer matrix $\mc{T}(z)$ and the Q-operator $\mc Q(z)$. 
A detailed description of the essential properties of $\mc{T}(z)$ and $\mc Q(z)$ can be found in \cite{BLZ97} (also see \cite{VW20} and references therein); the key relation that they satisfy that leads directly to the Bethe equations is of the form
\eq{
\mc T(z) \mc Q(z)= \alpha_+(z) \mc Q(q z) + \alpha_-(z) \mc Q(q^{-1}z),\label{eq:TQ1} \\
}
where $\alpha_\pm(z)$ are meromorphic functions and $q\in \C^\times$ is not a root of unity.

In the original papers of Baxter, the operator $\cQ(z)$ was constructed by a brilliant but ad hoc argument; the representation-theoretic construction of $\cQ(z)$ had to wait more than 20 years until the work of Bazhanov, Lukyanov and Zamolodchikov \cite{BLZ96,BLZ97,BLZ99}. 
The main idea of the latter approach is to construct both $\mc T(z)$ and $\mc Q(z)$ as partial traces over different representations of the universal R-matrix $\cR$ of $\uq$. 
The operator $\mc T(z)$ is a twisted trace over a two-dimensional $\uq$-representation $\Pi_z$, and $\mc Q(z)$ is a similarly twisted trace over an infinite-dimensional $U_q(\wh\mfb^+)$-representation $\rho_z$, where $\wh\mfb^+$ is the upper Borel subalgebra of $\wh\mfsl_2$ (the relevant representations are defined in Section \ref{sec:reps:plus} of the current paper). 
The relation \eqref{eq:TQ1} for closed spin chains then follows immediately by considering a short exact sequence (SES) of $\uqbp$-representations with $\Pi_z\ot \rho_z $ as its `middle' object (cf. \cite[Lem.~2 (2)]{FR99}).
For an arbitrary untwisted affine Lie algebra $\wh\mfg$ with upper Borel subalgebra $\wh\mfb^+$, the level-0 representation theory of $U_q(\wh\mfb^+)$ was studied in \cite{HJ12}; for the general connection with the theory of Baxter's Q-operators see \cite{FH15}.

As well as this direct SES route to the equation, there is an alternative strategy which we refer to as the ``factorization approach''; for closed chains see \cite{BS90,De05,DKK06,De07,BJMST09,BLMS10}. In fact, this approach was the one taken by Bazhanov, Lukyanov and Zamolodchikov.
The work that developed this formalism in language most similar to the current paper, is \cite{KT14}.

In this approach, a second operator $\wb{ \mc Q}(z)$ with similar properties to $\mc Q(z)$ is introduced as a trace of $\cR$ over another infinite-dimensional representation $\brho_z $ of $U_q(\wh\mfb^+)$. 
The affinized version $\ups_z$ of the $U_q(\mfsl_2)$-Verma module is also considered as well as an another infinite-dimensional filtered $\uqbp$-module $\phi_z$; these two representations depend on a complex parameter $\mu$.
The key connection between all representations is given by Theorem \ref{thm:O:plus}, which expresses the fact that particular pairwise tensor products are isomorphic as $U_q(\wh\mfb^+)$-modules by means of an explicit intertwiner $\cO$ (defined in Section \ref{sec:O+} of the current paper). 
At the level of the L-operators this implies
\eq{ \label{factorization:bulk:intro}
\cO_{12} \cL_{\vrho}(q^\mu z)_{13} \cL_{\brho}(q^{-\mu} z)_{23} =  
\cL_{\ups}(z)_{13} \cL_{\phi}(z)_{23} \cO_{12}, 
}
(see Theorem \ref{thm:fund} of the current paper), which is referred to as \emph{factorization} of the Verma module L-operator $\cL_\ups(z)$ in terms of the L-operators $\cL_{\vrho}(z)$ and $\cL_{\brho}(z)$ which are used to define $\mc Q(z)$, $\wb{\mc Q}(z)$ (the transfer matrix corresponding to the additional operator $\cL_{\phi} (z)$ is trivial).

Defining $\mc T_{\mu}(z)$ to be the transfer matrix that is the trace over the $\mu$-dependent representation $\ups_{z}$ of $\cR$ in the first space, Theorem \ref{thm:fund} yields a relation of the following form:
\eq{ 
\mc T_{\mu}(z) \: \propto \: \mc Q(zq^{-\mu/2}) \wb{\mc Q}(zq^{\mu/2}).\label{eq:TQQ}
}
The SES associated with $\ups_{z}$ in the case $\mu$ is an integer then leads to the key relation \eqref{eq:TQ1}. 

\subsection{Present work}

In the current work we are interested in an analogue of \eqref{factorization:bulk:intro} for open chains, setting out an approach to Q-operators which complements the SES approach of \cite{VW20}. 

The problem of Q-operators for open XXZ chains with diagonal boundaries was discussed in \cite{BT18} and in \cite{Ts21}.
The XXX version of this problem was solved already in \cite{FS15}.
Earlier, Baxter TQ-relations with more general boundary conditions were found in \cite{YNZ06} (XXZ) and \cite{YZ06} (XYZ) by spin-$j$ transfer matrix asymptotics.

Our main result is the following analogue of Theorem \ref{thm:fund}, which we call the \emph{boundary factorization identity} and answers in the positive a question raised in \cite[Sec.~5]{BT18}:
\eq{ \label{factorization:boundary:intro}
\cK_\ups(z)_1 \cR_{\ups\phi}(z^2) \cK_\phi(z)_2 \,\cO = \cO \cK_\vrho(q^{\mu}z)_1 \cR_{\vrho\brho}(z^2) \cK_\brho(q^{-\mu}z)_2
} 
where $z$ is a formal parameter (which can be specialized to generic complex numbers). 
The precise statement is given in Theorem \ref{thm:keyrelation:right}. 
This formula involves the actions of the universal R-matrix of $U_q(\wh\mfsl_2)$ in tensor products of the various infinite-dimensional representations introduced.
In addition, the various K-operators are diagonal solutions of reflection equations (boundary Yang-Baxter equations) \cite{Ch84,Sk88}.
They arise as actions of the universal K-matrix associated to the augmented q-Onsager algebra, a particular coideal subalgebra of $U_q(\wh\mfsl_2)$, which featured also in e.g.~\cite{BB13,RSV15,BT18,VW20}.
More precisely, diagonal solutions of the reflection equation with a free parameter, considered by Sklyanin in his 2-boundary version of the algebraic Bethe ansatz in \cite{Sk88}, are intertwiners for this algebra.

Equation \eqref{factorization:boundary:intro} has a natural diagrammatic formulation, see Section \ref{sec:boundaryfactorization}.
In a subsequent paper the authors will explain how \eqref{factorization:boundary:intro} yields relations analogous to \eqref{eq:TQQ} and hence \eqref{eq:TQ1} for open chains.\\

The proof of \eqref{factorization:boundary:intro} and of the well-definedness of the various K-operators is an application of the universal K-matrix formalism developed in \cite{AV22a,AV22b} which is built on the earlier works \cite{BW18,BK19}.
More precisely, it relies on an extension of the theory of K-matrices for finite-dimensional representations of quantum affine algebras in \cite{AV22b} to level-0 representations of $U_q(\wh\mfb^+)$, which we discuss in Section \ref{sec:augmentedqOns}.
The key point is that, for the special case of the augmented q-Onsager algebra there exists a universal element $\cK$, centralizing this subalgebra up to a twist, simultaneously satisfying three desirable properties.
\begin{enumerate}
    \item The element $\cK$ lies in (a completion of) the Borel subalgebra $U_q(\wh\mfb^+)$, so that the resulting family of linear maps is itself compatible with $U_q(\wh \mfb^+)$-intertwiners (which play an essential role in the algebraic theory of Baxter Q-operators).
    \item The coproduct of $\cK$ is of a particularly simple form, which is relevant for the proof of the boundary factorization identity.
    \item The linear operators accomplishing the action of $\cK$ in level-0 representations satisfy the untwisted reflection equation. 
\end{enumerate}  
Thus we obtain the factorization identity \eqref{factorization:boundary:intro} as a natural consequence of the representation theory of $U_q(\wh\mfsl_2)$. 
The main benefit of this universal approach is that laborious linear-algebraic computations are avoided; in particular, we not even need explicit expressions for the various factors.
Nevertheless, we do provide these explicit expressions, as we expect them to be useful in further work in this direction.
We also give an alternative computational proof of \eqref{factorization:boundary:intro}, to illustrate the power of the universal approach.

This is a `boundary counterpart' to the level-0 theory of the universal R-matrix, which we also include for reference.
We do this in Section \ref{sec:Uqhatsl2}, staying close to the original work by Drinfeld and Jimbo \cite{Dr85,Dr86,Ji86a,Ji86b}.
In particular, Theorem \ref{thm:R(z):action} states that the grading-shifted universal R-matrix has a well-defined action as a linear-operator-valued formal power series on any tensor product of level-0 representations of $U_q(\wh\mfb^+)$ and $U_q(\wh\mfb^-)$ (including finite-dimensional representations).
Often this well-definedness is tacitly assumed, see e.g.~\cite[Sec.~2.3]{VW20}.
It also follows from the Khoroshkin-Tolstoy factorization \cite{KT92} of the universal R-matrix, see also \cite{BGKNR10,BGKNR13,BGKNR14}; however we are unaware of such a factorization for the universal K-matrix.

\subsection{Outline}

In Section \ref{sec:Uqhatsl2} we study the action of the universal R-matrix of quantum affine $\mfsl_2$ on tensor products of level-0 representations of Borel subalgebras.
Section \ref{sec:augmentedqOns} is a `boundary analogue' to Section \ref{sec:Uqhatsl2}, where we consider the augmented q-Onsager algebra. 
We show that its \emph{(semi-)standard} universal K-matrix, see \cite{AV22a,AV22b}, has a well-defined action on level-0 representations of $U_q(\wh\mfb^+)$, see Theorem \ref{thm:K(z):action}, and satisfies the above three desirable properties.

In Section \ref{sec:Borelreps} we discuss the relevant representations of $U_q(\wh\mfb^+)$ in terms of (an extension of) the q-oscillator algebra, as well as the $U_q(\wh\mfb^+)$-intertwiner $\cO$. 
Various solutions of Yang-Baxter equations are obtained in Section \ref{sec:LandR} as actions of the universal R-matrix in tensor products of Borel representations. 
Similarly, in Section \ref{sec:K} we introduce solutions of the reflection equation as actions of the universal K-matrix in Borel representations.

We revisit the SES approach to Baxter's Q-operators for the open XXZ spin chain in light of the universal K-matrix formalism in Section \ref{sec:fusionintw}.
Next, in Section \ref{sec:boundaryfactorization} we give a diagrammatic motivation of the boundary factorization identity \eqref{factorization:boundary:intro} for the open XXZ spin chain, and provide a short proof using the level-0 theory developed in Section \ref{sec:augmentedqOns}.
Finally in Section \ref{sec:discussion} we summarize the main results and point out future work.

Some supplementary material is given in appendices. 
Namely, Appendix \ref{app:qexp} provides some background material on deformed Pochhammer symbols and exponentials. 
Moreover, Appendix \ref{app:R-operators} contains derivations of the explicit expressions of the two R-operators appearing in \eqref{factorization:boundary:intro}.
In Appendix \ref{app:altproof} we provide a computational alternative proof of the boundary factorization identity \eqref{factorization:boundary:intro}, relying on the explicit expressions of all involved factors.
The key tool of this proof is provided by Lemma \ref{lem:qexp:auxeqns}, which consists in two product formulas involving deformed Pochhammer symbols and exponentials. 
We emphasize that the main text and its results do not rely on Appendices \ref{app:R-operators} and \ref{app:altproof}.

\subsection*{Acknowledgments}

B.V. would like to thank A. Appel, P. Baseilhac and N. Reshetikhin for useful discussions. 
This research was supported in part by funding from EPSRC grant EP/R009465/1, from the Simons Foundation and the Centre de Recherches Mathématiques (CRM), through the Simons-CRM scholar-in-residence programme, and by the Galileo Galilei Institute (GGI) scientific programme on ‘Randomness, Integrability and Universality’. 
R.W. would like to acknowledge and thank CRM and the GGI for their hospitality and support.
A.C. is grateful for support from EPSRC DTP award EP/M507866/1.

\subsection*{Data availability statement}

Data sharing is not applicable to this article as no datasets were generated or analysed during the current study.

\section{Quantum affine $\mfsl_2$ and its universal R-matrix} \label{sec:Uqhatsl2}

In this section we study the action of the universal R-matrix of the quasitriangular Hopf algebra quantum affine $\mfsl_2$ on tensor products of level-0 representations (including infinite-dimensional representations) of the Borel subalgebras.
We give a basic survey of the algebras involved, the representations and the quasitriangular structure and show that the universal R-matrix has a well-defined action on tensor products of all level-0 representations of the Borel subalgebras.

\subsection{General overview of finite-dimensional R-matrix theory}

To formulate a quantum integrable system in terms of a transfer matrix built out of R-matrices, one needs finite-dimensional representations of a suitable quasitriangular Hopf algebra.
To get trigonometric R-matrices, one can proceed as follows.

Let $\mfg$ be a finite-dimensional simple Lie algebra and note that the untwisted loop algebra $L\mfg = \mfg \ot \C[t,t^{-1}]$ has a central extension $\wh{\mfg} = L\mfg \oplus \C c$.
In turn, this can be extended to $\wt{\mfg} = \wh{\mfg} \oplus \C d$ where $d$ satisfies $[d,\cdot] = t \frac{\sf d}{{\sf d}t}$.
For a fixed Cartan subalgebra $\mfh \subset \mfg$ we define
\[ 
\wh{\mfh} := \mfh \oplus \C c, \qq \wt{\mfh} := \wh{\mfh} \oplus \C d.
\]
The Lie algebra $\wt\mfg$ is a Kac-Moody algebra and hence has a non-degenerate bilinear form $(\cdot,\cdot)$, which restricts to a non-degenerate bilinear form on $\wt\mfh$. See e.g.~\cite{Ka90} for more detail.

The universal enveloping algebras $U(\wh\mfg)$ and $U(\wt\mfg)$ can be q-deformed, yielding non-cocommutative Hopf algebras (Drinfeld-Jimbo quantum groups) $U_q(\wh\mfg)$ and $U_q(\wt\mfg)$, see e.g.~\cite{Dr85,Dr86,Ji86a,KT92,Lu94}.
The nondegenerate bilinear form $(\cdot,\cdot)$ lifts to $U_q(\wt{\mfg})$ inducing a pairing between the q-deformed Borel subalgebras and hence a quasitriangular structure.
On the other hand, the subalgebra $U_q(\wh{\mfg})$ has a rich finite-dimensional representation theory, see e.g.~\cite{CP94,CP95,Ch02,HJ12}. 
The grading-shifted universal R-matrix has a well-defined action on tensor products of finite-dimensional representations of $U_q(\wh{\mfg})$ as a formal power series, see e.g.~\cite{Dr86,FR92,KS95,EM03,He19}).
We now discuss the natural extension of this theory to level-0 representations of Borel subalgebras, including various infinite-dimensional representations.
We will restrict to the case $\mfg = \mfsl_2$ (but the theory generalizes to any quantum untwisted affine algebra).

\subsection{Quantum affine $\mfsl_2$}

Denoting the canonical Cartan generator of $\mfsl_2$ by $h_1$, $\wh \mfh$ is spanned by $h_0 = c-h_1$ and $h_1$. 
The bilinear form on $\wt\mfh$ is defined by 
\[
(h_0,h_0)=(h_1,h_1)=-(h_0,h_1)=2, \qq (h_0,d)=1, \qq (h_1,d)=(d,d)=0.
\]

Fix $\eps \in \C$ such that $q = \exp(\eps)$ is not a root of unity.
For all $\mu \in \C$ we will denote $\exp(\eps \mu)$ by $q^\mu$.
First, we define $U_q(\mfg)$ as the algebra generated over $\C$ by $e$, $f$ and invertible $k$ subject to the relations
\eq{
k e = q^2 e k, \qq k f = q^{-2} f k, \qq [e,f] = \frac{k-k^{-1}}{q-q^{-1}}.
}
The following assignments determine a coproduct $\Del: U_q(\mfg) \to U_q(\mfg) \ot U_q(\mfg)$:
\eq{ \label{Delta:def}
\Del(e) = e \ot 1 + k \ot e, \qq \Del(f) = f \ot k^{-1} + 1 \ot f, \qq \Del(k^{\pm 1}) = k^{\pm 1} \ot k^{\pm 1}.
}
It uniquely extends to a Hopf algebra structure on $U_q(\mfg)$.
Now the main algebra of interest, $U_q(\wh\mfg)$, arises as follows.
\begin{defn}[Quantum affine $\mfsl_2$] \label{def:Uqhatsl2} 
We denote by $U_q(\wh\mfg)$ the Hopf algebra generated by two triples $\{ e_i, f_i, k_i \}$ ($i \in \{0,1\}$), such that:
\begin{enumerate}
\item the following assignments for $i \in \{0,1\}$ define Hopf algebra embeddings from $U_q(\mfg)$ to $U_q(\wh\mfg)$:
\eq{ 
e \mapsto e_i, \qq f \mapsto f_i, \qq k \mapsto k_i;
}
\item the following cross relations are satisfied:
\begin{gather}
k_i k_j = k_j k_i, \qq k_i e_{j} = q^{-2} e_{j} k_i, \qq k_i f_{j} = q^2 f_{j} k_i, \qq [e_i,f_{j}] = 0, \\
[e_i,[e_i,[e_i,e_{j}]_{q^2}]_1]_{q^{-2}} = [f_i,[f_i,[f_i,f_{j}]_{q^2}]_1]_{q^{-2}} = 0,
\end{gather}
for $i \ne j$, where we have introduced the notation $[x,y]_p := xy-pyx$.\hfill\defnend
\end{enumerate}
\end{defn}

Consider the affine Cartan subalgebra $\wh\mfh = \C h_0 \oplus \C h_1$.
Note that its q-deformation $U_q(\wh\mfh) = \langle k_0^{\pm 1}, k_1^{\pm 1} \rangle$ is isomorphic to the group algebra of the affine co-root lattice 
\eq{
\wh Q^\vee = \Z h_0 + \Z h_1 \subset \wh\mfh.
}
The nontrivial diagram automorphism $\Phi$ of the affine Dynkin diagram, i.e. the nontrivial permutation of the index set $\{0,1\}$, lifts to a linear automorphism $\Phi$ of $\wh\mfh$ which preserves the lattice $\wh Q^\vee$.
Accordingly, it also lifts to an involutive Hopf algebra automorphism of $U_q(\wh \mfg)$, also denoted $\Phi$, via the assignments 
\eq{ \label{Phi:def}
\Phi(e_i) = e_{\Phi(i)}, \qq \Phi(f_i) = f_{\Phi(i)}, \qq \Phi(k_i^{\pm 1}) = k_{\Phi(i)}^{\mp 1} \qq \text{for } i \in \{0,1\}.
}

\subsection{Quantized Kac-Moody algebra}

To define the quantized Kac-Moody algebra $U_q(\wt\mfg)$, one chooses an extension $\wt Q^\vee$ of $\wh Q^\vee$ (a lattice of rank 3 contained in $\wt\mfh$) preserved by $\Phi$.

\begin{rmk}
The standard extension of the affine co-root lattice $\Z h_0 + \Z h_1 + \Z d$ is not so convenient for us, mainly in view of the construction of the universal K-matrix in Section \ref{sec:uniK}.
Namely, extensions of $\Phi$ to $\wt\mfh$ which are compatible with the bilinear form on $\wt\mfh$ do not preserve this lattice, see also \cite[Sec.~2.6]{Ko14} and \cite[Sec.~3.14]{AV22a}. \hfill \rmkend
\end{rmk}

The most convenient choice is to use the \emph{principal grading} and set 
\eq{
d_{\sf pr} := -\frac{1}{8} h_0 + \frac{3}{8} h_1 + 2 d \in \mfh,
}
so that
\[
(d_{\sf pr},h_0)=(d_{\sf pr},h_1)=1, \qq (d_{\sf pr},d_{\sf pr})=0.
\]
Now we set $\Phi(d_{\sf pr})=d_{\sf pr}$ and obtain a linear automorphism $\Phi$ of $\wt\mfh$ preserving the lattice 
\[
\wt Q^\vee := \Z h_0 + \Z h_1 + \Z d_{\sf pr}.
\]
The corresponding dual map on $\wt\mfh^*$, also denoted by $\Phi$, preserves the extended affine weight lattice 
\eq{
\wt P = \{ \la \in \wt\mfh^* \, | \, \la(\wt Q^\vee) \subseteq \Z \}.
}

Accordingly, we define $U_q(\wt\mfg)$ as the Hopf algebra obtained by extending $U_q(\wh\mfg)$ by a group-like element\footnote{
It is equal to $\exp(\eps d_{\sf pr})$ if we define $U_q(\wt\mfg)$ as a topological Hopf algebra over $\C[[\eps]]$.
} $g$ satisfying 
\eq{ \label{g:relations}
g e_i = q e_i g, \qq g f_i = q^{-1} f_i g, \qq g k_i = k_i g.
}
Hence, the assignment $\Phi(g)=g$ together with \eqref{Phi:def} defines an involutive Hopf algebra automorphism of $U_q(\wt\mfg)$.

\subsection{Co-opposite Hopf algebra structure}

For any $\C$-algebra $A$, denote by $\si$ the algebra automorphism of $A \ot A$ which sends $a \ot a'$ to $a' \ot a$ for all $a,a' \in A$.
If $X \in A \ot A$ we will also write $X_{21}$ for $\si(X)$.

If $A$ is a bialgebra with coproduct $\Del$, the \emph{co-opposite bialgebra}, denoted $A^{\sf cop}$, is the bialgebra with the same underlying algebra structure and counit as $A$ but with $\Del$ replaced by 
\eq{
\Del^{\sf op} := \sigma \circ \Del
}
(if $A$ is a Hopf algebra with invertible antipode $S$, then $A^{\sf cop}$ is also a Hopf algebra with antipode $S^{-1}$). 
The assignments
\eq{ \label{om:def}
\om(e_i) = f_i, \qq \om(f_i) = e_i, \qq \om(k_i^{\pm 1}) = k_i^{\mp 1} \qq \text{for } i \in \{0,1\}, \qq \qq \om(g) = g^{-1}
}
define a bialgebra isomorphism from $U_q(\wt\mfg)$ to $U_q(\wt\mfg)^{\sf cop}$ (in particular, $(\om \ot \om) \circ \Del = \Del^{\sf op} \circ \om$) which commutes with $\Phi$.

\subsection{Weight modules} 

We review some basic representation-theoretic notions for $U_q(\wt\mfg)$ by means of which its universal R-matrix can be described.
Consider the commutative subalgebra 
\eq{
U_q(\wt\mfh) = \langle k_0^{\pm 1}, k_1^{\pm 1}, g^{\pm 1} \rangle \subset U_q(\wt\mfg).
}
Call a $U_q(\wt\mfg)$-module $M$ a $U_q(\wt\mfh)$-weight module if 
\[
M = \bigoplus_{\la \in \wt P} M_\la, \qq M_\la = \{ m \in M \, | \, k_i \cdot m = q^{\la(h_i)} m \text{ for } i \in \{0,1\}, \, g \cdot m = q^{\la(d_{\sf pr})} m \}.
\]
Elements of $M_\la$ are said to have weight $\la$.
The adjoint action of $U_q(\wt\mfh)$ (with its generators acting by conjugation) endows $U_q(\wt\mfg)$ itself with a $U_q(\wt\mfh)$-weight module structure, with elements of $U_q(\wt\mfh)$ of weight 0.
More precisely, the weights of 
$U_q(\wt\mfg)$ are given by the affine root lattice 
\[
\wh Q := \Z \al_0 + \Z \al_1 \subset \wt P
\]
($e_i$ has weight $\al_i$, $f_i$ has weight $-\al_i$).
The adjoint action of $U_q(\wt\mfh)$ preserves the subalgebras 
\eq{
U_q(\wh\mfn^+)= \langle e_0,e_1 \rangle, \qq U_q(\wh\mfn^-)= \langle f_0,f_1 \rangle
}
and the corresponding weights are given by the monoids $\pm \wh Q^+$ respectively, where  $\wh Q^+ := \Z_{\ge 0} \al_0 + \Z_{\ge 0} \al_1$.

\subsection{Quasitriangularity}

The universal R-matrix for $U_q(\wt\mfg)$ is an element of a completion of $U_q(\wt\mfg) \ot U_q(\wt\mfg)$ satisfying
\begin{gather}
\label{R:axiom1} \cR \Del(u) = \Del^{\rm op}(u) \cR \qq \text{for all } u \in U_q(\wt\mfg), \\ 
\label{R:axiom2} (\Del \ot \id)(\cR) = \cR_{13} \cR_{23}, \qq \qq (\id \ot \Del)(\cR) = \cR_{13} \cR_{12}
\end{gather}
and hence
\eq{ \label{R:YBE}
\cR_{12} \cR_{13} \cR_{23} = \cR_{23} \cR_{13} \cR_{12}.
}
Consider the quantum analogues of the Borel subalgebras, which are the Hopf subalgebras
\[
U_q(\wt\mfb^\pm) = \langle U_q(\wt\mfh), U_q(\wh\mfn^\pm) \rangle.
\]
The element $\mc{R}$ arises as the canonical element of the bialgebra pairing between $U_q(\wt\mfb^+)$ and the algebra $U_q(\wt{\mfb}^-)^{\sf op}$ (the bialgebra isomorphic as a coalgebra to $U_q(\wt{\mfb}^-)$ but with the opposite multiplication), see \cite{Dr85,Lu94}.
In particular, $\cR$ lies in a completion of $U_q(\wt\mfb^+) \otimes U_q(\wt\mfb^-)$.
Further, invariance properties of the bialgebra pairing imply
\begin{align} 
\label{omega:R}
(\om \ot \om)(\cR) &= \cR_{21}, \\
\label{Phi:R}
(\Phi \ot \Phi)(\cR) &= \cR.
\end{align}
Also, this pairing has a non-degenerate restriction to $U_q(\wh\mfn^+)_\la \times U_q(\wh\mfn^-)_{-\la}$ for all $\la \in \wh Q^+$; denote the canonical element of this restricted pairing by $\Theta_\la$.
With our choice of the coproduct we have
\eq{ \label{R:factorization}
\cR = \Theta^{-1} \cdot \ka^{-1}, \qq \Theta = \sum_{\la \in \wh Q^+} \Theta_\la, \qq 
}
A priori, $\Theta$ acts naturally on $U_q(\wt\mfg)$-modules with a locally finite action of $U_q(\wh\mfn^+)$ or $U_q(\wh\mfn^-)$.
We briefly explain one possible definition\footnote{
Note that in the topological Hopf algebra setting one simply has $\ka = q^{c \ot d + d \ot c + h_1 \ot h_1/2}$.
} of the element $\ka$.
The non-degenerate bilinear form $(\cdot,\cdot)$ on $\wt\mfh$ induces one on $\wt\mfh^*$, which we denote by the same symbol.
If $M,M'$ are $U_q(\wt\mfh)$-weight modules we define a linear map $\ka_M: M \ot M' \to M \ot M'$ by stipulating that it acts on $M_\la \ot M'_{\la'}$ ($\la,\la' \in \wt P$) as multiplication by $q^{(\la,\la')}$. 
The family of these maps $\ka_M$, where $M$ runs through all $U_q(\wt\mfh)$-weight modules, is compatible with $U_q(\wt\mfh)$-intertwiners. 
Hence it gives rise to a well-defined weight-0 element $\ka$ of the corresponding completion of $U_q(\wt\mfg) \ot U_q(\wt\mfg)$ which we call here \emph{weight completion}.
Similarly, we will define weight-0 elements of the weight completion of $U_q(\wt\mfg)$ itself using functions from $\wt P$ to $\C$.
See also \cite[Sec.~4.8]{AV22a} for more detail. 

\subsection{Level-0 representations} \label{sec:level0}

Consider the following subalgebras of $U_q(\wh\mfg)$:
\eq{
U_q(\wh\mfb^\pm) = \langle U_q(\wh\mfh), U_q(\wh\mfn^\pm) \rangle = U_q(\wt\mfb^\pm) \cap U_q(\wh\mfg).
}
Then $U_q(\wh\mfb^+)$ is isomorphic to the algebra with generators $e_i$, $k_i$ ($i \in \{0,1\}$) subject to those relations in Definition \ref{def:Uqhatsl2} which do not involve the $f_i$ (the proof of e.g.~\cite[Thm.~4.21]{Ja96} applies).
We say that a $U_q(\wh\mfb^+)$-module $V$ is \emph{level-0} if it decomposes as 
\eq{ \label{V:decomposition}
V = \bigoplus_{\ga \in \C^\times} V(\ga), \qq V(\ga) = \{ v \in V \, | \, k_0 \cdot v = \ga^{-1} v, \qu k_1 \cdot v = \ga v \}
}
with each \emph{weight subspace} $V(\ga)$ finite-dimensional.
Note that the class of finitely generated level-0 modules is closed under tensor products.
By the $U_q(\wh\mfg)$-relations we have $e_0 \cdot V(\ga) \subseteq V(q^{-2}\ga)$, $e_1 \cdot V(\ga) \subseteq V(q^2\ga)$. 
It is convenient to call the subset $\{ \ga \in \C^\times \, | \, \dim(V(\ga)) \ne 0 \}$ the \emph{support} of $V$.
If $V$ is a finite-dimensional $U_q(\wh\mfg)$-module then it is level-0 with support contained in $\pm q^\Z$, see e.g.~\cite[Prop.~12.2.3]{CP95}.

\begin{rmk} \label{}
It is known that there are no nontrivial finite-dimensional $U_q(\wt\mfg)$-modules.
More generally (note \cite[Prop.~3.5]{HJ12}), if $V$ is an irreducible level-0 $U_q(\wh\mfb^+)$-module with $\dim(V)>1$, then the $U_q(\wh\mfb^+)$-action does not extend to a $U_q(\wt\mfb^+)$-action.
To see this, choose distinct $\ga,\ga' \in \C^\times$ in the support of $V$.
By irreducibility, for any nonzero $v \in V(\ga)$, $v' \in V(\ga')$ there exist $x,x' \in U_q(\wh\mfb^+)$ such that $x \cdot v = v'$, $x' \cdot v' = v$.
Without loss of generality, we may assume that both $x$ and $x'$ have no term in $U_q(\wh\mfh)$, so that $x'x$ is not a scalar.
Assume now that the $U_q(\wh\mfb^+)$-action extends to an action of $U_q(\wt\mfb^+)$; then the action of $g$ must preserve $V(\ga)$.
For any nonzero $v \in V(\ga)$, acting with $g$ on $(x'x) \cdot v = v$ now yields a contradiction with \eqref{g:relations}. \hfill \rmkend
\end{rmk}

Analogous definitions and comments can be made for $U_q(\wh\mfb^-)$-modules.

\subsection{The action of $\cR$ on tensor products of level-0 modules} \label{sec:R:action}

We wish to connect the quasitriangular structure of $U_q(\wt\mfg)$ with the level-0 representation theory of $U_q(\wh\mfg)$, i.e. let the universal R-matrix of $U_q(\wt\mfg)$ act on tensor products of level-0 modules.
To do this, we follow the ideas from \cite[Sec.~13]{Dr86} (also see \cite[Sec.~4]{FR92}, \cite[Sec.~1]{He19}). 
If we write the action of $k_1$ on an arbitrary level-0 module $V$ as $\exp(\eps H_V)$, then note that the factor $\ka$ naturally acts on tensor products $V \ot V'$ of level-0 modules as $\exp(\eps H_{V} \ot H_{V'}/2)$.

To let $\Theta$ act on such tensor products, we extend the field of scalars $\C$ to the Laurent polynomial ring $\C[z,z^{-1}]$, where $z$ is a formal parameter.
The action of $\Theta$ is particularly well-behaved if we use the principal grading.
That is, we define a Hopf algebra automorphism $\Sigma_z$ of $U_q(\wt\mfg)[z,z^{-1}]$ such that
\eq{
\Sigma_z(e_i) = z e_i, \qq \Sigma_z(f_i) = z^{-1} f_i, \qq \Sigma_z|_{U_q(\wt\mfh)} = \id.
}
Straightforwardly one sees that
\begin{align} \label{omega:Sigma} 
\om \circ \Sigma_z &= \Sigma_{z^{-1}} \circ \om, \\
\label{Phi:Sigma}
\Phi \circ \Sigma_z &= \Sigma_z \circ \Phi.
\end{align}
Let the height function $\h: \wh Q \to \Z$ be defined by $\h(m_0 \al_0 + m_1 \al_1)=m_0+m_1$ for all $m_0,m_1 \in \Z$ and note that the number of elements of $\wh Q^+$ of given height is finite.
The key observation is that 
\eq{ \label{Theta:keyobservation}
(\Sigma_z \ot \id)(\Theta) = (\id \ot \Sigma_{z^{-1}})(\Theta) = \sum_{r \ge 0} z^r \sum_{\la \in \wh Q^+, \, \h(\la)=r} \Theta_\la,
}
is a formal power series in $z$ whose coefficients are finite sums and hence lie in $U_q(\wh\mfn^+) \ot U_q(\wh\mfn^-)$. 
Hence $(\Sigma_z \ot \id)(\Theta)=(\id \ot \Sigma_{z^{-1}})(\Theta)$ has a well-defined action as a linear-operator-valued formal power series on a tensor product of any $U_q(\wh\mfn^+)$-representation with any $U_q(\wh\mfn^-)$-representation.
Consider now the \emph{grading-shifted universal R-matrix}:
\eq{ \label{R(z)uni:def}
\cR(z) := (\Sigma_z \ot \id)(\cR) = (\id \ot \Sigma_{z^{-1}})(\cR).
}
Note that by applying $\Sigma_z \ot \id$ to \eqref{R:axiom1} we deduce that $\cR(z)$ commutes with $\Del(k_1) = \Del^{\sf op}(k_1) = k_1 \ot k_1$.
We collect the results obtained thus far, writing 
\[
M[[z]] = M \otimes \C[[z]]
\]
for any complex vector space $M$ (in particular, any complex unital associative algebra).

\begin{thrm} \label{thm:R(z):action}
Consider a pair of level-0 representations $\pi^\pm: U_q(\wh\mfb^\pm) \to \End(V^\pm)$.
Then\footnote{Note that in Section \ref{sec:LandR} we will use the notation $\mc R_{\pi^+\pi^-}(z)$ for a rescaled version of the action of the grading-shifted universal R-matrix.}
\eq{ \label{R(z):def}
\cR_{\pi^+\pi^-}(z) := (\pi^+ \ot \pi^-)(\cR(z)) \in \End(V^+ \ot V^-)[[z]]
}
is well-defined and commutes with $(\pi^+ \ot \pi^-)(\Del(k_1)) = \pi^+(k_1) \ot \pi^-(k_1)$.
\end{thrm}

From now on we will use the standard convention that if $\pi$ is any level-0 representation then the corresponding grading-shifted representation is denoted by a subscript $z$:
\eq{ \label{reps:gradingshift}
\pi_z := \pi \circ \Sigma_z.
}
Hence we may write 
\[
\cR_{\pi^+\pi^-}(z) = (\pi^+_z \ot \pi^-)(\cR) = (\pi^+ \ot \pi^-_{1/z})(\cR).
\]
Consider two indeterminates $z_1,z_2$.
Applying, say, $\Sigma_{z_1} \ot \id \ot \Sigma_{1/z_2}$, to \eqref{R:YBE}, we obtain a $\C[[z_1,z_2]]$-version of the universal Yang-Baxter equation which can be evaluated on suitable triple tensor products. 

\begin{prop} \label{prop:R(z):YBE}
If $\pi^+: U_q(\wh\mfb^+) \to \End(V^+)$, $\pi: U_q(\wh\mfg) \to \End(V)$ and $\pi^-: U_q(\wh\mfb^-) \to \End(V^-)$ are level-0 representations, then we have the following identity of linear-operator-valued formal power series in two indeterminates:
\eq{ \label{R(z):YBE}
\cR_{\pi^+\pi}(z_1)_{12} \; \cR_{\pi^+\pi^-}(z_1z_2)_{13} \; \cR_{\pi\pi^-}(z_2)_{23} = \cR_{\pi\pi^-}(z_2)_{23} \; \cR_{\pi^+\pi^-}(z_1z_2)_{13} \; \cR_{\pi^+\pi}(z_1)_{12}.
}
\end{prop}

Given a pair of level-0 representations $\pi^\pm: U_q(\wh\mfb^\pm) \to \End(V^\pm)$ it is often convenient to have an explicit expression of $\cR_{\pi^+\pi^-}(z)$ which does not rely on computing the coefficients of the series $\cR(z)$.
Essentially following Jimbo's approach from \cite{Ji86b}, we may try to solve a linear equation for $\cR_{\pi^+\pi^-}(z)$.
To derive such a linear equation, it is convenient to assume that, say, $\pi^+$ extends to a representation of $U_q(\wh\mfg)$.
In this case\footnote{
One can of course apply $\pi^+_z \ot \pi^-$ to \eqref{R:axiom1} for arbitrary $U_q(\wh\mfb^\pm)$-representations $\pi^\pm$, yielding \eqref{R(z):intw} for all $u \in U_q(\wh\mfg)$ such that $\Del(u)$ and $\Del^{\sf op}(u)$ both lie in $U_q(\wh\mfb^+) \ot U_q(\wh\mfb^-)$. 
However, by applying counits this subalgebra is seen to be equal to $U_q(\wh\mfb^+) \cap U_q(\wh\mfb^-) = U_q(\wh\mfh)$.
Hence, one would just recover the second statement of Theorem \ref{thm:R(z):action}.
}, one directly obtains the following result.

\begin{prop} \label{prop:R(z):intw}
If $\pi^+$ is a level-0 $U_q(\wh\mfg)$-representation and $\pi^-$ a level-0 $U_q(\wh\mfb^-)$-representation, then for all $u \in U_q(\wh\mfb^-)$ we have
\eq{ \label{R(z):intw}
\cR_{\pi^+\pi^-}(z) \cdot (\pi^+_z \ot \pi^-)(\Del(u)) = (\pi^+_z \ot \pi^-)(\Del^{\sf op}(u)) \cdot \cR_{\pi^+\pi^-}(z).
}
\end{prop}
Obviously there is a counterpart of Proposition \ref{prop:R(z):intw} with the role of $U_q(\wh\mfb^-)$ replaced by $U_q(\wh\mfb^+)$.

\begin{rmk} \label{rmk:R(z):YBE:define}
If the solution space of the linear equation \eqref{R(z):intw} is 1-dimensional, Proposition \ref{prop:R(z):intw} implies that any solution of \eqref{R(z):intw} must be a scalar multiple of $\cR_{\pi^+\pi^-}(z)$ and hence satisfy the Yang-Baxter equation.
This is well-known if both $V^\pm$ are finite-dimensional $U_q(\wh\mfg)$-modules.
In this case the existence of the universal R-matrix implies the existence of a solution of the intertwining condition \eqref{R(z):intw} depending rationally on $z$.
If $\pi^+$ and $\pi^-$ are also both irreducible then it is known, see e.g.~\cite[Sec.~4.2]{KS95} and \cite[Thm.~3]{Ch02}, that $V^+((z)) \ot V^-$ is irreducible as a representation of $U_q(\wh\mfg)((z))$ (extension of scalars to formal Laurent series); hence an application of Schur's lemma yields the 1-dimensionality of the solution space of \eqref{R(z):intw}.
In this case, the rational intertwiner is called \emph{trigonometric R-matrix}.
For more background and detail, see e.g.~\cite{He19} and \cite[Secs.~2.6 \& 2.7]{AV22b}. 

In the absence of a linear relation such as \eqref{R(z):intw}, one can use the Yang-Baxter equation \eqref{R(z):YBE} to determine an explicit expression for one of $\cR_{\pi^+\pi}(z)$, $\cR_{\pi^+\pi^-}(z)$, or $\cR_{\pi\pi^-}(z)$, provided the other two are known.
\hfill\rmkend
\end{rmk}

\subsection{Adjusting the grading} \label{sec:grading}

In this approach the use of the principal grading in Theorem \ref{thm:R(z):action} avoids further constraints on the representations (e.g.~local finiteness conditions). 
For completeness we briefly explain how to extend the results of Section \ref{sec:R:action} to arbitrary grading.
For nonnegative integers $s_0,s_1$ such that $s_0+s_1$ is nonzero, define a more general Hopf algebra automorphism $\Sigma^{s_0,s_1}_z$ of $U_q(\wt\mfg)[z,z^{-1}]$ by
\eq{
\Sigma^{s_0,s_1}_z(e_i) = z^{s_i} e_i, \qq \Sigma^{s_0,s_1}_z(f_i) = z^{-s_i} f_i, \qq \Sigma^{s_0,s_1}_z|_{U_q(\wt\mfh)} = \id
}
(note that the choice $s_0=0$, $s_1=1$ is used in in \cite[Eq.~(2.11)]{KT14}).

Rather than giving generalized versions of the main results above and of various statements in the remainder of this work, we make an observation which will allow the reader to generate these statements, as required. 
Recalling the decomposition \eqref{V:decomposition} and the associated terminology, suppose the level-0 $U_q(\wh\mfb^+)$-module $V$ is generated by a nonzero element of $V(\ga_0)$ for some $\ga_0 \in \C^\times$ (which includes all modules considered in this paper and all irreducible finite-dimensional $U_q(\wh\mfg)$-modules).
Then the support of $V$, see Section \ref{sec:level0}, is contained in $q^{2\Z} \ga_0$. 
Now for any indeterminate $y$ and any integer $m$, let $y^{mD}$ denote the linear map on $V$ which acts on $V(q^{-2m}\ga_0)[y,y^{-1}]$ as scalar multiplication by $y^m$. 

Writing the corresponding representation as $\pi: U_q(\wh\mfb^+) \to \End(V)$, the more general grading-shifted representation $\pi^{s_0,s_1}_z := \pi \circ \Sigma^{s_0,s_1}_z$ can be related to the representation shifted by the principal grading as follows.
Adjoining to the ring $\C[z,z^{-1}]$ a square root $Z$ of $z$, we have
\eq{
\pi^{s_0,s_1}_z = \Ad\big( Z^{(s_0-s_1)D} \big) \circ \pi_{Z^{s_0+s_1}},
}
where on the right-hand side $\Ad$ stands for `conjugation by'. 
See \cite[Sec.~5.2]{AV22b} for essentially the same point in the context of irreducible finite-dimensional $U_q(\wh\mfg)$-representations.

\section{The augmented q-Onsager algebra, its twist and its universal K-matrix} \label{sec:augmentedqOns}

In parallel with the previous section, we consider a particular subalgebra of $U_q(\wh\mfg)$ and extend some recent results on universal K-matrices \cite{AV22a,AV22b} in the context of (possibly infinite-dimensional) level-0 representations of Borel subalgebras of quantum affine $\mfsl_2$. 
For a related approach tailored to evaluation representations involving essentially the same subalgebra, see \cite{BT18}. 

\subsection{The twist map $\psi$}

We consider the following algebra automorphism and coalgebra antiautomorphism of $U_q(\wt\mfg)$:
\eq{ \label{psi:def}
\psi := \om \circ \Phi.
}
From \eqrefs{omega:R}{Phi:R} and \eqrefs{omega:Sigma}{Phi:Sigma}, respectively, we immediately deduce
\begin{align} 
\label{psi:R}
(\psi \ot \psi)(\cR) &= \cR_{21}, \\
\label{psi:Sigma}
\psi \circ \Sigma_z &= \Sigma_{z^{-1}} \circ \psi.
\end{align}

By the following result, P-symmetric R-matrices ($\cR(z)_{21} = \cR(z)$) naturally arise in tensor products of representations of the upper and lower Borel subalgebras on the same vector space, provided they are related through $\psi$ and the principal grading is used in the definition of grading-shifted universal R-matrix $\cR(z)$, see \eqref{R(z)uni:def}.

\begin{lemma} \label{lem:R:Psymmetry}
Consider two pairs of level-0 representations $\pi^\pm, \vrho^\pm: U_q(\wh\mfb^\pm) \to \End(V)$ such that 
\eq{ \label{pi:psi:relation}
\vrho^\mp = \pi^\pm \circ \psi.
}
Then $\cR_{\pi^+\pi^-}(z)_{21} = \cR_{\vrho^+\vrho^-}(z)$.
\end{lemma}

\begin{proof}
Unpacking the definitions \eqref{R(z):def} and \eqref{R(z)uni:def}, we have
\[
\cR_{\pi^+\pi^-}(z)_{21} = \Big( \big( (\pi^+ \ot \pi^-) \circ (\Sigma_z \ot \id)  \big)(\cR) \Big)_{21} = \big( (\pi^- \ot \pi^+) \circ (\id \ot \Sigma_z) \big)\big(\cR_{21}\big).
\]
Now using \eqrefs{psi:R}{psi:Sigma} we deduce 
\[
\cR_{\pi^+\pi^-}(z)_{21} = \big( (\pi^- \ot \pi^+) \circ (\psi \ot \psi) \circ (\id \ot \Sigma_{z^{-1}}) \big)(\cR).
\]
Applying \eqref{pi:psi:relation} and using \eqref{R(z):def} and \eqref{R(z)uni:def} once again, we obtain $\cR_{\vrho^+\vrho^-}(z)$ as required. 
\end{proof}

\subsection{The augmented q-Onsager algebra}

The map $\psi$ plays an important role in the theory of diagonal matrix solutions with a free parameter of the reflection equation in $U_q(\wh\mfg)$-modules.
Namely, fix a parameter $\xi \in \C^\times$ and consider the following subalgebra of $U_q(\wh\mfg)$, also called the \emph{(embedded) augmented q-Onsager algebra}:
\eq{ \label{Uqk:def}
U_q(\mfk) := \C\big\langle e_0 - q^{-1} \xi^{-1} k_0 f_1, \, e_1 - q^{-1} \xi k_1 f_0, \, k_0 k_1^{-1}, \, k_0^{-1} k_1 \big\rangle.
}
This is a left coideal: 
\eq{ 
\Del(U_q(\mfk)) \subseteq U_q(\wh\mfg) \ot U_q(\mfk).
}
The automorphism $\psi$ is the trivial q-deformation of a Lie algebra automorphism of $\wh\mfg$, also denoted $\psi$, and $U_q(\mfk)$ is the ($\xi$-dependent) coideal q-deformation of the universal enveloping algebra of the fixed-point subalgebra $\mfk = \wh\mfg^\psi$, in the style of \cite{Ko14} but with opposite conventions.

\begin{rmk}
See \cite[Rmk.\ 2.3]{VW20} for more background on this subalgebra.
Note that the definition of $U_q(\mfk)$ in \emph{loc.~cit.}~has a misprint: $\xi$ should be replaced by $\xi^{-1}$. \hfill \rmkend
\end{rmk}

To connect with the universal K-matrix formalism of \cite{AV22a,AV22b}, let $\wt S$ be the bialgebra isomorphism\footnote{In particular, $\wt S$, like the antipode $S$ itself, is an algebra antiautomorphism and a coalgebra antiautomorphism.} from $U_q(\wt\mfg)$ to $U_q(\wt\mfg)^{\sf op,cop}$ (also known as the \emph{unitary antipode}) defined by the assignments 
\eq{ \label{unitaryantipode}
\wt S(e_i) = -q k_i^{-1} e_i, \qq \wt S(f_i) = -q^{-1} f_i k_i, \qq \wt S(k_i^{\pm 1}) = k_i^{\mp 1}, \qq \wt S(g^{\pm 1}) = g^{\mp 1}.
}
Note that $\wt S^2=\id$.
Now consider\footnote{
In general, each element or map in the right coideal setting of \cite{Ko14,AV22a,AV22b} is denoted by a prime on the corresponding object in the current left coideal setting.
} 
the right coideal subalgebra 
\[
U_q(\mfk)' = \wt S(U_q(\mfk)) = \C \langle f_0 - q \xi^{-1} e_1 k_0^{-1}, f_1 - q \xi e_0 k_1^{-1}, k_0 k_1^{-1}, k_0^{-1} k_1 \rangle \]
which is the subalgebra considered in \cite[Sec.~9.7]{AV22a}, forming part of a more general family of right coideal subalgebras (quantum symmetric pair subalgebras) of quantum affine algebras as considered in \cite{Ko14,AV22a,AV22b}.

\subsection{Universal K-matrix} \label{sec:uniK}

By \cite[Thm.~8.5]{AV22a}, $U_q(\wt\mfg)$ is endowed with a so-called \emph{standard} universal K-matrix, which is an invertible element in a completion of $U_q(\wt\mfb^+)$ satisfying a twisted $U_q(\mfk)$-intertwining property and a twisted coproduct formula involving the universal R-matrix\footnote{
Note that our convention for the coproduct is as in \cite{AV22a}, but the ordering of the tensor product of the two Borel subalgebras is opposite. 
Hence the R-matrix in \cite{AV22a}, denoted here by $\cR'$, is equal to $\cR_{21}^{-1}$.
} 
\eq{
\cR' = \cR_{21}^{-1}.
} 
There is an action of invertible elements of a completion of $U_q(\wt\mfg)$, gauge-transforming the universal K-matrix and the twisting operator simultaneously, see \cite[Sec.~3.6]{AV22b}.
For the case under consideration, there exists a gauge transformation (a `Cartan correction', see \cite[Sec.~8.8]{AV22a}) that brings both the intertwining property and the coproduct formula for the universal K-matrix into a particularly nice form.
Moreover, the gauge-transformed universal K-matrix still resides in a completion of $U_q(\wt\mfb^+)$ and, when shifted by the principal grading, acts as a linear-operator-valued formal power series for all level-0 $U_q(\wh\mfb^+)$-modules.

To make this more precise, let $\Om: \wt P \to \C^\times$ be any group homomorphism such that $\Om(\al_0)=-\xi$ and $\Om(\al_1)=-\xi^{-1}$. 
Now define a function $G': \wt P \to \C^\times$ by
\eq{ \label{Gprime:def}
G'(\la) = \Om(\la) q^{-(\Phi(\la),\la)/2}.
}
Note that this is not a group homomorphism.
Define the corresponding linear operator acting on $U_q(\wt\mfh)$-weight modules as follows:
\eq{
G' \cdot v = G'(\la) v, \qq v \in V_\la, \qq \la \in \wt P.
}
Analogously to our definition of the factor $\ka$ of the universal R-matrix, we thus obtain an invertible element $G'$ of the weight completion of $U_q(\wt\mfg)$. 
Finally, let $\del=\al_0+\al_1$ be the basic imaginary root of $\wh\mfg$.
Then the following result is a special case of \cite[Sec.~9.7]{AV22a}, with the coproduct formula a direct consequence of \cite[(8.21)]{AV22a}.

\begin{prop}\label{prop:uniK:rightcoideal}
There exists an invertible element 
\eq{
\Upsilon' = \sum_{\la \in \Z_{\ge 0} \del} \Upsilon'_\la, \qq \Upsilon'_\la \in U_q(\wh{\mfn}^+)_\la,
}
such that the invertible element
\eq{
\cK' := G' \cdot \Upsilon'
}
satisfies
\begin{align}
\label{K':axiom1} \cK' \cdot u &= \psi(u) \cdot \cK' \qq \text{for all } u \in U_q(\mfk)', \\
\label{K':axiom2} \Del(\cK') &= (1 \ot \cK') \cdot (\psi \ot \id)(\cR') \cdot (\cK' \ot 1).
\end{align}
\end{prop}

\begin{rmk}
In general, a universal K-matrix $\cK'$ satisfying the simple 3-factor coproduct formula \eqref{K':axiom2} is called \emph{semi-standard}, see \cite[Sec.~8.10]{AV22a} and cf.\cite[Ex.~3.6.3 (2)]{AV22b}. 
It corresponds to a particular choice of a \emph{twist pair} $(\psi,J)$ where $\psi$ is a \emph{bialgebra isomorphism} from $U_q(\wt\mfg)$ to $U_q(\wt\mfg)^{\sf cop}$ (essentially the composition of $\omega$ with a diagram automorphism determined by the coideal subalgebra) and $J$ is the trivial Drinfeld twist $1 \ot 1$, see \cite[Sec.~2.4 and 2.5]{AV22a}. 
For any coideal q-deformation as considered by \cite{Ko14} (of the fixed-point subalgebra of a suitable Lie algebra automorphism $\theta: \wh\mfg \to \wh\mfg$) there exists a semi-standard K-matrix. 
Whenever $\theta$ is the composition of $\om$ and a diagram automorphism, as is the case here, this semi-standard universal K-matrix is a standard universal K-matrix, i.e., lies in a completion of $U_q(\wt\mfb^+)$, 
\hfill \rmkend
\end{rmk}

Now we transform this formalism \cite{AV22a} for the right coideal subalgebra $U_q(\mfk)'$, expressed in terms of the universal R-matrix $\cR'$, to a formalism for the left coideal subalgebra $U_q(\mfk) = \wt S(U_q(\mfk)')$, expressed in terms of the universal R-matrix $\cR$ as used in this paper. 
To do this, note that, when going from a $U_q(\wt\mfg)$-weight module to its dual, weights transform as $\la \mapsto -\la$. 
This defines the extension of $S$ and $\wt S$ to a map on the weight completion of $U_q(\wt\mfg)$.
Therefore $\wt S(\Om) = \Om^{-1}$ but the non-group-like factor of $G'$ is fixed by $\wt S$.
We define $G: \wt P \to \C^\times$ by
\eq{ \label{G:def}
G(\la) := \Om(\la) q^{(\Phi(\la),\la)/2}
}
so that $G = \wt S(G')^{-1}$.
Also, we set
\eq{\Upsilon := \wt S(\Upsilon')^{-1} = \sum_{\la \in \Z_{\ge 0} \del} \Upsilon_\la, \qq \Upsilon_\la \in \wt S(U_q(\wh\mfn^+)_\la) \subset U_q(\wh\mfh) \cdot U_q(\wh\mfn^+)_\la. 
}

\begin{prop} \label{prop:uniK:leftcoideal}
The element 
\eq{
\cK := \wt S(\cK')^{-1} = G \cdot \Upsilon
}
satisfies
\begin{align}\label{K:axiom1} \cK \cdot u &= \psi(u) \cdot \cK \qq \text{for all } u \in U_q(\mfk), \\
\label{K:axiom2} \Del(\cK) &= (\cK \ot 1) \cdot (\id \ot \psi)(\cR) \cdot (1 \ot \cK).
\end{align}
\end{prop}

\begin{proof}
This follows straightforwardly from Proposition \ref{prop:uniK:rightcoideal}.
Namely, we apply $\wt S$ to \eqref{K':axiom1} and $(\wt S \ot \wt S) \circ \sigma$ to \eqref{K':axiom2}, and use the fact that $\wt S \circ \psi = \psi \circ \wt S$ and $(\wt S \ot \wt S)(\cR) = \cR$.
\end{proof}

Note that $U_q(\wh\mfb^+)$ is a bialgebra and, as expected, the right-hand side of \eqref{K:axiom2} lies in a completion of $U_q(\wh\mfb^+) \ot U_q(\wh\mfb^+)$, since $\psi$ interchanges the two Borel subalgebras $U_q(\wh\mfb^\pm)$.
The reflection equation satisfied by the universal element $\cK$ is as follows:
\eq{ \label{K:RE}
\cR \cdot (\cK \ot 1) \cdot (\id \ot \psi)(\cR) \cdot (1 \ot \cK) = (1 \ot \cK) \cdot (\id \ot \psi)(\cR) \cdot (\cK \ot 1) \cdot \cR.
}
This is a consequence of the linear relation \eqref{R:axiom1} for $\cR$ and the coproduct formula \eqref{K:axiom2} for $\cK$, alongside \eqref{psi:R} and $\psi^2=\id$.

\subsection{The action of the universal K-matrix on level-0 representations}

To deduce that $\cK$ has a well-defined action on level-0 representations of, say, $U_q(\wh\mfb^+)$, we can proceed in a similar way to the case of the R-matrix. 
This builds on the finite-dimensional theory for more general quantum symmetric pair subalgebras in \cite[Sec.~4]{AV22b}.

First note that if $\pi$ is a level-0 representation, $\pi$ and the twisted representation $\pi \circ \psi$ coincide on $U_q(\wh\mfh)$.
Now let $z$ once again be a formal variable.
Note that by \eqref{G:def} the function $G$ sends the basic imaginary root $\del$ to 1.
Hence the proof of \cite[Prop.~4.3.1 (3)]{AV22b} implies that the corresponding factor $G$ of the universal K-matrix descends to level-0 modules. 
Furthermore, the argument that shows $\Sigma_z(\Theta)$ is a $U_q(\wh\mfn^+) \ot U_q(\wh\mfn^-)$-valued formal power series can be easily adapted to $\Upsilon$; it yields a formal power series with coefficients in $\wt S(U_q(\wh\mfn^+)) \subset U_q(\wh\mfb^+)$:
\[
\Sigma_z(\Upsilon) = \sum_{r \ge 0} z^r \sum_{\la \in \Z_{\ge 0} \del, \, \h(\la)=r} \Upsilon_\la.
\]\
Now consider the grading-shifted universal K-matrix:
\eq{ \label{K(z):def}
\cK(z) = \Sigma_z(\cK).
}

Noting that the form of $\Upsilon$ implies that $\cK$ commutes with $k_1$, we arrive at the following main result, which is a boundary analogue of Theorem \ref{thm:R(z):action}.

\begin{thrm} \label{thm:K(z):action}
Consider a level-0 representation $\pi: U_q(\wh\mfb^+) \to \End(V)$. 
Then\footnote{In Section \ref{sec:K} we will use this notation for a rescaled version of the action of the grading-shifted universal K-matrix.}
\eq{
\cK_{\pi}(z) := \pi(\cK(z)) \in \End(V)[[z]]
}
is well-defined and commutes with $\pi(k_1)$.\end{thrm}

We will also need boundary counterparts of Propositions \ref{prop:R(z):YBE} and \ref{prop:R(z):intw}.
Consider two indeterminates $z_1,z_2$.
Applying $\Sigma_{z_1} \ot \Sigma_{z_2}$ to \eqref{K:RE} and using \eqref{psi:Sigma}, we obtain the following reflection equation for the grading-shifted universal operators:
\eq{ \label{K(z)univ:RE}
\begin{aligned}
& \cR(z_1/z_2) \cdot (\cK(z_1) \ot 1) \cdot (\id \ot \psi)(\cR(z_1z_2)) \cdot (1 \ot \cK(z_2)) = \qq \\
&\qq = (1 \ot \cK(z_2)) \cdot (\id \ot \psi)(\cR(z_1z_2)) \cdot (\cK(z_1) \ot 1) \cdot \cR(z_1/z_2).
\end{aligned}
}
Recalling that the universal R-matrix $\cR$ lies in a completion of $U_q(\wh\mfb^+) \ot U_q(\wh\mfb^-)$ and applying a tensor product of suitable representations to \eqref{K(z)univ:RE}, one obtains the following \emph{right reflection equation} with multiplicative spectral parameters.

\begin{prop} \label{prop:K(z):RE}
Consider level-0 representations $\pi^+: U_q(\wh\mfb^+) \to \End(V^+)$ and $\pi: U_q(\wh\mfg) \to \End(V)$ such that $\pi \circ \psi = \pi$.
Then
\eq{ \label{K(z):RE} 
\begin{aligned}
& \cR_{\pi^+\pi}(z_1/z_2) (\cK_{\pi^+}(z_1) \ot \Id_V) \cR_{\pi^+\pi}(z_1z_2) (\Id_{V^+} \ot \cK_\pi(z_2)) = \qq \\
&\qq = (\Id_{V^+} \ot \cK_\pi(z_2)) \cR_{\pi^+\pi}(z_1z_2) (\cK_{\pi^+}(z_1) \ot \Id_V) \cR_{\pi^+\pi}(z_1/z_2).
\end{aligned}
}
\end{prop}

The use of linear relations to find explicit solutions of reflection equations was proposed in \cite{MN98,DG02,DM03}.
As before, we assume that $\pi$ extends to a $U_q(\wh\mfg)$-representation\footnote{
Analogous to the case of the R-matrix, we can observe that the intersection of $U_q(\mfk)$ and $U_q(\wh\mfb^+)$ is contained in $U_q(\wh\mfh)$.
Therefore, applying a level-0 representation $\pi$ to \eqref{K:axiom1} just recovers the second part of Theorem \ref{thm:K(z):action}.
}, in which case it restricts to a $U_q(\mfk)$-representation and we obtain the following result as a consequence of \eqref{psi:Sigma}.

\begin{prop} \label{prop:K(z):intw}
If $\pi: U_q(\wh\mfg) \to \End(V)$ is a level-0 representation such that $\pi \circ \psi = \pi$, then 
\eq{ \label{K(z):intw}
\cK_\pi(z) \cdot \pi_z(u) = \pi_{1/z}(u) \cdot \cK_\pi(z) \qq \text{for all } u \in U_q(\mfk).
}
\end{prop}

We close this section with some comments parallel to Remark \ref{rmk:R(z):YBE:define}.

\begin{rmk} 
If the solution space of \eqref{K(z):intw} is 1-dimensional, Proposition \ref{prop:K(z):intw} implies that any solution must be a scalar multiple of $\mc K(z)$ and hence automatically satisfy the reflection equation \eqref{K(z):RE}.
In the case that $\pi: U_q(\wh\mfb^+) \to \End(V)$ extends to a representation and $V$ is finite-dimensional, there is an analogue to Remark \ref{rmk:R(z):YBE:define}.
Namely, the solution space of \eqref{K(z):intw} for irreducible representations is 1-dimensional and the existence of a solution of the intertwining condition \eqref{K(z):intw} depending rationally on $z$ leads to a \emph{trigonometric K-matrix}. 
See \cite[Secs.~5 and 6]{AV22b} for more detail.

To explicitly determine $\cK_{\pi^+}(z)$ in the cases where $\pi^+: U_q(\wh\mfb^+) \to \End(V)$ does not extend to a $U_q(\wh\mfg)$-representation, we will use the reflection equation \eqref{K(z):RE}, with the other K-matrix $\cK_\pi(z)$ determined using Proposition \ref{prop:K(z):intw}.
\hfill \rmkend
\end{rmk}

\section{Borel representations in terms of the q-oscillator algebra} \label{sec:Borelreps}

\subsection{The infinite-dimensional vector space $W$}

The countably-infinite-dimensional vector space plays a central role in the  theory of Baxter's Q-operators. 
We may define it as the free $\C$-module over a given set $\{ w_j \}_{j \in \Z_{\ge 0}}$:
\eq{
W = \bigoplus_{j \ge 0} \C w_j.
}
Given this distinguished basis, elements of $\End(W)$ naturally identify with infinite-by-infinite matrices with the property that all but finitely many entries of each column are zero. 

It is convenient to work with a particular subalgebra of $\End(W)$ depending on the deformation parameter $q$.
More precisely, consider the $\C$-linear maps $a$, $a^\dag$ on $W$ defined by
\eq{
a \cdot w_{j+1} = w_j, \qq a \cdot w_0 = 0, \qq a^\dag \cdot w_j = \big( 1-q^{2(j+1)} \big) w_{j+1}
}
for all $j \in \Z_{\ge 0}$.
For the description of L-operators associated to $U_q(\wh\mfg)$ acting on $W \otimes \C^2$ (particular solutions of the Yang-Baxter equation), it is convenient to consider a linear operator $q^D$ which is a square root of $1-a^\dag a$, i.e. $q^D \cdot w_j = q^j w_j$ for $j \in \Z_{\ge 0}$.
Note that $q^D$ is invertible and we let $q^{-D}$ denote its inverse.

\begin{rmk}
Often the q-oscillator algebra is defined as an abstract algebra, generated by $a$, $a^\dag$ and $q^{\pm D}$ subject to certain relations, which naturally embeds into $\End(W)$.
This version of the q-oscillator algebra appeared in the guise of a topological algebra for instance in \cite[Sec. 2.3]{BGKNR10} and with slightly different conventions in \cite{KT14}\footnote{
The two vector spaces $W_1$ and $W_2$ introduced in \cite[Sec. 2.3]{KT14} are naturally isomorphic, so that the two algebras ${\rm Osc}_1$ and ${\rm Osc}_2$ can be identified with the same subalgebra of $\End(W_1) \cong \End(W_2)$.}. \hfill \rmkend
\end{rmk}

\subsection{Diagonal operators from functions and an extended q-oscillator algebra} \label{sec:extqosc}

To accommodate the action of the universal R and K-matrices on certain level-0 modules, we will need an extension of the commutative subalgebra $\langle q^{\pm D} \rangle$ and work over the commutative ring $\C[[z]]$. 

Denote by $\cF$ the commutative algebra of functions from $\Z_{\ge 0}$ to $\C[[z]]$.
For any $f \in \cF$ we define $f(D) \in \End(W)[[z]]$ via 
\eq{
f(D) \cdot w_j = f(j) w_j.
}
Thus, we obtain an algebra embedding $\cF \to \End(W)[[z]]$. 
Now we combine this with the maps $a$, $a^\dag$ defined above (viewed as maps on $W \ot \C[[z]]$, acting trivially on the second factor).

\begin{defn}
The \emph{(extended) q-oscillator algebra} is the subalgebra $\cA \subset \End(W)[[z]]$ generated by $a^\dag$, $a$ and $\cF(D)$.
\hfill \defnend 
\end{defn}

As can be verified on basis vectors, in $\cA$ one has the relations
\eq{ \label{A:basicrels2}
a a^\dag = 1-q^{2(D+1)}, \qq a^\dag a = 1-q^{2D}, \qq a f(D) = f(D+1) a, \qq f(D) a^\dag = a^\dag f(D+1).
}
One straightforwardly verifies that the subalgebras $\cF(D)$, $\langle a^\dag \rangle$ and $\langle a \rangle$ are self-centralizing.
Note that the operator 
\eq{
\ba^\dag := -q^{-2D} a^\dag \in \End(W)
}
sends $w_j$ to $(1-q^{-2(j+1)})w_{j+1}$.
Clearly, $\cA$ is also generated by $\ba^\dag$, $a$ and $\cF(D)$.
The transformation $q \mapsto q^{-1}$ defines an algebra automorphism of $\cA$, preserving the subalgebra $\cF(D)$, fixing the generator $a$ and interchanging the generators $\adag$ and $\badag$.

\subsection{Endomorphisms of $W \ot W$}

The linear maps
\[
a_1:=a \ot \Id_W, \qq a^\dag_1 := a^\dag \ot \Id_W, \qq
a_2 := \Id_W \ot a, \qq a^\dag_2 := \Id_W \ot a^\dag 
\]
together with $\cF(D_1) \cup \cF(D_2)$ generate $\cA \ot \cA$ over $\C[[z]]$.
We will need a larger subalgebra of $\End(W \ot W)$: we will allow all functions of two nonnegative integers as well as formal power series in certain locally nilpotent endomorphisms.

Denote by $\cF^{(2)}$ the commutative algebra of functions from $\Z_{\ge 0} \times \Z_{\ge 0}$ to $\C[[z]]$.
For any $f \in \cF^{(2)}$ we define $f(D_1,D_2) \in \End(W \ot W)[[z]]$ via 
\eq{
f(D_1,D_2) \cdot (w_j \ot w_k) = f(j,k) w_j \ot w_k,
}
yielding an algebra embedding $\cF^{(2)} \to \End(W \ot W)[[z]]$.
Now note that $a_1 a^\dag_2$ and $a^\dag_1 a_2$ are locally nilpotent endomorphisms of $W \ot W$.
Hence, for any $g_{k,\ell}, h_{k,\ell} \in \cF^{(2)}$ series of the form
\eq{ \label{series}
\sum_{k,\ell \ge 0} (a^\dag_2)^\ell g_{k,\ell}(D_1,D_2) a_1^k, \qq \sum_{k,\ell \ge 0} (a^\dag_1)^k h_{k,\ell}(D_1,D_2) a_2^\ell
}
truncate when applied to any basis vector $w_j \ot w_{j'}$. 
We obtain a class of well-defined elements of $\End(W \ot W)[[z]]$.
We denote by $\cA^{(2)}$ the $\C[[z]]$-span of the operator-valued formal series \eqref{series}, which is easily seen to be a subalgebra of $\End(W \ot W)[[z]]$.

\subsection{The Borel representations} \label{sec:reps:plus} \mbox{}
We introduce four level-0 representations of $U_q(\wh\mfb^+)$. 
First of all, let $\mu \in \C$ be a free parameter. 
It is straightforward to check that the following assignments define a representation $\ups$ of $U_q(\wh\mfg)$ on $W$:
\eq{ \label{hom:sigma}
\begin{aligned}
\ups(e_0) &= \ups(f_1) = \frac{1}{1-q^2} a^\dag, && \ups(k_0) = q^{-\mu+1+2D}, \\
\ups(e_1) &= \ups(f_0) = \frac{q^2}{1-q^2} a (q^{-\mu} - q^{\mu-2D}), \qq && \ups(k_1) = q^{\mu-1-2D}, 
\end{aligned}
}
The module structure on $W$ defined by $\ups$ is the evaluation Verma module: affinizations of finite-dimensional irreducible $U_q(\mfsl_2)$-modules arise as quotients if $\mu \in \Z_{>0}$ (also see \cite[Sec.~2.2]{KT14}).

We will in addition consider three $U_q(\wh{\mfb}^+)$-representations which do not extend to representations of $U_q(\wh{\mfg})$.
A useful reducible representation $\phi: U_q(\wh{\mfb}^+) \to \End(W)$ is given by
\eq{
\phi(e_0) = 0, \qq 
\phi(e_1) = \frac{q}{1-q^2} a, \qq
\phi(k_0) = q^{\mu+1+2D}, \qq
\phi(k_1) = q^{-\mu-1-2D}
}
which is closely related to the special evaluation homomorphism defined in \cite[Eq.~(4.6)]{KT14}.
The following representations $\vrho, \brho: U_q(\wh{\mfb}^+) \to \End(W)$ play an essential role in the definition of Baxter Q-operators:\eq{ \label{homs:plus}
\begin{aligned}
\vrho(e_0) &= \frac{1}{1-q^2} a^\dag, 
& \vrho(e_1) &= \frac{q^2}{1-q^2} a, 
& \vrho(k_0) &= q^{2D}, 
& \vrho(k_1) &=  q^{-2D}, \\
\brho(e_0) &= \frac{q^2}{1-q^2} \ba^\dag, \qu
& \brho(e_1) &= \frac{1}{1-q^2} a, \qu
& \brho(k_0) &= q^{2(D+1)}, \qu
& \brho(k_1) &= q^{-2(D+1)}.
\end{aligned}
}
They correspond to the representations $L^\pm_{1,a}$ introduced in \cite[Def.~3.7]{HJ12} (for suitable $a \in \C^\times$) called \emph{prefundamental} representations in \cite{FH15}, in which their role in the construction of Q-operators for closed chains is studied.

We will henceforth repeatedly denote grading-shifted representations by the notation \eqref{reps:gradingshift}.
Note that the grading-shifted representation $\ups_z$ is an algebra homomorphism from $U_q(\wh\mfg)$ to $\End(W)[z,z^{-1}]$.
Furthermore, the grading-shifted representations $\ups_z|_{U_q(\wh{\mfb}^+)}$, $\phi_z$, $\vrho_z$, $\brho_z$ are algebra homomorphisms from $U_q(\wh{\mfb}^+)$ to $\End(W)[z] \subset \End(W)[[z]]$.
Finally, note that $\vrho_z$, $\brho_z$ correspond to the representations defined by \cite[Eq.~(3.5)]{KT14}.

\begin{rmk} \label{rmk:signdifference}
The grading-shifted representation in \cite[Eq.~(2.9)]{VW20} is related to $\vrho_z$ by a factor of $-1$ in the actions of $e_0$ and $e_1$: in other words it is equal to $\vrho_{-z}$. 
Since the Baxter Q-operators only depend on $z^2$, see \cite[Lem.~4.5]{VW20}, there are no serious discrepancies. 
The benefit of the current choice is its consistency across the relevant level-0 representations, with $\ups$ having the same sign convention as finite-dimensional representations such as $\Pi$, see Section \ref{sec:LandR}. \hfill \rmkend
\end{rmk}

\subsection{The $U_q(\wh\mfb^+)$-intertwiner $\cO$}\label{sec:O+}

The tensor products $\vrho_{q^{-\mu/2} z} \ot \brho_{q^{\mu/2}z}$ and $\ups_{z} \ot \phi_{z}$ of shifted representations are closely related in the following sense: the two induced $U_q(\wh\mfb^+)$-actions on $W \ot W$ are conjugate by an element in $\cA^{(2)}$ which is independent of $z$.
More precisely, consider the deformed exponential
\eq{ \label{qexp:def}
e_{q^2}(x) = \sum_{k=0}^\infty \frac{x^k}{(q^2;q^2)_k}.
}
We refer to Appendix \ref{app:qexp} for more details on this formal series.
We define the following element of $\GL(W \ot W)$:
\eq{ \label{Oplus:def}
\cO = e_{q^2}(q^2 a_1 \ba^\dag_2)^{-1} q^{\mu (D_1-D_2)/2}.
}
The following statement is \cite[Eq.~(4.4)]{KT14} and connects to \cite[Thm.~3.8]{FH15}; for completeness we provide a proof in the present conventions.

\begin{thrm}\label{thm:O:plus}
The $U_q(\wh\mfb^+)$-representations $\vrho_{q^{-\mu/2} z} \ot \brho_{q^{\mu/2} z}$ and $\ups_{z} \ot \phi_{z}$ are intertwined by $\cO$:
\eq{ \label{O:plus:intw}
\cO \, \big( \vrho_{q^{-\mu/2} z} \ot \brho_{q^{\mu/2} z} \big)(\Del(u))  = \big( \ups_{z} \ot \phi_{z} \big)(\Del(u)) \,\, \cO \qq \text{for all } u \in U_q(\wh\mfb^+).
}
\end{thrm}

\begin{proof} 
The relations \eqrefs{qexp:commute}{qexp:Dba} can be evaluated at $y=q^2$, yielding
\begin{align*}
q^{\mu (D_2-D_1)/2} e_{q^2}(q^2a_1 \ba^\dag_2) \ba^\dag_2 &= \big( q^{-\mu/2} a^\dag_1 + q^{2(D_1+1)+\mu/2} \ba^\dag_2 \big) q^{\mu (D_2-D_1)/2} e_{q^2}(q^2a_1 \ba^\dag_2), \\
\multicolumn{2}{c}{$q^{\mu(D_2-D_1)/2} e_{q^2}(q^2a_1 \ba^\dag_2) \big( a_1 (q^{-2\mu} - q^{- 2D_1}) + q^{-2(D_1+1)} a_2 \big)= \qq \qq$} \\
\multicolumn{2}{c}{$\qq \qq = \big( a_1 q^{-3\mu/2} + q^{-\mu/2 - 2(D_1+1)} a_2 \big) q^{\mu (D_2-D_1)/2} e_{q^2}(q^2a_1 \ba^\dag_2)$,} \\
q^{\mu(D_2-D_1)/2} e_{q^2}(q^2a_1 \ba^\dag_2) q^{2(D_1+D_2+1)} &= q^{2(D_1+D_2+1)} q^{\mu(D_2-D_1)/2} e_{q^2}(q^2 a_1 \ba^\dag_2) ,\\
q^{\mu (D_2-D_1)/2} e_{q^2}(q^2 a_1 \ba^\dag_2) q^{-2(D_1+D_2+1)}  &= q^{-2(D_1+D_2+1)} q^{\mu (D_2-D_1)/2} e_{q^2}(q^2 a_1 \ba^\dag_2).
\end{align*}
These directly imply \eqref{O:plus:intw} for $u \in \{ e_0,e_1,k_0,k_1\}$.
\end{proof}

\subsection{Formalism for $U_q(\wh\mfb^-)$}

Recall from \eqref{psi:def} the automorphism $\psi$ which interchanges the two Borel subalgebras.
Note that the representation $\ups: U_q(\wh\mfg) \to \End(W)$ satisfies
\eq{ \label{sigma:psi:symmetry}
\ups = \ups \circ \psi.
}
Hence, it is natural to define representations of $U_q(\wh\mfb^-)$ corresponding to $\vrho$, $\brho$ and $\phi$, as follows:
\eq{ \label{maps:minus}
\vrho^- := \vrho\circ \psi, \qq
\brho^{\, -} := \brho \circ \psi, \qq
\phi^- := \phi \circ \psi.
}
Explicitly, we have
\eq{ \label{homs:minus}
\begin{aligned}
\vrho^-(f_0) &= \frac{q^2}{1-q^2} a, \qu & 
\vrho^-(f_1) &= \frac{1}{1-q^2} a^\dag, \qu
& \vrho^-(k_0) &= q^{2D}, & 
\vrho^-(k_1) &= q^{-2D}, \\
\brho^{\, -}(f_0) &= \frac{1}{1-q^2} a, &
\brho^{\, -}(f_1) &= \frac{q^2}{1-q^2} \ba^\dag, &
\brho^{\, -}(k_0) &= q^{2(D+1)}, & 
\brho^{\, -}(k_1) &= q^{-2(D+1)}, \\
\phi^-(f_0) &= \frac{q}{1-q^2} a, & 
\phi^-(f_1) &= 0, &
\phi^-(k_0) &= q^{\mu+1+2D}, \qu &
\phi^-(k_1) &= q^{-\mu-1-2D}.
\end{aligned}
}
By \eqref{psi:Sigma}, whereas the grading-shifted representations $\vrho_z$, $\brho_z$, $\phi_z$ take values in $\End(W) \ot \C[z]$, their negative counterparts $\vrho^-_z$, $\brho^{\, -}_z$, $\phi^-_z$ take values in $\End(W) \ot \C[z^{-1}]$.

Since $\psi$ is a coalgebra antiautomorphism, using \eqref{psi:Sigma} we immediately deduce the following characterization of the tensorial opposite of the intertwiner $\cO$.

\begin{crl}\label{crl:O:minus}
The linear map 
\eq{ \label{Omin:def}
\cO_{21} = e_{q^2}(q^2\ba^\dag_1 a_2)^{-1} q^{\mu (D_2-D_1)/2}  \in \End(W \ot W).
}
intertwines the $U_q(\wh{\mfb}^-)$-representations $\brho^{\, -}_{q^{-\mu/2} z} \ot \vrho^-_{q^{\mu/2} z}$ and $\phi^-_z \ot \ups_z$, viz. 
\eq{ \label{O:min:intw}
\cO_{21} \, \big( \brho^{\, -}_{q^{-\mu/2} z} \ot \vrho^-_{q^{\mu/2} z} \big)(\Del(u)) = \big( \phi^-_z \ot \ups_z \big)(\Del(u)) \, \, \cO_{21} \qq \text{for all } u \in U_q(\wh{\mfb}^-).
}
\end{crl}

\section{L-operators and R-operators} \label{sec:LandR}

In order to define L-operators, we recall the standard 2-dimensional representation
$\Pi : U_q(\wh{\mfg}) \rightarrow \End(\C^2)$ determined by
\eq{
\begin{aligned}
\Pi(e_0) = \Pi(f_1) &= \begin{pmatrix} 0 & 0 \\ 1 & 0 \end{pmatrix}, \qu & \qu 
\Pi(k_0) &= \begin{pmatrix} q^{-1} & 0 \\ 0 & q \end{pmatrix}, \\
\Pi(e_1) = \Pi(f_0) &= \begin{pmatrix} 0 & 1 \\ 0 & 0 \end{pmatrix}, \qu & \qu
\Pi(k_1) &= \begin{pmatrix} q & 0 \\ 0 & q^{-1} \end{pmatrix}.
\end{aligned}
}

In analogy with \eqref{sigma:psi:symmetry}, we have
\eq{
\label{Pi:psi:symmetry} \Pi = \Pi \circ \psi.
}

\subsection{L-operators for $U_q(\wh\mfb^+)$-modules} \label{sec:L:plus}

We will now obtain explicit formulas for certain scalar multiples of the four different actions of the grading-shifted universal R-matrix on $W \ot \C^2$.
In these cases both Theorem \ref{thm:R(z):action} and Proposition \ref{prop:R(z):intw} apply.
It turns out that the relevant linear equations all have 1-dimensional solution spaces over $\C[[z]]$.
The following linear operators are convenient scalar multiples.
\begin{align} 
\cL_\vrho(z) &= \begin{pmatrix} 
q^D & a^\dag q^{-D-1} z \\ 
a q^{D+1} z & q^{-D}-q^{D+2} z^2 
\end{pmatrix},\\[2mm]
\label{Lbar:plus:def} \cL_\brho(z) &= \begin{pmatrix} 
q^{D+1} - q^{-D+1} z^2 & \ba^\dag q^{-D} z \\ 
a q^{D} z & q^{-D-1} 
\end{pmatrix},\\[2mm]
\label{Lsigma:plus:def} \cL_\ups(z) &= \begin{pmatrix} 
q^D - q^{-D+\mu} z^2 & \adag q^{-D-2+\mu} z \\ 
a q \big( q^{D-\mu} - q^{-D+\mu} \big) z & q^{-D-1+\mu} - q^{D+1} z^2 
\end{pmatrix},\\[2mm]
\cL_\phi(z) &= \begin{pmatrix} 
q^{D+1} & 0 \\ a q^{D+1} z & q^{-D-\mu} 
\end{pmatrix}.
\end{align}

\begin{rmk}
We have abused notation by representing linear operators on $\End(W \ot \C^2)$ as $2 \times 2$ matrices with coefficients in $\End(W)$ (as opposed to the conventional usage that realizes operators on $\End(\C^2\ot W)$ in this way).  \hfill \rmkend
\end{rmk}

The following result is \cite[Cor.\ 4.2]{KT14}.

\begin{thrm}\label{thm:fund}  
The above L-operators satisfy the following relation in $\End(W \ot W \ot \C^2)[[z]]$:
\eq{ \label{fund:bulk:plus}
\cO_{12}  \cL_\vrho(q^{-\mu/2} z)_{13} \cL_\brho(q^{\mu/2} z)_{23} =  
\cL_\ups(z)_{13} \cL_\phi(z)_{23} \cO_{12} .
}
\end{thrm}

\begin{proof} 
From the coproduct formula \eqref{R:axiom2} one deduces
\begin{align*}
 \cL_\vrho(q^{-\mu/2} z)_{13} \cL_\brho(q^{\mu/2} z)_{23} \; &\propto \; (\vrho_{q^{-\mu/2} z} \ot \brho_{q^{\mu/2} z} \ot \Pi)\big((\Delta \ot \id)(\cR)\big), \\
\cL_\ups(z)_{13} \cL_\phi(z)_{23} \; &\propto \; (\ups_{z} \ot \phi_{z} \ot \Pi)\big((\Delta \ot \id)(\cR)\big).
\end{align*}
Now Theorem \ref{thm:O:plus} implies \eqref{fund:bulk:plus} up to a scalar. 
By applying both sides to $w_0 \ot w_0 \ot ({1 \atop 0})$ one observes that the scalar is 1.
\end{proof}

Given the L-operators for the various $U_q(\wh\mfb^+)$-representations, Lemma \ref{lem:R:Psymmetry} provides us with L-operators for the corresponding $U_q(\wh\mfb^-)$-representations: $\cL^-_\pi(z) = \cL_\pi(z)_{21}$ for $\pi \in \{ \vrho, \brho, \ups, \phi \}$.
These are scalar multiples of $\cR_{\Pi \vrho^-}(z)$, $\cR_{\Pi \brho^{\, -}}(z)$, $\cR_{\Pi \ups}(z)$ and $\cR_{\Pi \phi^-}(z)$, respectively.
Theorem \ref{thm:fund} immediately yields the following result:

\begin{crl}\label{thm:fund:minus}  
The following relation in $\End(\C^2 \ot W \ot W)[[z]]$ is satisfied:
\eq{ \label{fund:bulk:minus}
\cO_{32} \cL^-_{\vrho}(q^{-\mu/2} z)_{13} \cL^-_{\brho}(q^{\mu/2} z)_{12} = \cL^-_{\ups}(z)_{13} \cL^-_{\phi}(z)_{12} \cO_{32} .
}
\end{crl}

\subsection{Actions of $\cR(z)$ on tensor products of infinite-dimensional Borel representations}

By Theorem \ref{thm:R(z):action}, the grading-shifted universal R-matrix has well-defined actions on the tensor product of the level-0 modules $(\ups,W)$ and $(\phi^-,W)$ and on the tensor product of the level-0 modules $(\vrho,W)$ and $(\brho^{\, -},W)$ as $\End(W \ot W)$-valued formal power series.
Note that, using the terminology of Section \ref{sec:level0}, $\C w_0 \ot w_0 \subset W \ot W$ is the subspace of weight $q^{-2}$ and hence $w_0 \ot w_0$ is an eigenvector of both actions of the universal R-matrix with invertible eigenvalues.
It is convenient to use rescaled linear-operator-valued formal power series
\eq{
\cR_{\vrho\brho}(z), \cR_{\ups\phi}(z) \in \End(W \ot W)[[z]],
}
uniquely defined by the condition that they fix $w_0 \ot w_0$:
\eq{ \label{Rrhobrho:Rsitau:def}
\begin{aligned}
\cR_{\vrho \brho}(z) \, &\propto \, (\vrho \ot \brho^{\, -})(\cR(z)), & \cR_{\vrho \brho}(z) \cdot (w_0 \ot w_0) = w_0 \ot w_0, \\
\cR_{\ups \phi}(z) \, &\propto \, (\ups \ot \phi^-)(\cR(z)), & \qq \cR_{\ups \phi}(z) \cdot (w_0 \ot w_0) = w_0 \ot w_0.
\end{aligned}
}
These power series will appear in the boundary factorization identity.
In appendix \ref{app:R-operators} we obtain explicit expressions for $\cR_{\vrho \brho}(z)$ and $\cR_{\ups\phi}(z)$, although we will not need these for the proof of the boundary factorization identity using the universal K-matrix formalism of Section \ref{sec:augmentedqOns}.

\section{K-operators} \label{sec:K}

In this section we consider solutions of reflection equations associated to the subalgebra $U_q(\mfk)$. 

\subsection{Right K-operators}

By Theorem \ref{thm:K(z):action}, applying any of the level-0 $U_q(\wh\mfb^+)$-representations $\vrho$, $\brho$, $\ups$, $\phi$ to the grading-shifted universal K-matrix associated to $U_q(\mfk)$ we obtain $\End(W)$-valued formal power series, satisfying the reflection equation \eqref{prop:K(z):RE}.
Moreover, in terms of the terminology of Section \ref{sec:level0}, the weight subspaces of all four actions are all 1-dimensional and therefore $w_0$ is an eigenvector of each action with invertible eigenvalue.  
We will consider the scalar multiples of these linear operators which fix $w_0$:
\eq{
\label{K:def}
\cK_\pi(z) \, \propto \, \pi(\cK(z)), \qq \cK_\pi(z) \cdot w_0 = w_0.
}
for $\pi \in \{ \vrho, \brho, \ups, \phi \}$.
It is convenient to obtain explicit expressions by applying Propositions \ref{prop:K(z):RE} and \ref{prop:K(z):intw}. 
These could be found independently of the universal K-matrix formalism, either by solving the reflection equations directly in all cases or by following the approach outlined in \cite{DG02,DM03,RV16} (this relies on the irreducibility of certain tensor products as $U_q(\mfk)((z))$-modules; otherwise the reflection equation must be verified directly).

First of all, the linear operator
\eq{ \label{eq:RightBoundary:V}
K_{\Pi}(z) = 
\begin{pmatrix} 
\xi z^2 - 1 & 0 \\ 0 & \xi - z^2
\end{pmatrix} 
\in \End({\C^2}) [[z]]
}
is, up to a scalar, the unique solution of the $U_q(\mfk)$-intertwining condition 
\eq{
K_\Pi(z) \Pi_z(u) = \Pi_{1/z}(u) K_\Pi(z) \qq \text{for all } u \in U_q(\mfk).
}
By Theorem \ref{thm:K(z):action}, it is proportional to the action of the grading-shifted universal K-matrix in the representation $(\Pi,\C^2)$.

Recall that $\Pi \circ \psi = \Pi$.
Hence, motivated by Proposition \ref{prop:K(z):RE}, we consider the right reflection equation for $\pi \in \{\vrho, \brho, \ups, \phi \}$:
\eq{
\label{eq:RightBoundary:W}
\cL_\pi(\tfrac{y}{z}) \cK_{\pi}(y) \cL_{\pi}(yz) K_{\Pi}(z) = K_{\Pi}(z) \cL_\pi(y z) \cK_{\pi}(y) \cL_{\pi}(\tfrac{y}{z}) \in \End (W \ot {\C^2})[[y/z,z]].
}

\begin{lemma} \label{lem:K:explicit}
We have
\eq{ \label{eq:KMatrices}
\begin{aligned}
\cK_\vrho(z) &= (-q^{-D} \xi)^{D} (q^{2}\xi^{-1}z^2;q^2)_{D}, & \cK_\brho(z) &= (q z^2)^{-D} (q^{2} \xi^{-1} z^{-2};q^2)_D^{-1}, \\
\cK_\ups(z) &= z^{-2D} \frac{(q^{2-\mu} \xi^{-1} z^2;q^2)_D}{(q^{2-\mu} \xi^{-1} z^{-2};q^2)_D}, & 
\qq \cK_\phi(z) &= (-q^{-\mu - D-1} \; \xi)^D.
\end{aligned}
}
\end{lemma}

Note that these expressions were already given in \cite{BT18} in different conventions. 
For completeness we sketch a proof relying on the universal K-matrix formalism.

\begin{proof}[Proof of Lemma \ref{lem:K:explicit}.]
For $\cK_\ups(z)$, by a straightforward check, the intertwining condition
\eq{ \label{Ksigma:intertwine}
\cK_\ups(z) \ups_{z}(u) = \ups_{1/z}(u) \cK_\ups(z) \qq \text{for all } u \in U_q(\mfk)
}
can be solved to find $\cK_{\ups}(z)$, making use of Proposition \ref{prop:K(z):intw}.  
Since $\cK(z)$ commutes with the action of $k_1$ it follows that $\cK_\ups(z) = f(D)$ for some $f \in \cF$.
Now imposing \eqref{Ksigma:intertwine} for the generators $e_0-q^{-1}\xi^{-1}k_0f_1$ and $e_1-q^{-1}\xi k_1f_0$ yields the recurrence relation
\[
\frac{f(D+1)}{f(D)} = \frac{1-q^{2(D+1)-\mu} \xi^{-1} z^2}{z^2-q^{2(D+1)-\mu} \xi^{-1}}.
\]
In particular, the linear relation \eqref{Ksigma:intertwine} has a 1-dimensional solution space.
Together with the constraint $f(0)=1$ it yields the formula given in \eqref{eq:KMatrices}.

For $\pi \in \{ \vrho, \brho, \phi \}$, it is convenient to consider the linear space
\eq{ \label{eq:REsolspace}
{\sf RE}_\pi := \{ \cK_\pi(y) \in \cF(D) \, | \, \eqref{eq:RightBoundary:W} \text{ is satisfied} \}
}
and use Proposition \ref{prop:K(z):RE} to find the explicit expression, relying on the second part of Theorem \ref{thm:K(z):action} for the fact that $\cK_\pi(y)$ lies in $\cF(D)$. 
Indeed, the operator $K_{\vrho}(z)$ was obtained in \cite[Sec.~2.4]{VW20} as the unique element of the 1-dimensional linear space ${\sf RE}_{\vrho}$ which fixes $w_0$.
In an analogous way we obtain the result for $K_{\brho}(z)$.

Note that $\phi$ is a reducible representation.
Indeed, the solution space of \eqref{eq:RightBoundary:W} with $\pi=\phi$ is infinite-dimensional:~the general solution $\cK_\phi(z)$ is of the form $(- q^{-\mu - D -1} \; \xi)^D p$ with $p$ in the centralizer of $a$ in $\cA$,
i.e.~a polynomial in $a$ with coefficients in $\C[[z]]$. 
Since $\cK_\phi(z) \in \cF(D)$, $p$ is a scalar. 
The requirement that $w_0$ is fixed forces $p=1$.
\end{proof}

\subsection{Left K-operators} \label{sec:leftK}

We now obtain linear-operator-valued power series satisfying a reflection equation for the left boundary by using a well-established bijection, see \cite[Eq.~(15)]{Sk88}, between its solution set and the solution set of the right reflection equation.
For fixed $\txi \in \C^\times$ we define 
\eq{ \label{TildeK:V}
\wt K_{\Pi}(z) := (1-q^2 \txi^{-1} z^2)^{-1} (1-q^2 \txi z^2)^{-1} \big(
 K_{\Pi}(q z)^{-1}|_{\xi \mapsto \txi^{-1}} \big) = 
\begin{pmatrix}
q^2 \txi z^2 -1 & 0\\ 0 & \txi - q^2z^2
\end{pmatrix}.
}
Also, for $\pi \in \{ \vrho, \brho, \ups, \phi \}$  we define
\eq{  \label{TildeK:W}
\tcK_{\pi}(z) := \cK_{\pi}(q z)^{-1}|_{\xi \mapsto \txi^{-1}}.
}
Note that $\cL_\pi(\ga z)$ is invertible in $\End(W \ot \C^2)[[z]]$ for all $\ga \in \C$.
We define 
\eq{ \label{tildeL:def}
\tcL_\pi(z) = \cL_\pi(q^2z)^{-1}.
}

\begin{lemma}
For all $\pi \in \{\vrho, \brho, \ups, \phi \}$ the \emph{left reflection equation} holds:
\eq{ \label{LeftBoundary}
\tcK_\pi(y) \tcL_\pi(yz) \wt K_{\Pi}(z) \cL_\pi(\tfrac{y}{z}) = \cL_\pi(\tfrac{y}{z}) \wt K_{\Pi}(z)\tcL_\pi(yz)\tcK_\pi(y) \qu \in \End (W \ot {\C^2})[[y/z,z]].
}
\end{lemma}

\begin{proof}
The desired equation \eqref{LeftBoundary} can be rewritten as 
\[
\wt K_{\Pi}(z)^{-1}\tcL_\pi(yz)^{-1} \tcK_\pi(y)^{-1} \cL_\pi(\tfrac{y}{z}) = \cL_\pi(\tfrac{y}{z}) \tcK_\pi(y)^{-1}\tcL(yz)^{-1}\wt K_{\Pi}(z)^{-1}.
\]
By \eqrefs{TildeK:V}{tildeL:def}, this is equivalent to the right-reflection equation \eqref{eq:RightBoundary:W} with $y \mapsto qy$, $z \mapsto qz$ and $\xi \mapsto \txi^{-1}$.
\end{proof}

Using the explicit formulas \eqref{eq:RightBoundary:V} and \eqref{eq:RightBoundary:W} we obtain that the solutions of the left reflection equations \eqref{TildeK:W} are the following $\End(W)$-valued formal power series in $z$:
\begin{equation}
\label{K-tilde}
\begin{aligned}
\tcK_\vrho(z) &= (-q^D \txi)^D (q^4 \txi z^2;q^2)_{D}^{-1}, &
\tcK_\brho(z) &= (q^3 z^2)^D (\txi z^{-2} ;q^2)_{D}, \\
\tcK_{\ups}(z) &= (q z)^{2D} \frac{(q^{-\mu} \txi  z^{-2};q^2)_{D}}{(q^{4-\mu} \txi z^2;q^2)_{D}}, &
\qq \tcK_{\phi}(z) &= (-q^{\mu + D + 1} \txi)^D.
\end{aligned}
\end{equation}

\section{Fusion intertwiners revisited} \label{sec:fusionintw}

In this short intermezzo we explain how the universal K-matrix formalism naturally leads to relations involving K-operators and $U_q(\mfb^+)$-intertwiners called \emph{fusion intertwiners} which play a key role in the short exact sequence approach to the Q-operator.
These intertwiners were discussed in \cite{VW20} and the relevant relations with K-matrices were shown by a linear-algebraic computation relying on the explicit expressions of the various constituent factors, see \cite[Lem.~3.2]{VW20}.
In other words, the representation-theoretic origin of these relations was unclear, which we now remedy.

Level-0 representations of $U_q(\wh\mfb^+)$ are amenable to scalar modifications of the action of $U_q(\mfh) = \langle k_1^{\pm 1} \rangle$, see also \cite[Rmk.~2.5]{HJ12}.
In particular, for $r \in \C^\times$, define a modified Borel representation $\vrho$ as follows:
\eq{
\vrho_r(e_i) = \vrho(e_i), \qq \vrho_r(k_0) = r \vrho(k_0), \qq \vrho_r(k_1) = r^{-1} \vrho(k_1)
}
and consider the grading-shifted representation $\vrho_{r,z} := (\vrho_r)_z$.
There exist $U_q(\wh\mfb^+)$-intertwiners
\begin{align*}
\iota(r) &: (\vrho_{qr,qz},W) \to (\vrho_{r,z} \ot \Pi_z,W \ot \C^2), \\
\tau(r) &: (\vrho_{r,z} \ot \Pi_z,W \ot \C^2) \to (\vrho_{q^{-1}r,q^{-1}z},W),
\end{align*}
called \emph{fusion intertwiners}, which take part in the following short exact sequence:
\begin{equation} \label{SES:plus}
\begin{tikzcd}
0 \arrow[r] & (\vrho_{qr,qz},W) \arrow[r,"\iota(r)"] & (W \ot \C^2,\vrho_{r,z} \ot \pi_z) \arrow[r,"\tau(r)"]  & (\vrho_{q^{-1}r,q^{-1}z},W) \arrow[r] & 0 
\end{tikzcd}
\end{equation}
Explicitly\footnote{The sign mismatch with \cite[Eq.~(3.1)]{VW20} is explained in Remark \ref{rmk:signdifference}.}, we have
\eq{
\iota(r) = \begin{pmatrix} q^{-D} a^\dag \\ -q^{D+1} r \end{pmatrix}, \qq
\tau(r) = \begin{pmatrix} q^D, & q^{-D} r^{-1} a^\dag \end{pmatrix}.
}
Analogously to Theorem \ref{thm:fund}, fusion relations for the L-operators $\cL(r,z)$, defined as suitable scalar multiples of $(\vrho_{r,z} \ot \Pi)(\cR)$, now follow from these intertwining properties and the coproduct formulas for $\cR$ \eqref{R:axiom2}, see \cite[Eqns.~(3.8) and (3.9)]{VW20}.

Recalling the universal object $\cK$ and Theorem \ref{thm:K(z):action}, we define the corresponding K-operator $\cK_\vrho(r,z)$ as the unique scalar multiple of $\vrho_{r,z}(\cK)$ which fixes $w_0$ (cf.~\cite[Prop.~2.5]{VW20}).
Then
\eq{
(\vrho_{r,z} \ot \Pi_z)(\Del(\cK))  \qq \propto \qq \cK_\vrho(r,z)_1
\cL(r,z^2) K_\Pi(z)_2
}
as a consequence of \eqref{K:axiom2}.
Since $\cK$ lies in a completion of $U_q(\wh\mfb^+)$, the intertwining properties of $\iota(r)$ and $\tau(r)$ now directly yield the following fusion relation for the K-operator:
\begin{align*}
\cK_\vrho(r,z)_1 \cL(r,z^2) K_\Pi(z)_2 \iota(r) \qq &\propto \qq \iota(r) \cK_\vrho(qr,qz) \\
\tau(r) \cK_\vrho(r,z)_1 \cL(r,z^2) K_\Pi(z)_2 \qq &\propto \qq \cK_\vrho(q^{-1}r,q^{-1}z) \tau(r),
\end{align*}
with the scalar factors determined by applying the two sides of the equation to $w_0$, say. 
We will be able to prove a boundary counterpart of the factorization identity \eqref{fund:bulk:plus} using similar ideas.

We recover, with a much smaller computational burden, the key result \cite[Lemma 3.2]{VW20} (a similar relation for left K-operators can easily be deduced from this, as explained in the last sentence of \cite[Proof of Lemma 3.2]{VW20}).
In the approach to Baxter's Q-operator using short exact sequences, the fusion relations for L and K-operators induce fusion relations for 2-boundary monodromy operators, see \cite[Lem.~4.2]{VW20} from which Baxter's relation \eqref{eq:TQ1} follows by taking traces, see \cite[Sec.~5.2]{VW20}.

\section{Boundary factorization identity} \label{sec:boundaryfactorization}

In motivating and presenting the key boundary relations, it is very useful to introduce a graphical representation of spaces and operators. 
Let us introduce the following pictures for the different representations introduced in Sections \ref{sec:Borelreps} and \ref{sec:LandR}:\\
\begin{center}
  \begin{tikzpicture}[scale=0.6]
\draw(0,0) node[left]{$\vrho_z=\quad z$};
 \draw[blueline=0.5](0,0) -- (2,0); 
\draw(7,0) node[left]{$\brho_z=\quad z $};
\draw[redline=0.5] (7,0) -- (9,0);
\draw(14,0) node[left]{$\phi_{z}=\quad z$};
\draw[greenline=0.5] (14,0) -- (16,0);
 \draw(0,-2) node[left]{$\vrho^-_z=\quad z$};
 \draw[dashblueline=0.5](0,-2) -- (2,-2); 
 \draw(7,-2) node[left]{$\brho^{\, -}_z=\quad z$};
 \draw[dashredline=0.5] (7,-2) -- (9,-2);
 \draw(14,-2) node[left]{$\phi^-_{z}=\quad z$};
 \draw[dashgreenline=0.5] (14,-2) -- (16,-2);
 \draw(3,-4) node[left]{$\ups_{z}=\quad z$};
 \draw[wavy=0.5](3,-4) -- (5,-4);
 \draw(10,-4) node[left]{$\Pi_z=\quad z $};
 \draw[aline=0.5] (10,-4) -- (12,-4);
\end{tikzpicture}
\end{center}
\vspace{5mm}
For any vector spaces $V$, $V'$, denote by $\cP$ the linear map from $V \ot V'$ to $V' \ot V$ such that $\cP(v \ot v') = v' \ot v$ for all $v \in V$, $v' \in V'$.
Also set $z = z_1/z_2$.
We then have the following pictures for L-operators and R-operators:

\[
\begin{array}{ll}
\begin{tikzpicture}[scale=0.6]
\draw(-1.7,0) node[left]{$\cP \cL_\vrho(z) =$}; 
\draw[aline=0.3] (0,1)--(0,-1);
\draw[blueline=0.3] (-1,0)--(1,0); 
\draw (0,1) node[above]{$z_2$}; 
\draw (-1,0) node[left]{$z_1$}; 
\end{tikzpicture}
\qq
&
\qq
\begin{tikzpicture}[scale=0.6]
\draw(-1.7,0) node[left]{$\cP \cL_\ups(z)=$}; 
\draw[aline=0.3] (0,1)--(0,-1);
\draw[wavy=0.3] (-1,0)--(1,0);
\draw (0,1) node[above]{$z_2$}; 
\draw (-1,0) node[left]{$z_1$}; 
\end{tikzpicture}
\\
\begin{tikzpicture}[scale=0.6]
\draw(-1.7,0) node[left]{$\cP \cL_\brho(z)=$};
\draw[aline=0.3] (0,1)--(0,-1);
\draw[redline=0.3] (-1,0)--(1,0);
\draw (0,1) node[above]{$z_2$};
\draw (-1,0) node[left]{$z_1$}; 
\end{tikzpicture}
\qq
&
\qq
\begin{tikzpicture}[scale=0.6]
\draw(-1.7,0) node[left]{$\cP \cL_\phi(z)=$};
\draw[aline=0.3] (0,1)--(0,-1);
\draw[greenline=0.3] (-1,0)--(1,0);
\draw (0,1) node[above]{$z_2$};
\draw (-1,0) node[left]{$z_1$}; 
\end{tikzpicture}
\\
\begin{tikzpicture}[scale=0.6]
\draw(-1.7,0) node[left]{$\cP \cL_{\vrho}^-(z) =$}; 
\draw[aline=0.3] (0,1)--(0,-1);
\draw[dashblueline=0.3] (1,0)--(-1,0); 
\draw (0,1) node[above]{$z_1$}; 
\draw (1,0) node[right]{$z_2$}; 
\end{tikzpicture}
\qq
&
\qq
\begin{tikzpicture}[scale=0.6]
\draw(-1.7,0) node[left]{$\cP \cL^-_{\ups}(z)=$}; 
\draw[aline=0.3] (0,1)--(0,-1);
\draw[wavy=0.3] (1,0)--(-1,0);
\draw (0,1) node[above]{$z_1$};
\draw (1,0) node[right]{$z_2$};  
\end{tikzpicture}
\\
\begin{tikzpicture}[scale=0.6]
\draw(-1.7,0) node[left]{$\cP \cL^-_{\brho}(z)=$}; 
\draw[aline=0.3] (0,1)--(0,-1);
\draw[dashredline=0.3] (1,0)--(-1,0);
\draw (0,1) node[above]{$z_1$};
\draw (1,0) node[right]{$z_2$};
\end{tikzpicture}
\qq
&
\qq
\begin{tikzpicture}[scale=0.6]
\draw(-1.7,0) node[left]{$\cP \cL^-_{\phi}(z)=$};
\draw[aline=0.3] (0,1)--(0,-1);
\draw[dashgreenline=0.3] (1,0)--(-1,0);
\draw (0,1) node[above]{$z_1$};
\draw (1,0) node[right]{$z_2$}; 
\end{tikzpicture}
\\
\hspace{-4pt}
\begin{tikzpicture}[scale=0.6]
\draw(-1.7,0) node[left]{$\cP \cR_{\vrho\brho}(z) =$};
\draw[dashredline=0.3] (0,1)--(0,-1);
\draw[blueline=0.3] (-1,0)--(1,0);
\draw (0,1) node[above]{$z_2$};
\draw (-1,0) node[left]{$z_1$}; 
\end{tikzpicture}
\qq
&
\qq
\hspace{-4pt}
\begin{tikzpicture}[scale=0.6]
\draw(-1.7,0) node[left]{$\cP \cR_{\ups\phi}(z)=$};
\draw[dashgreenline=0.3] (0,1)--(0,-1);
\draw[wavy=0.3] (-1,0)--(1,0);
\draw (0,1) node[above]{$z_2$};
\draw (-1,0) node[left]{$z_1$}; 
\end{tikzpicture}
\end{array}
\]
We now make the following definitions\footnote{These are the modified forms of the R-matrices that appear in the corresponding left reflection equations, see \cite[Eq.~(13)]{Sk88}.}:
\eq{
\label{RsitauRrhobrho:tilde} 
\wt{\cR}_{\vrho\brho}(z) := \cR_{\vrho\brho}(q^2z)^{-1}, \qq
\wt{\cR}_{\ups\phi}(z) := \cR_{\ups\phi}(q^2z)^{-1}.
}
and represent these modified R-matrices by the following pictures:
\[
\begin{tikzpicture}[scale=0.6]
\draw(-1,0) node[left]{$\wt{\cR}_{\vrho\brho}(z) \cP=$}; 
\draw[dashredline=0.3] (0,1)--(0,-1);
\draw[blueline=0.3] (1,0)--(-1,0); 
\draw (0,0) node[bblob]{};
\draw (0,1) node[above]{$z_2$};
\draw (1,0) node[right]{$z_1$}; 
\end{tikzpicture}  
\qq
\qq
\qq
\begin{tikzpicture}[scale=0.6]
\draw(-1,0) node[left]{$\wt{\cR}_{\ups\phi}(z)\cP=$}; 
\draw[dashgreenline=0.3] (0,1)--(0,-1);
\draw[wavy=0.3] (1,0)--(-1,0);
\draw (0,0) node[bblob]{};
\draw (0,1) node[above]{$z_2$};
\draw (1,0) node[right]{$z_1$}; 
\end{tikzpicture}  
\]

The various right-boundary K-matrices are represented as follows:

\[
\begin{tikzpicture}[scale=0.6]  
\begin{scope} [rotate=-45]
\draw (0,-2) node[left]{$\cK_{\rho}(z)=$};
\draw[blueline=0.5](0,0)  -- (2,0);
\draw (-0.3,-0.3) node[left]{$z$};
\draw[dashblueline=0.5](2,0)  -- (2,-2);
\draw (2,-2) node[left]{$z^{-1}$};
\end{scope}
\end{tikzpicture}  
\qq
\begin{tikzpicture}[scale=0.6]  
\begin{scope} [rotate=-45,xshift=4cm,yshift=4cm]
\draw (0,-2) node[left]{$\cK_{\brho}(z)=$};
\draw[redline=0.5](0,0)  -- (2,0);  
\draw (-0.3,-0.3) node[left]{$z$};
\draw[dashredline=0.5](2,0)  -- (2,-2);
\draw (2,-2) node[left]{$z^{-1}$};
\end{scope}
\end{tikzpicture}  
\qq
\begin{tikzpicture}[scale=0.6]  
\begin{scope} [rotate=-45,xshift=8cm,yshift=8cm]
\draw (0,-2) node[left]{$\cK_{\ups}(z)=$};
\draw[wavy=0.5](0,0)  -- (2,0);
\draw (-0.3,-0.3) node[left]{$z$};
\draw[wavy=0.5](2,0)  -- (2,-2);
\draw (2,-2) node[left]{$z^{-1}$};
\end{scope}
\end{tikzpicture}  
\qq
\begin{tikzpicture}[scale=0.6]  
\begin{scope} [rotate=-45,xshift=12cm,yshift=12cm]
\draw (0,-2) node[left]{$\cK_{\phi}(z)=$};
\draw[greenline=0.5](0,0)  -- (2,0);
\draw (-0.3,-0.3) node[left]{$z$};
\draw[dashgreenline=0.5](2,0)  -- (2,-2);
\draw (2,-2) node[left]{$z^{-1}$};
\end{scope}
\end{tikzpicture}  
\]

The left-boundary K-matrices defined in Section \ref{sec:leftK} are represented by the natural analogues of these pictures. For example:\begin{center}
\begin{tikzpicture}[scale=0.6]
  \begin{scope} [rotate=-45,xshift=4cm,yshift=4cm];
  \draw (-.5,-2.5) node[left]{$\wt\cK_{\rho}(z)=$};
  \draw[blueline=0.5](0,-2)  -- (0,0);  \draw (0,0) node[right]{$z$};
 \draw[dashblueline=0.5](2,-2)  -- (0,-2);  \draw (2,-2) node[right]{$z^{-1}$};
 \end{scope}
 \end{tikzpicture}  
\end{center}

Making use of these pictures, we see that Theorem \ref{thm:fund} and Corollary \ref{thm:fund:minus} are represented by 
\begin{align*}
\begin{tikzpicture}[scale=0.5]  
\draw[redline=0.2](1,0)  -- (4,0);
\draw (1,0) node[left]{$q^{\mu/2}z_1$};
\draw[blueline=0.2](1,-2) -- (4,-2);
\draw (1,-2) node[left]{$q^{-\mu/2}z_1$};
\draw[greenline=0.8](6,0) -- (9,0);
\draw (9,0) node[right]{$z_1$};
\draw[wavy=0.8](6,-2) -- (9,-2);
\draw (9,-2) node[right]{$z_1$};
\draw (5,-1) ellipse (1cm and 2cm);
\draw(6,-1) node[left]{$\cO$};
\draw[aline=0.5](2.5,1) -- (2.5,-3);
\draw (2.5,1) node[above]{$z_2$};
\draw(11,-1) node[]{$=$};
\draw[redline=0.2,xshift=15cm](1,0)  -- (4,0);
\draw[xshift=15cm](1,0) node[left]{$q^{\mu/2}z_1$};
\draw[blueline=0.2,xshift=15cm](1,-2) -- (4,-2);
\draw[xshift=15cm] (1,-2) node[left]{$q^{-\mu/2}z_1$};
\draw[greenline=0.8,xshift=15cm](6,0) -- (9,0);
\draw[xshift=15cm] (9,0) node[right]{$z_1$};
\draw[wavy=0.8,xshift=15cm](6,-2) -- (9,-2);
\draw[xshift=15cm] (9,-2) node[right]{$z_1$};
\draw[xshift=15cm] (5,-1) ellipse (1cm and 2cm);
\draw[xshift=15cm](6,-1) node[left]{$\cO$};
\draw[aline=0.5,xshift=15cm](7.5,1) -- (7.5,-3);
\draw[xshift=15cm] (7.5,1) node[above]{$z_2$};
\end{tikzpicture}
\\
\begin{tikzpicture}[scale=0.5]
\draw[dashgreenline=0.8](4,0) -- (1,0);
\draw (1,0) node[left]{$z_1$};
\draw[wavy=0.8](4,-2) -- (1,-2);
\draw (1,-2) node[left]{$z_1$};
\draw[dashredline=0.2](9,0) -- (6,0);
\draw (9,0) node[right]{$q^{-\mu/2} z_1$};
\draw[dashblueline=0.2](9,-2) -- (6,-2);
\draw (9,-2) node[right]{$q^{\mu/2}z_1$};
\draw (5,-1) ellipse (1cm and 2cm);
\draw(6,-1) node[left]{$ \cO_{21}$};
\draw[aline=0.5](7.5,1) -- (7.5,-3);
\draw (7.5,1) node[above]{$z_2$};
\draw(13,-1) node[]{$=$};
\draw[dashgreenline=0.8,xshift=15cm](4,0) -- (1,0);
\draw[xshift=15cm] (1,0) node[left]{$z_1$};
\draw[wavy=0.8,xshift=15cm](4,-2) -- (1,-2);
\draw[xshift=15cm] (1,-2) node[left]{$z_1$};
\draw[dashredline=0.2,xshift=15cm](9,0) -- (6,0);
\draw[xshift=15cm] (9,0) node[right]{$q^{-\mu/2} z_1$};
\draw[dashblueline=0.2,xshift=15cm](9,-2) -- (6,-2);
\draw[xshift=15cm] (9,-2) node[right]{$q^{\mu/2}z_1$};
\draw[xshift=15cm] (5,-1) ellipse (1cm and 2cm);
\draw[xshift=15cm](6,-1) node[left]{$\cO_{21}$};
\draw[aline=0.5,xshift=15cm](2.5,1) -- (2.5,-3);
\draw[xshift=15cm] (2.5,1) node[above]{$z_2$};
\end{tikzpicture}
\end{align*}

For the compatibility with the right boundary we claim that
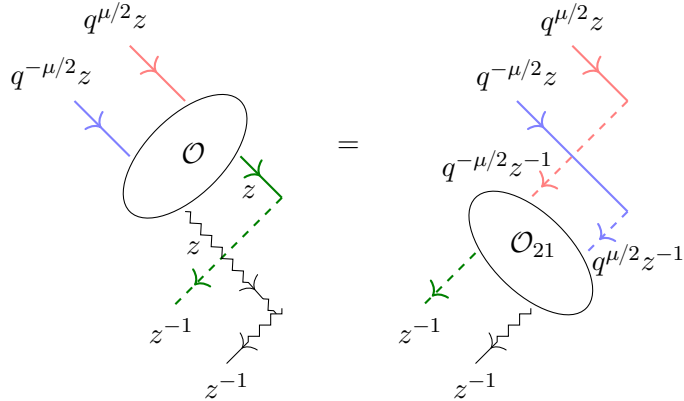
\begin{center}
  \begin{tikzpicture}[scale=0.52]
  \begin{scope} [rotate=-45]
  \draw[redline=0.5](2,0)  -- (4,0);  \draw (2,1) node[left]{$q^{\mu/2}z$};
  \draw[blueline=0.5](2,-2) -- (4,-2); \draw (2,-1) node[left]{$q^{-\mu/2}z$};
  \draw[greenline=0.5](6,0) -- (7.5,0);   \draw (6.5,-0.2) node[below]{$z$};
   \draw[dashgreenline=0.8](7.5,0) -- (7.5,-4);   \draw (7.5,-4) node[below]{$z^{-1}$};
  \draw[wavy=0.8](6,-2) -- (9.5,-2.2);  \draw (6.5,-2.2) node[below]{$z$};
   \draw[wavy=0.8](9.5,-2) -- (9.5,-4);  \draw (9.5,-4) node[below]{$z^{-1}$};
  \draw (5,-1) ellipse (1cm and 2cm);
  \draw(5.5,-0.4) node[left]{$\cO$};
   \end{scope}
    \draw(7,-4) node[]{=};
  \begin{scope} [rotate=45,xshift=-1.5cm,yshift=-12cm]
  \draw[dashgreenline=0.8](4,0) -- (2,0);  \draw (2,0) node[below]{$z^{-1}$};
  \draw[wavy=0.8](4,-2) -- (2,-2); \draw (2,-2) node[below]{$z^{-1}$};
  \draw[dashredline=0.9](9.5,0) -- (6,0);   \draw (5.5,.7) node[above]{$q^{-\mu/2}z^{-1}$};
  \draw[redline=0.5](9.5,2) -- (9.5,0);   \draw (9.5,2) node[above]{$q^{\mu/2}z$};
  \draw[blueline=0.25](7.5,2) -- (7.5,-2);  \draw (7.5,2) node[above]{$q^{-\mu/2}z$};
  \draw[dashblueline=0.7](7.5,-2) -- (6,-2);  \draw (7.3,-2.6) node[below]{$q^{\mu/2}z^{-1}$};
  \draw (5,-1) ellipse (1cm and 2cm);
  \draw(5.8,-1.5) node[left]{$\cO_{21}$};
      \end{scope}
   \end{tikzpicture}
\end{center}
which corresponds to the following identity in $\cA^{(2)}$:
\eq{ \label{keyrelation-init:right}
\cK_\ups(z)_1 \cR_{\ups\phi}(z^2) \cK_\phi(z)_2 \,\cO = \cO \, \cK_\vrho(q^{-\mu/2}z)_1 \cR_{\vrho\brho}(z^2) \cK_\brho(q^{\mu/2}z)_2,
}
which we call the \emph{right boundary factorization identity}.
The diagrams above serve as a motivation for the identity, which we now prove using results from Section \ref{sec:augmentedqOns} (an alternative computational proof of Theorem \ref{thm:keyrelation:right} is given in Appendix C).

\begin{thrm} \label{thm:keyrelation:right}
For all $\mu \in \C$, all $q \in \C^\times$ not a root of unity and all $\xi \in \C^\times$, relation \eqref{keyrelation-init:right} is satisfied.
\end{thrm}

\begin{proof}
The proof is analogous to the proof of Theorem \ref{thm:fund}.
We first note that 
\[ 
\begin{aligned}
\big( \vrho_{q^{-\mu/2}z} \ot \brho_{q^{\mu/2}z} \big)\big( (\id \ot \psi)(\cR) \big) &= \big( \vrho_{q^{-\mu/2}z} \ot \brho^{\, -}_{q^{-\mu/2}z^{-1}} \big)(\cR) && \propto \; \cR_{\vrho\brho}(z^2), \\
\big( \ups_{z} \ot \phi_{z} \big)\big( (\id \ot \psi)(\cR) \big) &= \big( \ups_{z} \ot \phi^-_{z^{-1}} \big)(\cR) && \propto \; \cR_{\ups\phi}(z^2).
\end{aligned}
\]
Noting the coproduct formula \eqref{K:axiom2}, we obtain
\[
\begin{aligned}
\cK_\vrho(q^{-\mu/2}z)_1 \cR_{\vrho\brho}(z^2) \cK_\brho(q^{\mu/2}z)_2 \qu &\propto \qu \big( \vrho_{q^{-\mu/2}z} \ot \brho_{q^{\mu/2}z} \big)(\Del(\cK)), \\
\cK_\ups(z)_1 \cR_{\ups\phi}(z^2) \cK_\phi(z)_2 \qu &\propto \qu \big( \ups_{z} \ot \phi_{z} \big)(\Del(\cK)).
\end{aligned}
\]
Now Theorem \ref{thm:O:plus} implies \eqref{keyrelation-init:right} up to a scalar. 
The fact that all factors fix $w_0 \ot w_0$ shows that the scalar is 1.
\end{proof}

Compatibility with the left boundary requires that 

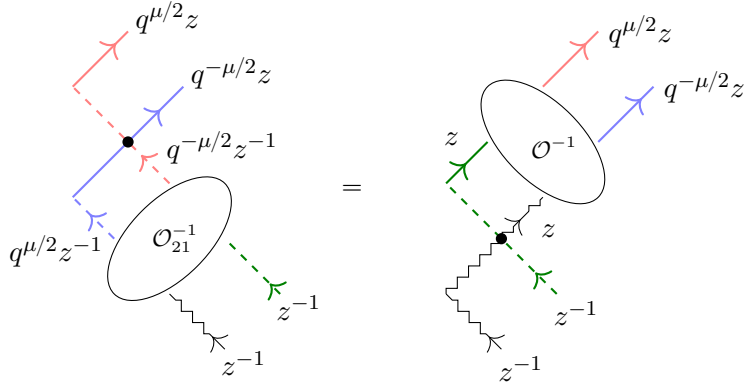
\begin{center}
  \begin{tikzpicture}[scale=0.52]
  \begin{scope}[rotate=-45]
  \draw[dashredline=0.3](4,0) -- (0.5,0);  
  \draw[redline=0.8](0.5,0) -- (0.5,2);  
  \draw[dashblueline=0.6](4,-2) -- (2.5,-2);   
  \draw[blueline=0.8](2.5,-2) -- (2.5,2);     
  \draw(2.5,0) node[bblob]{};
  \draw[dashgreenline=0.2](8,0) -- (6,0);   
  \draw (8,-0.5) node[right]{$z^{-1}$};
  \draw[wavy=0.2](8,-2) -- (6,-2);  
  \draw (8,-2.5) node[right]{$z^{-1}$};
  \draw (5,-1) ellipse (1cm and 2cm);
  \draw(5.7,-0.3) node[left]{ \mbox{\footnotesize $\cO_{21}^{-1}$}};
  \draw (0.2,2.2) node[right]{$q^{\mu/2} z$}; \draw (2.2,2.2) node[right]{$q^{-\mu/2} z$};
   \draw (2.7,-2.8) node[below]{$q^{\mu/2} z^{-1}$};   
 \draw (4.3,1.6) node{$ q^{-\mu/2} z^{-1}$};
   \end{scope}
  \draw(7.5,-3) node[]{$=$};
   \begin{scope}[rotate=45,xshift=2.5cm,yshift=-9cm]
   \draw[dashgreenline=0.2](2.5,-4) -- (2.5,0);  \draw (2.5,-4.7) node{$z^{-1}$}; \draw(2.5,-2) node[bblob]{};
  \draw[wavy=0.2](0.5,-4) -- (0.5,-2); \draw (0.5,-4.7) node{$z^{-1}$};
  \draw[greenline=0.5](2.5,0)  -- (4,0);  \draw (3.4,0.7) node{$z$};
  \draw[wavy=0.8](0.5,-2) -- (4,-2); \draw (3.5,-2.7) node{$z$};
  \draw[redline=0.8](6,0) -- (8,0);   \draw (8,0) node[right]{$q^{\mu/2} z $};
  \draw[blueline=0.8](6,-2) -- (8,-2);  \draw (8,-2) node[right]{$q^{-\mu/2} z $};
  \draw(5,-1) ellipse (1cm and 2cm);
  \draw(5.7,-1.9) node[left]{  \mbox{\footnotesize $\cO^{-1}$}};
    \end{scope}
\end{tikzpicture}
\end{center}

The identity in $\cA^{(2)}$ corresponding to this is 
\eq{
\label{keyrelation:left}
\tcK_\brho(q^{\mu/2}z,\wt\xi)_2 \wt{\cR}_{\vrho\brho}(z^2) \tcK_\vrho(q^{-\mu/2}z,\wt\xi)_1 \cO^{-1} =
\cO^{-1} \tcK_\phi(z,\wt\xi)_2 \wt{\cR}_{\ups\phi}(z^{2}) \tcK_\ups(z,\wt\xi)_1.
}

\begin{thrm} \label{thm:keyrelation:left}
Relation \eqref{keyrelation:left} is satisfied.
\end{thrm}

\begin{proof}
Given the definitions \eqref{K-tilde} and \eqref{RsitauRrhobrho:tilde}, this follows straightforwardly by inverting \eqref{keyrelation-init:right} and replacing $(z,\xi) \mapsto (qz,\txi^{-1})$.
\end{proof}

\section{Discussion} \label{sec:discussion}
The main result of this paper is Theorem \ref{thm:keyrelation:right} which can be viewed as a boundary analogue of Theorem \ref{thm:fund}. 
To establish this result, first we needed to show that all R and K-operators involved in equation \eqref{keyrelation-init:right} are well-defined actions of the universal elements $\cR$ and $\cK$ on the infinite-dimensional $U_q(\wh\mfb^+)$-modules involved. 
The key fact that allows for this is that $\cR$ and $\cK$ live in completions of $U_q(\wh\mfb^+)\ot U_q(\wh\mfb^-)$ and of $U_q(\wh\mfb^+)$, respectively. 
This is very familiar for $\cR$ but for $\cK$ relies on the recent works \cite{AV22a,AV22b}. 
Introducing the $U_q(\wh\mfb^+)$-intertwiner $\cO$ and the formula for $\Delta(\cK)$ given by \eqref{K:axiom2}, relation \eqref{thm:keyrelation:left} follows immediately from the intertwining property of $\cO$. 

The open Q-operator $\cQ(z)$ of \cite{VW20} is the trace of a product of R and K-operators over the $U_q(\wh\mfb^+)$-module $(\vrho_z,W)$ and there is a similar construction of an open Q-operator $\wb{\cQ}(z)$. 
In a future paper, the authors will present this construction and the use of Theorem \ref{thm:keyrelation:left} in deriving a boundary analogue of the factorization relation $\mc T_{\mu}(z) \: \propto \: \mc Q(zq^{-\mu/2}) \wb{\mc Q}(zq^{\mu/2})$.
They will also develop the analogous theory for different coideal subalgebras, in particular those for which non-diagonal solutions of the reflection equation are intertwiners.
There is a quite subtle rational degeneration of the construction in the present paper.
The first-named author will study this in a separate paper, giving an alternative approach to Q-operators for the open XXX spin chain, cf.~\cite{FS15}.

\appendix

\section{Deformed Pochhammer symbols and exponentials} \label{app:qexp}

This appendix is independent from the main text, but provides identities which are used there.
We review some basic theory of deformed Pochhammer symbols and exponentials (as formal power series) with a deformation parameter $p \in \C^\times$, which corresponds to $q^2$ in the main text. 

\subsection{Deformed Pochhammer symbols}

Let $x$ be a formal variable.
For $n \in \Z$, the (finite) deformed Pochhammer symbol $(x;p)_n \in \C[[x]]$ is defined by
\eq{
(x;p)_n := \begin{cases} \displaystyle \prod_{m=0}^{n-1} (1-x p^m) & \text{if } n \ge 0, \\ \displaystyle \prod_{m=n}^{-1} (1-x p^m)^{-1} & \text{if } n <0 \end{cases}
}
(the definition for $n<0$ is understood as a product of geometric series); since its constant coefficient is nonzero, it is invertible.
For all $p \in \C^\times$ and $n \in \Z_{\ge 0}$ we have the following basic identity in $\C[[x]]$, see \cite[(I.2), (I.3)]{GR90}:
\eq{ \label{qPochhammer:identity}
(x;p)_{-n} = (p^{-n} x;p)_n^{-1} = (x/p;p^{-1})_n^{-1} = (-x)^{-n} p^{n(n+1)/2} (p/x;p)_n^{-1}.
}
Assuming $|p|<1$, the infinite deformed Pochhammer symbol
\eq{ 
(x;p)_\infty := \prod_{m=0}^\infty (1-x p^m)
}
is an invertible formal power series with well-defined coefficients in $\C$.
The following identity holds in $\C[[x]]$, see \cite[(I.5)]{GR90}:
\eq{ \label{qPochhammer:finitetoinfinite}
(x;p)_n = \frac{(x; p)_\infty}{(p^n x ; p)_\infty}.
}

\subsection{Deformed exponentials}

From now on we assume that $p$ is not a root of unity. 
In particular, $(p;p)_k \ne 0$ for all $k \in \Z_{\ge 0}$.
The \emph{deformed exponential} is the invertible formal power series
\eq{ \label{qexp:def:app}
e_{p}(x) := {}_1\phi_0(0;-;p,x) = \sum_{k=0}^\infty \frac{x^k}{(p;p)_k}.
}
The ordinary exponential formal power series arises as the termwise limit $\lim_{p \to 1} e_p((1-p)x) = {\mathrm e}^x$.
This series satisfies the functional relation
\eq{ \label{qexp:funcrel}
e_{p}(p x) = (1-x) e_{p}(x),
}
see \cite[Sec. 1.3]{GR90}.
Since constants are the only formal power series which are invariant under $x \mapsto px$, an inspection of constant coefficients shows that \eqref{qexp:funcrel} implies
\eq{ \label{qexp:qpochhammer}
e_{p}(x) = \frac{1}{(x;p)_\infty} \qq \text{if } |p|<1.
} 
Similarly we consider the invertible formal power series
\eq{ \label{QEXP:def:app}
E_p(x) := {}_0\phi_0(-;-;p,-x) = \sum_{k=0}^\infty \frac{p^{k(k-1)/2} x^k}{(p;p)_k}.
}
Then $E_p(-x)^{-1}$ also satisfies \eqref{qexp:funcrel} and by comparing constant coefficients again we deduce $e_{p}(x) = E_p(-x)^{-1}$.
By evaluating \eqref{qPochhammer:identity} at $x=1$, we obtain $E_p(-x) = e_{p^{-1}}(p^{-1}x)$ and hence
\eq{ \label{qexp:inverseformula}
e_{p}(x) = e_{p^{-1}}(p^{-1}x)^{-1} \qu \in \C[[x]].
}

Deformed exponentials in $x$ and $y$ satisfy various useful identities in particular deformations of the commutative algebra $\C[[x,y]]$.
For instance, in any algebra generated by the symbols $x$ and $y$ such that $yx = \ga xy$ for $\ga \in \C$, the definition implies the following identity:
\eq{ \label{qexp:commrel}
    y e_p(x) = e_p(\ga x) y
}
which we will use repeatedly.
For a survey of product formulas analogous to $\exp(x)\exp(y)=\exp(x+y)$, see \cite{Ko97}.
We will need the following result.

\begin{lemma} \label{lem:qexp:product}
Let $x,y$ be elements of an algebra such that $y x = p x y$.
The following identities hold as formal power series in $x,y$:
\begin{align}
\label{qexp:product1} e_p(x) e_p(y) &= e_p(x+y), \\
\label{qexp:product2} e_p(y) e_p(x) &= e_p\big(x(1-y)\big) e_p(y) = e_p(x) e_p(-xy) e_p(y) = e_p(x) e_p\big((1-x)y\big).
\end{align}
\end{lemma}

\begin{proof}
\eqref{qexp:product1} is a direct consequence of the well-known q-binomial formula, see e.g.~\cite[Ex. 1.35]{GR90}.
For \eqref{qexp:product2} see \cite[Prop. 3.2]{Ko97}.
\end{proof}

\subsection{Deformed exponentials as linear maps} 

Let $V$ be a $\C$-linear space.
Call an operator $f$ on $V$ \emph{locally nilpotent} if for all $v \in V$ there exists $o(v) \in \Z_{\ge 0}$ such that $f^{o(v)}(v)=0$ (note that nilpotent operators are locally nilpotent and if $V$ is finite-dimensional the converse is true).
If $f$ is nilpotent, the deformed exponential $e_p(f)$ defines an invertible map on $V$. 
If additionally $y$ is an indeterminate then $e_p(yf)$ is a well-defined invertible element of $\End(V)[[y]]$.

In the case $V=W \ot W$ the following commutation relations for linear-operator valued formal series are satisfied, expressed in terms of the linear operators $a$, $\adag$, $\badag$, $f(D)$ ($f \in \cF$) on $W$ defined in Section \ref{sec:extqosc}.

\begin{lemma} \label{lem:qexp:rels1}
Let $y$ be a formal variable. 
In $\End(W \ot W)[[y]]$ the following identities hold:
\eq{ \label{qexp:commute}
\big[ e_p(y a_1 \ba^\dag_2), f(D_1+D_2) \big] = \big[ e_p(y a_1 \ba^\dag_2) , a_1  \big] = \big[ e_p(y a_1 \ba^\dag_2) , \ba^\dag_2 \big] = 0
}
for all $f \in \cF$ and 
\begin{align} 
\label{qexp:adag}
\big[ e_p(y a_1 \ba^\dag_2), a^\dag_1 \big] &= y p^{D_1} \ba^\dag_2 e_p(y a_1 \ba^\dag_2), \\
\label{qexp:Dba}
\big[ e_p(y a_1 \ba^\dag_2), p^{-D_1} a_2 \big] &= y e_p(y a_1 \ba^\dag_2) a_1 p^{-D_1}.
\end{align}
\end{lemma}

\begin{proof}
Note that \eqref{qexp:commute} follows directly from the definition of the deformed exponential.
A straightforward inductive argument using \eqref{A:basicrels2} yields
\begin{align}
[a^{k+1},a^\dag] &= (1-p^{k+1}) p^D a^k,
\label{A:rels1}  \\
[(\ba^{\dag})^{k+1},a]_{p^{k+1}} &= (1-p^{k+1}) (\ba^{\dag})^k, \label{barA:rels1}
\end{align}
for all $k \in \Z_{\ge 0}$, which imply \eqref{qexp:adag} and \eqref{qexp:Dba}, respectively.
\end{proof}

\section{Explicit expressions for R-operators} \label{app:R-operators}

In this appendix we derive explicit formulas for $\cR_{\vrho\brho}(z)$ and $\cR_{\ups\phi}(z)$, defined by \eqref{Rrhobrho:Rsitau:def} as images of the universal R-matrix $\cR$ fixing $w_0 \ot w_0$. 
We expect that these will be useful in further studies of Baxter's Q-operators for the open XXZ spin chain; for now they will allow us to give a proof of the boundary factorization identity which does not rely on the universal K-matrix formalism. 
First we note that, by the second part of Theorem \ref{thm:R(z):action}, $\cR_{\vrho\brho}(z)$ and $\cR_{\ups\phi}(z)$ lie in the centralizer
\eq{ \label{A2:0:def}
\cA^{(2)}_0 := \Big\{ X \in \cA^{(2)} \, \Big| \, \big[ X, q^{D_1+D_2} \big] = 0 \Big\}.
}
One straightforwardly verifies that $\cA^{(2)}_0$ is generated by elements of the form
\eq{ \label{series0}
\sum_{k \ge 0} (\ba^\dag_2)^k f_k(D_1,D_2) a_1^k, \qq \sum_{k \ge 0} (a^\dag_1)^k f_k(D_1,D_2) a_2^k, \qq \qq f_k \in \cF^{(2)}.
}
Hence, elements of $\cA^{(2)}_0$ in fact commute with all elements of the form $f(D_1+D_2)$ ($f \in \cF$).

\subsection{Explicit expression for $R_{\ups\phi}(z)$}

We first state and prove an explicit formula for $R_{\ups\phi}(z)$.
We keep using the shorthand notation $p=q^2$.

\begin{thrm}
For all $z \in \C$ we have 
\eq{ \label{Rsitau:formula}
\cR_{\ups\phi} (z) = e_p(z a^\dag_1 a_2) q^{(\mu-1)(D_2-D_1)-2 D_1 (D_2+1)}.
}
\end{thrm}

\begin{proof}
From Proposition \ref{prop:R(z):intw} we deduce that $\cR_{\ups\phi} (z)$ is a solution of the linear relation
\eq{
\label{X:equation:1} X (\ups_{z} \ot \phi^-)(\Del(u)) = (\ups_{z} \ot \phi^-)(\Del^{\rm op}(u)) X \qu \text{for all } u \in U_q(\wh\mfb^-).
}
First of all, note that the element in the right-hand side of \eqref{Rsitau:formula} satisfies \eqref{X:equation:1} with $u \in \{ k_0,k_1 \}$ and so it suffices to prove that the vector space
\eq{ 
\mc X = \Big\{ X \in \cA^{(2)}_0 \, \Big| \, X \text{ satisfies } \eqref{X:equation:1} \text{ for } u \in \{ f_0,f_1\} \Big\}
}
is spanned by 
$e_p(z^2 a^\dag_1 a_2) q^{(\mu-1)(D_2-D_1)-2 D_1 (D_2+1)}$.
Using the explicit formulas \eqref{Delta:def}, \eqref{hom:sigma} and \eqref{homs:minus}, we obtain that \eqref{X:equation:1} is equivalent to the system
\begin{align*}
X \Big( z^{-1} a_1 (q^{-\mu}-q^{\mu-2D_1} ) q^{-\mu-2D_2-1} + q^{-1} a_2 \Big) &= \Big( z^{-1} a_1 (q^{-\mu}-q^{\mu-2D_1} ) + q^{\mu-2(D_1+1)} a_2 \Big) X, \\
X a^\dag_1 q^{\mu+1+2D_2} &= a^\dag_1 X.
\end{align*}
Without loss of generality we may write $X = \wt X q^{(\mu-1)(D_2-D_1)-2 D_1 (D_2+1)}$ with $\wt X \in \cA^{(2)}_0$.
Hence \eqref{X:equation:1} is equivalent to
\eq{
z^{-1} [\wt X ,a_1 (1-p^{\mu - D_1} ) ] = p^{\mu -D_1- 1} a_2 \wt X - \wt X p^{D_1} a_2 , \qq \qq [\wt X, a^\dag_1] = 0.
}
It is straightforward to check that the centralizer in $\cA^{(2)}_0$ of $a^\dag_1$ is the subalgebra generated by elements of the form $\sum_{k \ge 0} (a^\dag_1)^k f_k(D_2) a_2^k$ with $f_k \in \cF$.
It follows that $\wt X$ is of this form.
Therefore \eqref{X:equation:1} is equivalent to the single equation
\[
\sum_{k \ge 0} \big[(a^\dag_1)^k,a_1 (1-p^{\mu-D_1} ) \big]  f_k(D_2) a_2^k = z \sum_{k \ge 0} (a^\dag_1)^k \big( p^{\mu-D_1-k-1} f_k(D_2+1) - p^{D_1} f_k(D_2) \big) a_2^{k+1}.
\]
The commutator vanishes if $k=0$ so in the left-hand side we replace $k$ by $k+1$.
For $k \ge 0$ we have
\begin{align*}
\big[(a^\dag)^{k+1},a (1-p^{\mu-D} ) \big] 
&= (a^\dag)^k (1-p^{k+1}) (p^{\mu - D - k - 1} - p^D).
\end{align*}
Hence \eqref{X:equation:1} is equivalent to 
the recurrence relation 
\[
(1-p^{k+1}) \big( p^{\mu-D_1-k-1} - p^{D_1} \big) f_{k+1}(D_2) = z \big( p^{\mu-D_1-k-1} f_k(D_2+1) - p^{D_1} f_k(D_2) \big).
\]
Viewing $\cF^{(2)}(D_1,D_2)$ as an $\cF(D_2)$-module, the elements $p^{\pm D_1}$ are linearly independent. 
Hence the above recurrence relation is equivalent to the system
\[
(1-p^{k+1}) f_{k+1}(D) = z f_k(D + 1), \qq f_k(D+1) = f_k(D).
\]
This is in turn equivalent to $f_k(D) \in (p;p)_k^{-1} z^k \C$ for $k \in \Z_{>0}$, as required.
\end{proof}

\subsection{The automorphism $\chi$ and the q-oscillator subalgebra $\wt{\cA}$}

To obtain an expression for $R_{\ups\phi}(z)$ in terms of deformed exponentials, it is very convenient to point out an additional automorphism $\chi$.
It cannot be defined on all of $\cA$ so we will consider a subalgebra $\wt{\cA}$.
First, consider the subalgebra $\wt{\cF}(D) \subset \cF(D)$ generated by
\[
p^{\pm D(D+1)/2}, \qq \ga^D, \qq (p \wt \ga;p)_D^{\pm 1}, \qq (p \ga z^2;p)_D, \qq (-\ga z^2)^{-D} (p \ga^{-1} z^{-2};p)_D^{-1}
\]
for all $\ga \in \C^\times$ and $\wt \ga \in \C^\times \backslash p^{\Z}$. 

For elements of $\wt{\mc F}(D)$, unlike general elements of $\mc F(D)$, the symbol $D$ can be formally evaluated at negative integers.
Accordingly, we define an involutive automorphism $\chi$ of $\wt{\cF}(D)$ accomplishing the formal replacement $D \mapsto -D-1$.
To be more precise, we set
\eq{ \label{chi:def:1}
\begin{aligned}
\multicolumn{2}{c}{$\chi\big( p^{\pm D(D+1)/2} \big) = p^{\pm D(D+1)/2}, \qq \qq \chi\big(\ga^D \big) = \ga^{-D-1}$, \qq \qq} \\
\chi\big( (p \wt \ga;p)_D^{\pm 1} \big) &= (1 - \wt \ga)^{\mp 1} p^{\pm D(D+1)/2} (-\wt \ga)^{\mp D} (p \wt \ga^{-1};p)_D^{\mp 1}, \\
\chi\big( (p \ga z^2;p)_D \big) &= (1 - \ga z^2)^{-1} p^{D(D+1)/2} (-\ga z^2)^{-D} (p \ga^{-1} z^{-2};p)_D^{-1}, \\
\chi\big( (-\ga z^2)^{-D} (p \ga^{-1} z^{-2};p)_D^{-1} \big) &= (1 - \ga z^2) p^{-D(D+1)/2} (p \ga z^2;p)_D.
\end{aligned}
}

We denote the subalgebra of $\End(W)$ generated by $a^\dag$, $a$ and $\wt{\mc F}(D)$ by $\wt{\cA}$.
It is straightforward to check that $\chi$ extends to a (non-involutive) algebra automorphism of $\wt{\cA}$ by means of the assignments
\eq{ \label{chi:def:2}
\chi(a) = \badag, \qq \chi(\adag) = a.
}

We can formulate a completion of the tensor product $\wt{\cA} \ot \wt{\cA}$ in a similar way as for $\cA \ot \cA$.
More precisely, we consider the subalgebra $\wt{\cF}^{(2)}$ of $\cF^{(2)}$ generated by the subsets $\wt{\cF}(D_1)$, $\wt{\cF}(D_2)$ and the special elements $p^{\pm D_1(D_2+1)}$.
The completed tensorial square of $\wt{\cA}$ is defined to be the subalgebra $\wt{\cA}^{(2)}$ of $\End(W \ot W)$ generated by the elements \eqref{series} with $g_{k,\ell} , \, h_{k,\ell} \in \wt{\cF}^{(2)}$.
Note that the boundary factorization identity \eqref{keyrelation-init:right} is an identity in the subalgebra $\wt{\cA}^{(2)} \subset \End(W \ot W)[[z]]$. 

The automorphism 
\eq{ \label{chi2:def}
\chi^{(2)} := \sigma \circ (\chi \ot \chi^{-1})
}
of $\wt{\cA} \ot \wt{\cA}$ naturally extends to an automorphism of $\wt{\cA}^{(2)}$, fixing $p^{\pm D_1(D_2+1)}$ and acting termwise on power series in locally nilpotent operators.

\begin{rmk}
The map $\chi$ can be seen as an infinite-dimensional version of conjugation by anti-diagonal matrices; certain $U_q(\wh\mfb^+)$-representations are naturally related this way.
For instance, for the 2-dimensional representation $\Pi$, note that $\Ad(J) \circ \Pi = \Pi \circ \Phi$ where $\Ad$ denotes `conjugation by' and $J = \big( \begin{smallmatrix} 0 & 1 \\ 1 & 0 \end{smallmatrix} \big)$.
In the same way, $\chi$ relates the prefundamental representations $\vrho$ and $\brho$ up to a twist by the diagram automorphism $\Phi$: $\chi \circ \vrho = \brho \circ \Phi$.
Hence, the condition \eqref{Phi:R} and the 1-dimensionality of the solution space of the relevant linear equation implies
$(\Ad(J) \ot \chi) ( \cL_\vrho(z) ) = \cL_\brho(z)$.
At the same time, a suitable scalar multiple of $\cR_{\Pi \, \Pi}(z)$, i.e. the R-matrix for the XXZ chain, is fixed by $\Ad(J \ot J)$ and we will see in Section \ref{app:Rrhobrho:explicit} that the same statement is true for $\cR_{\vrho\brho}(z)$ and $\chi^{(2)}$.

From \eqref{Uqk:def} it follows that $\Phi(U_q(\mfk)) = U_q(\mfk)|_{\xi \mapsto \xi^{-1}}$. 
Hence, the boundary counterparts of the above relations also involve inversion of the free parameter $\xi$:
\[
\Ad(J)\big( K_{\Pi}(z) \big)|_{\xi \mapsto \xi^{-1}} = -\xi \, K_{\Pi}(z), \qq \chi(\cK_\vrho(z))|_{\xi \mapsto \xi^{-1}}  = q^{-1} (z^{2} - \xi^{-1})^{-1} \cK_\brho(z).
\]
In fact, applying $\chi \ot \Ad(J)$ to the reflection equation \eqref{eq:RightBoundary:W} with $\pi = \vrho$ and inverting $\xi$ we see that \[
\cK_\vrho(z) \mapsto \chi(\cK_\vrho(z))|_{\xi \mapsto \xi^{-1}}
\]
defines a bijection: ${\sf RE}_\vrho \to {\sf RE}_\brho$ of the solution spaces defined in \eqref{eq:REsolspace}. \hfill \rmkend
\end{rmk}

We can use the map $\chi^{(2)}$ to generate further relations similar to those in Lemma \ref{lem:qexp:rels1}.

\begin{lemma} \label{lem:qexp:rels2}
Let $y$ be a formal parameter. 
In $\End(W \ot W)[[y]]$ the following identities hold:
\begin{align} 
\label{qexp:badag} [\ba^\dag_2,e_p(y a^\dag_1 a_2)] &= y e_p(y a^\dag_1 a_2) a^\dag_1 p^{-D_2-1}, \\
\label{qexp:Dadagba}
[\ba^\dag_1 a_2,e_p(y a_1 \ba^\dag_2) ] &= y \big( e_p(y a_1 \ba^\dag_2)  p^{-D_1-1} - p^{-D_2-1} e_p(y a_1 \ba^\dag_2) \big). 
\end{align}
\end{lemma}

\begin{proof}
In this proof we view the algebra $\cA$ as a subalgebra of $\End(W)[[y]]$, and similarly for $\cA^{(2)}$.
To prove \eqref{qexp:badag}, first we apply $\chi^{(2)}$ to \eqref{qexp:adag}, obtaining
\eq{
\label{qexp:ba} [e_p(y a_1 \ba^\dag_2 ),a_2]  = y a_1 p^{-D_2-1} e_p(y a_1 \ba^\dag_2).
}
Now consider the unique involutive algebra anti-automorphism $\eta: \cA \to \cA$ which exchanges $a$ and $a^\dag$ and fixes $f(D)$ for all $f \in \cF$ and the unique involutive algebra anti-automorphism $\wb\eta: \cA \to \cA$ which exchanges $a$ and $\ba^\dag$ and fixes $f(D)$ for all $f \in \cF$.
Then $\eta^{(2)} := \eta \ot \wb\eta$ is an algebra antiautomorphism of $\cA \ot \cA$.
It extends in a natural way to an algebra antiautomorphism of $\cA^{(2)}$.
By applying $\eta^{(2)}$ to \eqref{qexp:ba} we obtain \eqref{qexp:badag}.

Finally, to prove \eqref{qexp:Dadagba}, upon right-multiplying \eqref{qexp:Dba} by $p^{D_1+D_2+1}$ we obtain 
\eq{ 
\label{qexp:babD} [e_p(y a_1 \ba^\dag_2), a_1 p^{D_2}] = y e_p(y a_1 \ba^\dag_2) a_1 p^{D_2}.
}
From \eqref{qexp:adag} and \eqref{qexp:babD} it follows that
\eq{
\label{qexp:adagbabD}
\begin{aligned}
[e_p(y a_1 \ba^\dag_2), a^\dag_1 a_2 p^{D_2}] &= y \Big( \ba^\dag_2 p^{D_1} e_p(y a_1 \ba^\dag_2) a_2 + a^\dag_1  e_p(y a_1 \ba^\dag_2) a_1 \Big) p^{D_2} \\
&= y \Big( p^{D_1} e_p(y a_1 \ba^\dag_2) \big(p^{D_2}-1\big) + \big(1-p^{D_1}\big) e_p(y a_1 \ba^\dag_2) p^{D_2} \Big) \\
&= y \big( e_p(y a_1 \ba^\dag_2) p^{D_2} - p^{D_1} e_p(y a_1 \ba^\dag_2) \big).
\end{aligned}
}
Now \eqref{qexp:Dadagba} follows as the $\chi^{(2)}$-image of \eqref{qexp:adagbabD}.
\end{proof}

\subsection{Explicit expression for $\cR_{\vrho\brho}(z)$} \label{app:Rrhobrho:explicit}

To aid the computation of $\cR_{\vrho\brho}(z)$, consider the subalgebra $\wt{\cA}^{(2)}_0 = \wt{\cA}^{(2)} \cap \cA^{(2)}_0$, which is preserved by $\chi^{(2)}$.

\begin{lemma} \label{prop:R:chi2}
$\cR_{\vrho\brho}(z)$ is a $\wt{\cA}^{(2)}_0$-valued formal power series whose coefficients are fixed by $\chi^{(2)}$.
\end{lemma}

\begin{proof}
It is clear from \eqref{homs:plus} and \eqref{homs:minus} that $\vrho \ot \brho^{\, -}$ takes values in $\wt{\cA} \ot \wt{\cA} \subset \wt{\cA}^{(2)}$.
Now recall \eqref{R:factorization} and note that the factor $\ka$ acts as $p^{D_1(D_2+1)}$.
Furthermore, noting the form of $(\Sigma_z \ot \id)(\Theta)$ given by \eqref{Theta:keyobservation} with the components $\Theta_\la$ lying in $U_q(\wh\mfn^+)_\la \ot U_q(\wh\mfn^+)_{-\la}$ ($\la \in \wh Q^+$), we obtain that the action of $\mc{R}(z)$ on $(\vrho \ot \brho^{\, -},W \ot W)$ is by an element of $\wt{\cA}^{(2)}_0$.
For the second part, note that
\[
\chi^{(2)} \circ (\vrho \ot \brho^{\, -}) = (\chi^{-1} \ot \chi) \circ (\brho^{\, -} \ot \vrho) \circ \sigma = (\vrho \ot \brho^{\, -}) \circ (\om \ot \om) \circ \sigma.
\]
Applying this to $\cR(z)$, making use of \eqref{R(z)uni:def}, \eqref{omega:Sigma} and \eqref{omega:R}, we obtain $\chi^{(2)}(\cR_{\vrho\brho}(z)) = \cR_{\vrho\brho}(z)$.
\end{proof}

In the derivation of the formula for $\cR_{\vrho\brho}(z)$, we rely on the following result.

\begin{lemma} \label{lem:centralizer}
The centralizer of the subset $\{ \adag_1, \badag_2 \}$ in $\cA^{(2)}$ is equal to $\C[[z]]$.
\end{lemma}

\begin{proof}
This centralizer is the intersection of the centralizer of $\adag_1$ and the centralizer of $\badag_2$, which are easily found to be equal to
\[
\bigg\{ \sum_{k,\ell \ge 0} (a^\dag_1)^k f_{k,\ell}(D_2) a_2^\ell \, \bigg| \, f_{k,\ell} \in \cF \bigg\}, \qq
\bigg\{ \sum_{k,\ell \ge 0} (\ba^\dag_2)^k g_{k,\ell}(D_1) a_1^\ell \, \bigg| \, g_{k,\ell} \in \cF \bigg\},
\]
respectively.
Clearly their intersection is trivial.
\end{proof}

Now we are ready to state and prove a formula for $\cR_{\vrho\brho}(z)$ in terms of deformed exponentials.

\begin{thrm} \label{thm:Rrhobrho:formula}
For all $z$ we have 
\eq{ \label{Rrhobrho:formula}
\cR_{\vrho\brho}(z) = e_{q^2}(q^3 z a_1 \ba^\dag_2) e_{q^2}(q^{-1} z a^\dag_1 a_2) q^{-2D_1(D_2+1)}.
}
\end{thrm}

\begin{proof}
Clearly, $w_0 \ot w_0$ is fixed by the expression on the right-hand side of \eqref{Rrhobrho:formula}. 
In the following we initially work over the ring $\C[[z,z_2]]$ for some new indeterminate $z_2$ and write $z_1 = z z_2$.
By applying $\vrho_{z_1} \ot \Pi_1 \ot \brho^{\, -}_{z_2}$ to \eqref{R:YBE} and left and right-multiplying by $\cL^-_{\brho,23}(z_2^{-1})^{-1}$ we obtain
\eq{ \label{Rrhobrho:equation}
\cR_{\vrho\brho}(z)_{12}  \cL_\vrho(z_1)_{13} \cL^-_{\brho}(z_2^{-1})^{-1}_{32} = \cL^-_{\brho}(z_2^{-1})^{-1}_{32}  \cL_\vrho(z_1)_{13} \cR_{\vrho\brho}(z)_{12}
}
an equation in $(\wt{\cA}^{(2)} \ot \End({\C^2}))[[z_2]]$.
By a direct computation we obtain
\eq{
\cL^-_\brho(z_2^{-1})^{-1} = \frac{1}{z_2^2-1}
\begin{pmatrix} q^{-D-1} z_2^2 & \ba^\dag q^{-D-1} z_2 \\ a q^{D-1} z_2 & q^{D+1} z_2^2 - q^{-D-1} \end{pmatrix} \in \End({\C^2}) \ot \wt{\cA}.
}
Now we consider the equation
\eq{ \label{X:equation:2}
(z_2^2-1) X_{12}  \cL_\vrho(z_1)_{13} \cL^-_{\brho}(z_2^{-1})^{-1}_{32} = 
(z_2^2-1) \cL^-_{\brho}(z_2^{-1})^{-1}_{32}  \cL_\vrho(z_1)_{13} X_{12} 
}
in $(\wt{\cA}^{(2)} \ot \End({\C^2}))[[z_2]]$, for some $X \in \wt{\cA}^{(2)}_0$ such that $\chi^{(2)}(X)=X$.
It suffices to prove that
\eq{ 
\mc X = \Big\{ X \in \wt{\cA}^{(2)}_0 \, \Big| \, X \text{ satisfies } \eqref{X:equation:2} \text{ and is fixed by } \chi^{(2)} \Big\},
}
which by Lemma \ref{prop:R:chi2} contains $(\vrho_z \ot \brho^{\, -})(\cR)$, is spanned by the element given in the right-hand side of \eqref{Rrhobrho:formula}.

By considering explicit expressions for $(z_2^2-1)  \cL_\vrho(z_1)_{13} \cL^-_{\brho}(z_2^{-1})_{32}^{-1}$ and 
$(z_2^2-1) \cL^-_{\brho}(z_2^{-1})_{32}^{-1}  \cL_\vrho(z_1)_{13}$, we obtain that \eqref{X:equation:2} amounts to the system
\begin{gather*}
X \big( q^{D_1 - D_2 -1} - a^\dag_1 a_2 q^{-D_1 +D_2 -2} z \big) = \big( q^{D_1 - D_2 -1} - a_1 \ba^\dag_2 q^{D_1 - D_2} z \big) X, \\
X \Big( \big( \ba^\dag_2 q^{D_1 - D_2 -1} + a^\dag_1 q^{-D_1 - D_2 -2} z \big) - a^\dag_1 q^{-D_1 + D_2} z z_2^2 \Big) =\\
\qq \qq = \Big( \ba^\dag_2 q^{-D_1-D_2-1} - \big( a^\dag_1 q^{-D_1-D_2-2} + \ba^\dag_2 q^{D_1-D_2+1} z \big) z z_2^2 \Big)  X , \\
X \Big( a_2 q^{-D_1+D_2-1} - \big( a_1 q^{D_1-D_2} + a_2 q^{D_1+D_2+1} z \big) z z_2^2 \Big) = \\ 
\qq \qq = \Big( \big( a_2 q^{D_1+D_2-1} + a_1 q^{D_1-D_2} z \big)  - a_1 q^{D_1+D_2+2} z z_2^2 \Big) X, \\
X \Big( q^{-D_1+D_2+1} + q^{D_1-D_2+1} z^2 - a_1 \ba^\dag_2 q^{D_1-D_2} z \Big) = \Big( q^{-D_1+D_2+1} + q^{D_1-D_2+1} z^2  - a^\dag_1 a_2 q^{-D_1+D_2-2} z \Big) X
\end{gather*}
for $X \in \wt{\cA}_0^{(2)}$ fixed by $\chi^{(2)}$.
Since $\C[[z,z_2]] \cong (\C[[z]])[[z_2]]$, considering expansion coefficients with respect to $z_2$, we can use $[X,q^{D_1+D_2}]=0$ to deduce that the above system is equivalent to
\begin{gather}
\label{X:a}
\begin{aligned}
X a_2 q^{-2D_1} &= \big( a_2 + a_1 q^{-2D_2+1} z \big) X, \qq & 
a_1 X &= X  \big( a_1 q^{-2(D_2+1)} + q^{-1} a_2 z \big), \\
X a^\dag_1 q^{2(D_2+1)} &= \big( a^\dag_1 + \ba^\dag_2 q^{2D_1+3} z \big) X, &
\ba^\dag_2 X &= X \big( \ba^\dag_2 q^{2D_1} + q^{-1} a^\dag_1 z \big) , 
\end{aligned} \\
\begin{aligned}
\label{X:D}
\big[ X , q^{2D_1} \big] &= \big( X a^\dag_1 a_2 q^{2D_2-1} - a_1 \ba^\dag_2 q^{2D_1+1} X \big) z, \\
\big[ X, q^{2D_2} + q^{2D_1} z^2 \big] &= \big( X a_1 \ba^\dag_2 q^{2D_1-1} - a^\dag_1 a_2 q^{2D_2-3} X \big) z.
\end{aligned}
\end{gather}
Note that $q^{-2D_1(D_2+1)} \in \wt{\cA}^{(2)}_0$ is fixed by $\chi^{(2)}$.
Hence without loss of generality we may write 
\eq{
X = \wt X q^{-2D_1(D_2+1)},
}
for some $\wt X \in \wt{\cA}^{(2)}_0$ fixed by $\chi^{(2)}$.
The system \eqrefs{X:a}{X:D} is equivalent to
\begin{align}
\label{tildeX:a} [\wt X, a_2] &= q^{-2D_2+1} a_1 \wt X z,
& [a_1, \wt X]  &= \wt X q^{2D_1-1} a_2 z, \\
\label{tildeX:adag} [\wt X, a^\dag_1] &= \ba^\dag_2 q^{2D_1+3} \wt X z, 
& [\ba^\dag_2,\wt X] &= \wt X a^\dag_1 q^{-2D_2-3} z, \\
\label{tildeX:D} \big[ \wt X, q^{2D_1} \big] &= \big( \wt X a^\dag_1 a_2 q^{2D_1-1} - a_1 \ba^\dag_2 q^{2D_1+1} \wt X \big) z, \hspace{-50mm} \\
\label{tildeX:DbD} \big[ \wt X, q^{2D_2} + q^{2D_1} z^2 \big] &= \big( \wt X a_1 \ba^\dag_2 q^{2D_2+3}  - a^\dag_1 a_2 q^{2D_2-3} \wt X \big) z. \hspace{-50mm}
\end{align}
Since $\chi^{(2)}$ fixes $\wt X$, the equations in \eqref{tildeX:a} and the equations in \eqref{tildeX:adag} are pairwise equivalent.
At the same time, the system \eqrefs{tildeX:a}{tildeX:adag} implies \eqref{tildeX:D} and \eqref{tildeX:DbD}.
To show this, since $[\wt X,q^{2D_1}] = [a^\dag_1 a_1,\wt X]$ from \eqrefs{tildeX:a}{tildeX:adag} we obtain
\begin{align*}
& [\wt X,q^{2D_1}]  + a_1 \ba^\dag_2 q^{2D_1+1} \wt X z -  \wt X a^\dag_1 a_2 q^{2D_1-1} z = \\
&\qq = a_1 \ba^\dag_2 q^{2D_1+1} \wt X z - [\wt X,a^\dag_2 ] a_1 + a^\dag_1 [a_1,\wt X] -  \wt X a^\dag_1 a_2 q^{2D_1-1} z  \\
&\qq = \big( \ba^\dag_2 q^{2D_1+3} [a_1, \wt X] - [\wt X,a^\dag_1] a_2 q^{2D_1-1} \big) z,
\end{align*}
which vanishes, thereby recovering \eqref{tildeX:D}.
Applying $\chi^{(2)}$ to \eqref{tildeX:D} we obtain $[\wt X,q^{-2D_2}] = \big( \wt X a^\dag_1 a_2 q^{-2D_2-1} - a_1 \ba^\dag_2 q^{-2D_2+1} \wt X \big) z$.
Left-and-right multiplying this by $q^{2D_2}$ 
and using \eqrefs{tildeX:a}{tildeX:adag} to rewrite the result we obtain
\eq{ \label{tildeX:bD}
[\wt X,q^{2D_2}] = \big( \ba^\dag_2 \wt X a_1 q^{2D_2+3} - q^{2D_2-1} a^\dag_1 \wt X a_2 \big) z.
}
Finally, using \eqref{tildeX:bD} and again \eqrefs{tildeX:a}{tildeX:adag}, we derive that
\begin{align*}
& [\wt X,q^{2D_2} + q^{2D_1} z^2] - \wt X a_1 \ba^\dag_2 q^{2D_2+3} z + a^\dag_1 a_2 q^{2D_2-3} \wt X z = \\
&\qq = \ba^\dag_2 \wt X a_1 q^{2D_2+3} z - q^{2D_2-1} a^\dag_1 \wt X a_2 z + [\wt X,q^{2D_1}] z^2 + \\
& \qq \qq - (\ba^\dag_2 \wt X - \wt X a^\dag_1 q^{-2D_2-3} z) a_1 q^{2D_2+3} z + a^\dag_1  q^{2D_2-1} (\wt X a_2 - a_1 q^{-2D_2+1} \wt X z) z \\
& \qq = \big( \wt X a^\dag_1 a_1 + a^\dag_1 a_1 \wt X + [\wt X, 1 - a^\dag_1 a_1] \big) z^2
\end{align*}
which vanishes, thereby proving \eqref{tildeX:DbD} as well.

We have obtained that the system \eqrefs{tildeX:a}{tildeX:DbD} is equivalent to the system \eqref{tildeX:adag}. 
Writing $p=q^2$, without loss of generality we set
\[
\wt X = Y e_p(q^3 z a_1 \ba^\dag_2) e_p(q^{-1} z a^\dag_1 a_2)
\]
for some $Y \in \wt{\cA}^{(2)}_0$ fixed by $\chi^{(2)}$, noting that $e_p(q^3 z a_1 \ba^\dag_2)$ and $e_p(q^{-1} z a^\dag_1 a_2)$ lie in $\wt{\cA}^{(2)}_0$ and are fixed by $\chi^{(2)}$. 
The theorem now follows from the following claim. \\

\emph{Claim:} \eqref{tildeX:adag} is satisfied if and only if $Y \in \C[[z]]$.\\

In the special case $Y=1$, \eqref{tildeX:adag} is indeed satisfied:
\begin{align*} 
[\wt X, a^\dag_1] - \ba^\dag_2 q^{2D_1+3} z \wt X
&= \Big( [ e_p(q^3 z a_1 \ba^\dag_2), a^\dag_1] -  \ba^\dag_2 q^{2D_1+3} z e_p(q^3 z a_1 \ba^\dag_2) \Big) e_p(q^{-1} z a^\dag_1 a_2), \\
[\ba^\dag_2, \wt X] - \wt X a^\dag_1 q^{-2D_2-3} z 
&= e_p(q^3 z a_1 \ba^\dag_2) \Big( [\ba^\dag_2, e_p(q^{-1} z a^\dag_1 a_2)] - e_p(q^{-1} z a^\dag_1 a_2) a^\dag_1 q^{-2D_2-3} z \Big),
\end{align*}
with the expressions in parentheses vanishing by virtue of \eqref{qexp:adag} and \eqref{qexp:badag} (with $y=q^{-1}z$).

For general $Y$ we therefore have
\begin{align*}
[\wt X, a^\dag_1] - \ba^\dag_2 q^{2D_1+3} z \wt X &= [Y,a^\dag_1] e_p(q^3 z a_1 \ba^\dag_2) e_p(q^{-1} z a^\dag_1 a_2), \\
[\ba^\dag_2, \wt X] - \wt X a^\dag_1 q^{-2D_2-3} z &= [\ba^\dag_2, Y] e_p(q^3 z a_1 \ba^\dag_2) e_p(q^{-1} z a^\dag_1 a_2).
\end{align*}
Both right-hand sides vanish, i.e. \eqref{tildeX:adag} is indeed satisfied, if and only if $Y$ lies in the centralizer in $\wt{\cA}^{(2)}$ of $\{ a^\dag_1, \ba^\dag_2 \}$, which is trivial by Lemma \ref{lem:centralizer}. This proves the claim.
\end{proof}

\section{An alternative proof of the main theorem} \label{app:altproof}

In this part of the appendix we give a proof of the boundary factorization identity \eqref{keyrelation-init:right} independent of the universal K-matrix formalism, instead relying on the explict expressions obtained in Appendix \ref{app:R-operators}.
Before we state and prove a key lemma, note that expressions of the form $e_p(\ga^D y)$ where $\ga \in \C^\times$ and $y$ is an indeterminate give rise to well-defined $\End(W)$-valued formal power series, sending $w_j$ to $e_p(\ga^j y) w_j$.

\begin{lemma} \label{lem:qexp:auxeqns}
Let $y$ be a formal parameter and let $p$ be a nonzero complex number, not a root of unity.
In $\End(W \ot W)[[y]]$ we have the identities
\begin{align} \label{qexp:auxeqn1}
e_p(p a_1 \ba^\dag_2) (y;p)_{D_1} &= (y;p)_{D_1} e_p(-a_1 \ba^\dag_2 p^{D_1} y) e_p(p a_1 \ba^\dag_2) \\
\label{qexp:auxeqn2}
e_p(p a_1 \ba^\dag_2) (p^{1-D_1} y;p)_{D_1}^{-1} e_p(p y \ba^\dag_1 a_2) &= e_p(p y \ba^\dag_1 a_2) (p^{1-D_2} y;p)_{D_2}^{-1} e_p(p a_1 \ba^\dag_2).
\end{align}
\end{lemma}

\begin{proof}
Note that 
\[
W \ot W = \bigoplus_{m \in \Z_{\ge 0}} (W \ot W)_m, \qq (W \ot W)_m := \bigoplus_{j,k \ge 0 \atop j+k=m} \C w_j \ot w_k.
\]
Because each factor in \eqrefs{qexp:auxeqn1}{qexp:auxeqn2} preserves each finite-dimensional subspace $(W \ot W)_m$, it suffices to prove the restrictions of \eqrefs{qexp:auxeqn1}{qexp:auxeqn2} to $(W \ot W)_m$, where $m \in \Z_{\ge 0}$ is fixed but arbitrary.
Note that on $(W \ot W)_m$ the operators appearing as arguments of the deformed exponentials are nilpotent. 
Therefore the operators on the left- and right-hand side of the restricted equations depend rationally on $p$ and hence it suffices to prove them with $p$ restricted to an open subset of $\C$.

We will prove the restriction of \eqref{qexp:auxeqn1} to $(W \ot W)_m$ for all $p \in \C$ such that $|p|<1$. 
Combining \eqref{qPochhammer:finitetoinfinite} and \eqref{qexp:qpochhammer} we obtain $(y;p)_D = \frac{e_p(p^Dy)}{e_p(y)}$; as a consequence, \eqref{qexp:auxeqn1} is equivalent to
\eq{
e_p(p a_1 \ba^\dag_2) e_p(p^{D_1} y) = e_p(p^{D_1} y) e_p(-a_1 \ba^\dag_2 p^{D_1} y) e_p(p a_1 \ba^\dag_2).
\label{qexp:auxeqn1:2}
}
But this equation follows directly from \eqref{qexp:product2} and the observation $(a_1 \ba^\dag_2)(p^{D_1}y) = p(p^{D_1}y)(a_1 \ba^\dag_2)$.

On the other hand\footnote{
We will need \eqref{qexp:auxeqn2} with $|p|<1$, but we are not aware of a direct proof of this.
}, we will prove the restricted version of \eqref{qexp:auxeqn2} for all $p \in \C^\times$ such that $|p|>1$.
In this case, for all $j \in \Z_{\ge 0}$ we have
\[
(p^{1-j}y;p)_j^{-1} = (y;p^{-1})_j^{-1} = \frac{(p^{-j}y;p^{-1})_\infty}{(y;p^{-1})_\infty} \in \C[[y]].
\]
From \eqref{qexp:qpochhammer} and \eqref{qexp:inverseformula} we deduce the identity 
\[ 
(p^{-j}y;p^{-1})_\infty = e_{p^{-1}}(p^{-j} y)^{-1} = e_p(p^{1-j} y) \in \C[[y]].
\]
Hence $(p^{1-D}y;p)_D^{-1} = (y;p^{-1})_\infty^{-1} e_p(p^{1-D}y)$ in $\End(W)[[y]]$ and \eqref{qexp:auxeqn2} is equivalent to
\eq{ \label{qexp:auxeqn2:2}
e_p(p a_1 \ba^\dag_2) e_p(p^{1-D_1}y) e_p(p y \ba^\dag_1 a_2) = e_p(p y \ba^\dag_1 a_2) e_p(p^{1-D_2}y) e_p(p a_1 \ba^\dag_2).
}

To prove \eqref{qexp:auxeqn2:2}, note that \eqref{qexp:Dadagba} can be evaluated at $y=p$, and the result can be rewritten as
\[
e_p(p a_1 \ba^\dag_2) \big( p^{-D_1} + \ba_1^\dag a_2 \big) = \big( p^{-D_2} + \ba_1^\dag a_2  \big) e_p(p a_1 \ba^\dag_2).
\]
By iteration we obtain 
\eq{ \label{qexp:auxeqn2:3}
e_p(p a_1 \ba^\dag_2) e_p\big( p^{1-D_1}y + p y \ba^\dag_1 a_2 \big) = e_p\big( p^{1-D_2} y + p y \ba^\dag_1 a_2 \big) e_p(p a_1 \ba^\dag_2).
}
Note that $(\ba^\dag_1 a_2) p^{1-D_1} = p \, p^{1-D_1} (\ba^\dag_1 a_2)$ and $p^{1-D_2} (\ba^\dag_1 a_2) = p \, (\ba^\dag_1 a_2) p^{1-D_2}$. 
Applying \eqref{qexp:product1}, we obtain \eqref{qexp:auxeqn2:2}, as required.
\end{proof}

\begin{proof}[Alternative proof of Theorem \ref{thm:keyrelation:right}.]
By virtue of \eqref{Oplus:def}, the desired identity, viz.
\eq{ \label{app:keyrelation:right}
\cK_\ups(z)_1 \cR_{\ups\phi}(z^2) \cK_\phi(z)_2 \,\cO = \cO \cK_\vrho(q^{-\mu/2}z)_1 \cR_{\vrho\brho}(z^2) \cK_\brho(q^{\mu/2}z)_2
}
for arbitrary $\mu \in \C$ and $q, \xi \in \C^\times$ such that $q$ is not a root of unity, is equivalent to
\eq{ \label{keyrelation:right:0}
\begin{aligned}
& e_{p}(p a_1 \ba^\dag_2) \cK_\ups(z)_1 \cR_{\ups\phi}(z^2) \cK_\phi(z)_2 e_{p}(p a_1 \ba^\dag_2)^{-1} = \\
&\qq = q^{\mu (D_1-D_2)/2} \cK_\vrho(q^{-\mu/2}z)_1 \cR_{\vrho\brho}(z^2) \cK_\brho(q^{\mu/2}z)_2 q^{\mu (D_2-D_1)/2} 
\end{aligned}
}
where $p=q^2$. 
The strategy of the proof is as follows.
We move various simple factors in $\cF^{(2)}(D_1,D_2)$ to the right in both sides of \eqref{keyrelation:right:0}, thus bringing them to a similar form.
Then more advanced product formulas involving q-exponentials and finite q-Pochhammer symbols yield the desired equality.

More precisely, we set $\ga = p q^{-\mu} \xi^{-1} \in \C^\times$ and from \eqref{qPochhammer:identity} deduce 
\[
(\ga z^{-2};p)_j^{-1} = p^{j(1-j)/2} (-\ga^{-1} z^2)^j (p^{1-j}\ga^{-1} z^2;p)_j^{-1}
\]
for all $j \in \Z_{\ge 0}$.
Using the identities $q^{-D^2} a^\dag = -q \ba^\dag q^{-D^2}$ and $q^{-D^2} a = a q^{2D-1} q^{-D^2}$, we obtain, for the left-hand side of \eqref{keyrelation:right:0},
\eq{
\begin{aligned}
& e_{p}(p a_1 \ba^\dag_2) \cK_\ups(z)_1 \cR_{\ups\phi}(z^2) \cK_\phi(z)_2 e_{p}(p a_1 \ba^\dag_2)^{-1} = \\
& \; = e_{p}(p a_1 \ba^\dag_2) \big( -\ga^{-1} q^{1-D_1} \big)^{D_1} (\ga z^2;p)_{D_1} (p^{1-D_1}\ga^{-1} z^2;p)_{D_1}^{-1} e_{p}( z^{2} a^\dag_1 a_2) \cdot \\
& \; \qq \cdot q^{(2\mu -1)D_1-2D_2 - 2D_1D_2 - D_2^2} (-\xi)^{D_2} e_{p}(p a_1 \ba^\dag_2)^{-1} \\
& \; = e_{p}(p a_1 \ba^\dag_2) (\ga z^2;p)_{D_1}  (p^{1-D_1}\ga^{-1} z^2;p)_{D_1}^{-1}  e_{p}(p \ga^{-1} z^2 \ba^\dag_1 a_2) (-q^{-D_1-D_2-2} \xi)^{D_1+D_2} e_{p}(p a_1 \ba^\dag_2)^{-1} \hspace{-5pt} \\
& \; = e_{p}(p a_1 \ba^\dag_2) (\ga z^2;p)_{D_1}  (p^{1-D_1} \ga^{-1} z^2;p)_{D_1}^{-1}  e_{p}(p \ga^{-1} z^2 \ba^\dag_1 a_2) e_{p}(p a_1 \ba^\dag_2)^{-1} (-q^{-D_1-D_2-2} \xi)^{D_1+D_2}. \hspace{-5pt} 
\end{aligned}
}
Similarly, for the right-hand side of \eqref{keyrelation:right:0} we obtain
\eq{ 
\begin{aligned}
& q^{\mu(D_1-D_2)/2} \cK_\vrho(q^{-\mu/2}z)_1 \cR_{\vrho\brho}(z^2) \cK_\brho(q^{\mu/2}z)_2 q^{\mu(D_2-D_1)/2} = \\
& \; = (\ga z^2;p)_{D_1} q^{\mu(D_1-D_2)/2 - D_1^2} (-\xi)^{D_1} e_{p}(q^3 z^2 a_1 \ba^\dag_2) e_{p}(q^{-1} z^2 a_1^\dag a_2) \cdot \\
& \; \qq \cdot (p^{1-D_2} \ga^{-1} z^2;p)_{D_2}^{-1} q^{\mu(D_2-D_1)/2 -2(D_1+D_2)-2D_1D_2-D_2^2} (-\xi)^{D_2} = \\
& \; = (\ga z^2;p)_{D_1} e_{p}(-a_1 \ba^\dag_2 q^{2D_1} \ga z^2) e_{p}(p \ga^{-1} z^2 \ba^\dag_1 a_2) (p^{1-D_2} \ga^{-1} z^2;p)_{D_2}^{-1} (-q^{-D_1-D_2-2} \xi)^{D_1+D_2}.
\end{aligned}
}
Therefore \eqref{keyrelation:right:0} is equivalent to
\eq{ \label{keyrelation:right:1}
\begin{aligned}
& e_{p}(p a_1 \ba^\dag_2) (\ga z^2;p)_{D_1}  (p^{1-D_1} \ga^{-1} z^2;p)_{D_1}^{-1}  e_{p}(p \ga^{-1} z^2 \ba^\dag_1 a_2) e_{p}(p a_1 \ba^\dag_2)^{-1} = \\
& \qq \qq = (\ga z^2;p)_{D_1} e_{p}(-a_1 \ba^\dag_2 p^{D_1} \ga z^2) e_{p}(p \ga^{-1} z^2 \ba^\dag_1 a_2) (p^{1-D_2} \ga^{-1} z^2;p)_{D_2}^{-1}.
\end{aligned}
}
Applying \eqref{qexp:auxeqn1} with $y=\ga z^2$ and \eqref{qexp:auxeqn2} with $y = \ga^{-1} z^2$, we deduce \eqref{keyrelation:right:1}, as required.
\end{proof}

\end{document}